\pdfoutput=1
\documentclass[12pt,reqno]{article}
\usepackage[markup=underlined]{changes}
\RequirePackage{etex}

\usepackage{mathpazo}

\usepackage{etex}
\usepackage{graphicx}
\usepackage{setspace}
\usepackage{amsmath}
\usepackage{amsfonts}
\usepackage{titlesec}
\usepackage{pictex}
\usepackage{comment}
\usepackage{amsthm}
\usepackage{graphics,epsfig,verbatim,bm,latexsym,url,amsbsy}
\usepackage{rotating}
\usepackage[authoryear,round]{natbib}
\usepackage{float}
\usepackage[position=bottom]{subfig}
\usepackage{mathrsfs}
\usepackage{multirow}
\usepackage{array}
\usepackage{bigints}
\usepackage{bbm}
\usepackage{geometry}
\usepackage{bbm}
\usepackage[ruled,vlined]{algorithm2e}
\usepackage{soul}
\usepackage{mathtools}
\usepackage{amssymb}
\usepackage{enumitem}
\usepackage{booktabs}
\usepackage{tikz}
\usetikzlibrary{arrows.meta,decorations.pathreplacing}

\DeclarePairedDelimiterX{\inp}[2]{\langle}{\rangle}{#1, #2}

\usepackage{hyperref}
\hypersetup{
	colorlinks=true,
	linkcolor=red!60!black,
	citecolor=blue!60!black,
	filecolor=magenta,      
	urlcolor=blue!60!black,
}
\usepackage[nameinlink,noabbrev,sort,capitalise]{cleveref}
\crefname{onlineappendix}{Online Appendix}{Online Appendices}

\newcommand{\de}{\mathop{}\!\mathrm{d}}
\newcommand{\e}{\varepsilon}
\newcommand{\Ex}{\mathbb{E}}

\newcommand{\R}{\mathbb{R}}

\renewcommand{\Pr}{\mathbb{P}}

\newcommand{\supp}{\mathrm{supp}}

\newcommand{\SPE}{\ensuremath{\mathsf{SPE}}}
\newcommand{\PBE}{\ensuremath{\mathsf{PBE}}}

\crefname{thm}{Theorem}{Theorems}

\theoremstyle{definition} \newtheorem{example}{Example}
\theoremstyle{definition} 
\theoremstyle{definition} 

\theoremstyle{definition} 
\theoremstyle{definition} \newtheorem{definition}{Definition}
\theoremstyle{definition} 
\theoremstyle{definition} 
\theoremstyle{definition} \newtheorem{lemma}{Lemma}
\theoremstyle{plain} 
\theoremstyle{definition} \newtheorem{assumption}{Assumption}
\theoremstyle{definition} 
\theoremstyle{definition} 
\theoremstyle{definition}
\theoremstyle{definition}\newtheorem{proposition}{Proposition}
\theoremstyle{definition} 
\theoremstyle{definition} 
\theoremstyle{definition} 
\theoremstyle{definition} 
\theoremstyle{definition} 
\theoremstyle{definition}

\DeclareMathOperator*{\argmax}{argmax}

\theoremstyle{definition} 
\usepackage{mathtools}
\titleformat{\section}
		{\bfseries\center \MakeUppercase }
         {\thesection}
        {0.5em}
        {}
        []

\titleformat{\subsection}[runin]
        {\normalfont\bfseries}
        {\thesubsection}
        {0.5em}
        {\addperiod}
        []
\newcommand{\addperiod}[1]{#1.}
\allowdisplaybreaks

\begin{document}

\title{\textsc{Racing to Ruin}\thanks{Fudenberg: MIT Department of Economics \url{drew.fudenberg@gmail.com}; Koh: Columbia University Department of Economics \url{andrew.koh@columbia.edu}. We thank [acknowledgments to be added] and National Science Foundation  grant SES-2417162 for financial support.}}

\author{\makebox[.25\linewidth]{Drew Fudenberg}\\ \normalsize{MIT} \\  
\and \makebox[.25\linewidth]{Andrew Koh}\\ \normalsize{Columbia}}

\date{\today}
\maketitle
\begin{abstract}
We study R\&D competition in the shadow of disaster: advancing the technology frontier raises the risk of permanently ending all firms’ payoffs. Under perfect monitoring and common knowledge of rationality, the equilibrium frontier is bounded below by the optimal stopping time of a monopolist, and above by that of a representative firm that persistently but mistakenly believes its rival is about to stop. We then analyze how the frontier is shaped by transparency (speed of monitoring) and trust (belief in the rationality of rival firms). 

\end{abstract}

\thispagestyle{empty} \vspace{-1em}

\newpage \setcounter{page}{1}

\section{Introduction}

In January 2026 at the World Economic Forum, Demis Hassabis, chief executive of Google DeepMind, was asked whether he would advocate a pause if he knew that every other competitor would pause too: ``I think so.'' Dario Amodei, chief executive of Anthropic, concurred: ``I would prefer that. I think that would be better for the world.'' \citep{hashim2026davos}. This suggests that each firm would prefer to slow down, but only if its rivals do the same. But Amodei went on to add: ``it's very hard to have an enforceable agreement where they slow down and we slow down.'' Why exactly is coordination hard? And what would it take? 

To answer these questions, we develop a simple model of R\&D competition between duopolists in the shadow of disaster:\footnote{The model may also apply to players as nation states e.g., US-China competition.} as frontier firms scale the technology, they raise the hazard of an event that permanently drives all firms' flow payoffs to zero. The hazard is a known function of the firms' technology levels, and it comes from developing the technology, not from using it. Stopping therefore eliminates future disaster risk, even as firms continue to enjoy the flow payoffs from the technology they have already built; it is as if disaster strikes exactly when the technology crosses some unknown threshold. Only the leader's state governs whether that threshold is crossed; whether their rival has technology that is nearly as advanced is irrelevant.

Each firm's flow profit rises in its own technology and falls in that of its rival.  When risk is low relative to the profits from scaling, a firm prefers to keep developing regardless of what its rival does: the competitive motive dominates. But when risk is high relative to the marginal profit from scaling, a firm prefers to stop if and only if its rival stops as well, so the problem is coordination, not competition.

 We show that even under perfect observability  competition still carries the technology frontier to a well-defined threshold, and that degrading the timeliness of information pushes the frontier further out. We also show that there is a tight connection between the equilibrium technology frontier and the technology chosen by a single representative agent that when both firms are active  evaluates the prospect of scaling further as if the technology of their  competitor is frozen at $t$. Thus they are persistently overoptimistic about the profits from scaling, and stop only at the time $\tau^R$ when even this overoptimism about profits does not outweigh the disaster risk. 

When firms observe each other's behavior perfectly, $\tau^R$ is an upper bound on the frontier firm's technology. The logic is that observability facilitates cooperation: at $\tau^R$, a firm that is contemplating stopping knows that if it 
 stops, their rival will see the stop. We then develop a matching lower-bound on the technology frontier across all equilibria, and show that the two bounds are close in several parametric models of competition. 

The argument above relied on both \emph{transparency}---firms' stopping decisions are perfectly and immediately observed---and \emph{trust}: it is common knowledge that firms are rational. We relax each in turn. We first suppose that stopping is observed with delay: a firm learns of its rival's stop after a random detection lag distributed exponentially with rate $\omega$ indexing the transparency/opacity of the environment. When monitoring is very noisy, news travels too slowly for a unilateral stop to be reciprocated in time, so racing forever becomes self-enforcing: there are then equilibria in which neither firm ever stops and disaster arrives with probability one, even though no such equilibrium survives under perfect transparency. When monitoring is sufficiently precise, every equilibrium stops in finite time, but a new temptation appears: each firm would like to stop second, and exits only upon confirmation that the rival has stopped. We show that under a regularity condition on the tail payoffs, the frontier overshoots $\tau^R$ by at most $O(1/\omega)$ in expectation. 

We next suppose that, in addition to imperfect transparency, firms are also unsure about the rationality of their rival---believing it to be a ``crazy type'' that never stops in the spirit of \citet{kreps1982reputation}. To understand the role of trust, we specialize to a stationary environment with $\tau^R = 0$, so that at each technology state, firms would prefer to stop if their rival also does so. 
Our main result partitions the trust (belief in their rival's rationality) into three zones. With low trust, every equilibrium races to ruin: the disaster arrives with probability one. With intermediate trust,  immediate stopping and racing to ruin are both equilibria. With high trust, in every equilibrium, the probability that two rational firms race forever vanishes quadratically in the prior odds ratio of rationality. Transparency is double-edged: faster detection makes it cheaper to wait for confirmation that a rival has stopped before stopping oneself instead of stopping unconditionally, so at intermediate trust, increasing transparency can first destroy the early-stopping  equilibrium (by making this free-riding deviation attractive) before restoring it as detection becomes fast enough to make stopping self-enforcing.

\section{Related Literature}

Our paper draws on several literatures: dynamic R\&D and patent races, wars of
attrition and exit games, contests, and the recent
economics of AI and existential risk. 

\paragraph{R\&D races and preemption.}
A large literature studies dynamic competition to develop a new technology. The early models of \citet{loury1979market}, \citet{lee1980market}, and \citet{dasgupta1980uncertainty} relate market structure to the intensity and speed of R\&D through a single investment choice;   \citet{reinganum1982dynamic} introduces commitment to  time-varying open-loop  strategies. \citet{fudenberg1983preemption}, \citet{harris1985perfect} and \citet{ harris1987racing} study preemption, leapfrogging, and the discouragement of laggards in models where firms observe and respond to their rivals' progress.\footnote{\citet{fudenberg1983preemption} introduced multistage patent races, and
\citet{aghion2001competition} embeds step-by-step innovation in a growth
setting. } Preemption
also drives the timing games of \citet{fudenberg1985preemption}. \citet{weeds2002strategic} and \citet{hopenhayn2011preemption} study
strategic delay and preemption when investment is a real option or progress is
privately observed. Our model shares the competitive force of these races---flow
profits rise in own technology and fall in rivals'---but because the frontier
generates a common disaster hazard, the object of interest is how far the frontier is carried before firms stop.

\paragraph{Wars of attrition and exit.}
When disaster risk is high relative to profits, our firms would prefer to stop
if and only if their rivals stop, and under imperfect monitoring the equilibrium
becomes a war of attrition. The continuous-time war of attrition originates with
\citet{maynardsmith1974theory} and is developed by \citet{hendricks1988war},
\citet{bulow1999generalized}, and \citet{krishna1997analysis}, which relate it
to the all-pay auction; the exit dynamics of declining industries in
\citet{fudenberg1986theory} are a close antecedent to our stopping game.
Our extension to imperfect monitoring of stopping time  is related to other work in which
information about others' behavior arrives with delay:
\citet{fudenberg2014delayed} studies repeated games with observation lags, and
\citet{daley2012waiting} analyzes dynamics driven by the stochastic arrival of
news.  And our extension to  ``crazy'' types builds on the predation model of \citet{kreps1982reputation}.

\paragraph{Contests.}
Our leading example casts flow payoffs as shares of a prize determined by a
logit contest success function, following \citet{tullock1980efficient} and the
axiomatic treatment of \citet{skaperdas1996contest}. We also draw on the broader
theory of contests and all-pay auctions, e.g.  \citet{baye1996allpay},
\citet{moldovanu2001optimal}, \citet{siegel2009allpay}, and 
\citet{konrad2009strategy}. 




\paragraph{AI races, existential risk, and regulation.}
Our motivation comes from the economics of transformative and potentially
dangerous technologies. \citet{jones2024ai} and \citet{trammell2024existential}
study the optimal growth--risk tradeoff when scaling raises existential risk,
and \citet{ord2020precipice} argues that a wiser society would proceed more
cautiously; our results identify a strategic source of excessive
risk-taking that arises even with risk-neutral, patient agents.
\citet{armstrong2016racing} studies a static race between AI developers who choose how much safety to sacrifice for speed, so the issues of observability and dynamic coordination do not arise;  \citet{buenodemesquita2026agi} enrich this static game with market structure. \citet{koh2024robust} analyzes robust regulation of multiple AI firms with differing risk aversion but does not model competition.

\section{Baseline Model}
\paragraph{Time and Technology} There are two firms $\mathcal N={1,2}$. Each firm advances its technology at unit speed until it stops, and stopping is irreversible. In the baseline model, firms perfectly and immediately observe their rival’s past stopping behavior, and each has access to private randomization. A strategy specifies a randomized stopping rule after every history; a strategy profile and the private randomization devices induce $\overline{\mathbb R}_+$-valued stopping times $\bm\tau=(\tau_1,\tau_2)$ adapted to the firms’ information. A firm’s time-$t$ stopping decision cannot depend on its rival’s simultaneous time-$t$ decision. Since stopping is irreversible, firm $i$’s technology at calendar date $s$ is $s\wedge\tau_i$, and its eventual technology is $\tau_i$.\footnote{As shown in \cref{app:extform}, this stopping-time formulation  is without loss, because it generates exactly the outcome distributions generated by randomized stopping times in the sense of \citet{touzi2002dynkin}: adapted, right-continuous, nondecreasing $[0,1]$-valued processes recording cumulative stopping probability. Because payoffs are continuous in the stopping dates, the payoffs from stopping just before or just after one’s rival converge to the simultaneous-stopping payoff as the gap between the stopping dates vanishes. Thus, resolving an atom at a common stopping date does not require the extended strategy space of \citet{fudenberg1985preemption}.}

\paragraph{Flow Payoffs} For a technology profile $\bm t=(t_1,t_2)$, firm $i$ receives flow payoff
$
\pi(t_i,t_j),
$
where $t_i$ is firm $i$'s own technology state ($=$ time spent in the race) and $\bm{t}_{j}$ is its rival's technology state. We impose symmetry across firms: every firm has the same payoff function $\pi$, and $\pi$ treats rivals anonymously. Thus relabelling the firms relabels their payoffs but does not change any payoff level. We write $D_i \pi(t_i,t_j)$ for the derivative with respect to firm $i$'s own technology state and $D_j\pi(t_i,t_j)$ for the derivative with respect to firm $j$'s technology state.

\paragraph{Disaster risk} The \emph{frontier technology scale} of a technology profile $\bm t$ is $T=\max_i t_i$. There is a risk of a market collapse that arrives at a rate $\lambda: \R_+ \to \R_+$ where $\lambda(x)$ is the hazard rate when the frontier technology scale is $x$. Flow payoffs are zero after the collapse. Our leading interpretation of the collapse is the extinction of the industry: after it occurs, no firm is around to enjoy flow payoffs.\footnote{The starkest examples are `loss of control' scenarios in which beyond an (unknown) technology threshold, AI might be able to rapidly improve itself leading to extinction\citep{bostrom2014superintelligence,bengio2024managing}.}  This follows the analyses of analyze growth under existential risk in \citet{jones2024ai} and \citet{trammell2024existential}, and the value-of-life literature \citep{rosen1988value,murphy2006value,hall2007value}, in which the level of flow utility relative to death pins down how much mortality risk an agent will accept in exchange for higher consumption.

Firms discount the future at a common discount rate of $r > 0$. For a realized stopping profile $\bm\tau\in\overline\R_+^2$, define the instantaneous disaster risk at time $s$ by
\[
\lambda_s(\bm\tau):=\lambda(s)\cdot \mathbf{1}\left\{s\le \max_{j\in\mathcal N}\tau_j\right\}.
\]
Let $\Lambda_s(\bm\tau):=\int^s_0 \lambda_u(\bm\tau)\de u$ be cumulative risk up to $s$. Hence, payoffs can be written as
\[
U_i(\bm\tau)
:=
\int^\infty_0 e^{-rs - \Lambda_s(\bm\tau)}
\pi\Big(s\wedge\tau_i,\ s\wedge\tau_j\Big)\de s.
\]
To ensure that payoffs are finite and stopping behavior is well-behaved, we will maintain the following assumption on payoffs and risk. 
\begin{assumption}\label{ass:payoffs}
    $\pi$ and $\lambda$ satisfy:
    \begin{itemize}
      \item[(i)] \textbf{Increasing hazard}: $\lambda$ is continuously differentiable and nondecreasing.
        \item[(ii)] \textbf{Regularity}: $\pi$ is twice continuously differentiable and strictly positive, strictly increasing in own technology and strictly decreasing in each rival's technology: $D_i \pi>0$ and, for every $j\neq i$, $D_j\pi<0$.
        \item[(iii)] \textbf{Log-supermodularity}: for all firms $i\neq j$, $D_{ij}\log\pi(t_i,t_j)\geq0$.
        \item[(iv)] \textbf{Single-crossing}: for every firm $i$, the map 
$
\left(t_i,t_j\right)\mapsto D_i\log\pi(t_i,t_j)-\lambda(t_i)
$
is continuous and, for every rival technology $\bm{t}_{-i}$, strictly decreasing in $t_i$ over $[t_j,+\infty)$.

      \item[(v)] \textbf{Integrability:}
\[
\int_0^\infty e^{-rs}\,\sup_{t\le s}\Big\{e^{-\int_0^t\lambda(u)\de u}\,\pi(t,0)\Big\}\,\de s\ <\ \infty .
\]
    \end{itemize}
\end{assumption}

\cref{ass:payoffs} (i)--(v) are fairly mild.   The increasing hazard assumption is consistent with $\lambda$ constant, which fits the model where the unknown disaster threshold is exponentially distributed. 
Regularity (ii) asks only that flow payoffs be smooth, strictly positive, increasing in own technology and decreasing in each rival's. Log-supermodularity (iii) says that when a rival's technology is higher, the marginal payoff from increasing one's own technology is weakly higher; this holds whenever payoffs are additively or multiplicatively separable. Single-crossing (iv) holds if log profits are concave in own technology while the hazard is nondecreasing, and at least one of these conditions is strict. Because $\lambda$ is nondecreasing. Integrability (v) guarantees that every strategy profile generates a finite expected payoff, uniformly (\cref{lem:cont}).

\paragraph{A representative firm with misspecified beliefs} We begin by constructing a single agent with incorrect beliefs about the other firm's strategy. Consider a history at which no disaster has arrived and both firms have continued until $t$. Fix a generic firm $i\in\mathcal N$. The representative firm compares stopping immediately with continuing alone to some $\tau\geq t$, taking the other firm's technology as frozen at $t$. The continuation payoff from choosing $\tau$ is thus

\[
R_t(\tau)
:=
\underbrace{\int^\tau_t e^{-\int^s_t (r+\lambda(z))\de z}\pi(s,t)\de s}_{\substack{\text{Incremental payoff from scaling} \\ \text{own tech from $t$ to $\tau$}}}
+
\underbrace{e^{-\int^\tau_t (r+\lambda(z))\de z}\frac{\pi(\tau,t)}{r}}_{\text{Payoff from $\tau$ on}}
\]
where we use the convention that 
$R_t(+\infty)
:=
\lim_{\tau\to+\infty}R_t(\tau)$  for the payoff from never stopping. We define 
\[
\tau^R
:=
\inf\left\{
t\geq 0:
R_t(t)
=
\sup_{\tau\geq t} R_t(\tau)
\right\}.
\]
as the solution to the misspecified firm's problem. The following lemma establishes that there exists a unique solution to the representative agent's stopping problem. 

\begin{lemma}\label{lem:R-single-peaked}
For each $t\geq 0$, $R_t$ has a unique maximizer $\chi(t)$ on $[t,+\infty) \cup +\infty$ and $\chi$ is continuous in $t$. Moreover,
$
R_t(t)=\sup_{\tau\geq t}R_t(\tau) \text{ iff }
D_i\log\pi(t,t)\leq\lambda(t).$ 
Thus, if $\tau^R<+\infty$, then
$
R_{\tau^R}(\tau^R)=\sup_{\tau\geq \tau^R}R_{\tau^R}(\tau)$  and if $t<\tau^R$, then $D_i\log\pi(t,t)>\lambda(t)$.
\end{lemma}

\cref{lem:R-single-peaked} reduces the representative firm's problem to a race between two rates: continuing for an instant scales all future flow profits up at the proportional rate $D_i\log\pi(t,\ldots,t)$, while exposing the  continuation value to catastrophe at rate $\lambda(t)$. Single-crossing (\cref{ass:payoffs}(iv)) ensures that, holding the rival's frozen, the former crosses the latter only once from above so $\tau^R$ is well defined. The threshold $\tau^R$ is the first date at which immediate stopping becomes optimal under the optimistic benchmark in which the rival's technology remains frozen. It will be our benchmark for how far competition scales the technology frontier.

\section{Perfect Transparency}
We start by analyzing the case of perfect transparency in which firms' actions (stopping decisions) are perfectly and immediately observed. We study subgame perfect equilibria (henceforth just equilibria), and  denote equilibrium outcomes with $\SPE(\pi,\lambda)$.

\paragraph{Upper bounds on the frontier}
We first bound the frontier from above, for every equilibrium, by the representative threshold $\tau^R$. 

\begin{proposition}[Upper bound]\label{prop:bounds}
For every $\bm\tau\in\SPE(\pi,\lambda)$, $\max_i\tau_i\le\tau^R$ almost surely. 
\end{proposition}

A sole survivor at $\tau^R$ prefers to stop even if its rival's technology is frozen. Hence, if one of the two active firms stops at $\tau^R$, the remaining firm also stops then; anticipating this, no firm continues past the threshold. More broadly, in the game after crossing $\tau^R$ each player prefers to stop as long as the other player does so quickly enough. But because actions are perfectly observed, the logic of backward induction ensures that coordination to stop by $\tau^R$ succeeds.\footnote{The argument underpinning \cref{prop:bounds} is quite general and can be extended to $n \geq 2$ firms: each firm who stops at $\tau^R$ knows that its remaining $n - 1$ rivals would observe it, and each of these rivals know that if they stop, then its remaining $n-2$ rivals would observe their stop, and so on until a single firm remains and prefers to stop. See also \cite{gale1995dynamic} who studies a dynamic coordination game with irreversible investment.} The proof is simple and instructive:

\begin{proof}[Proof of \cref{prop:bounds}]
If $\tau^R=+\infty$, the claim is immediate. Suppose $\tau^R<+\infty$. First note that, for every firm $i$, every rival technology level$x_j\leq\tau^R$ for all $j\neq i$, and every $x_i\geq\tau^R$,
$
D_i\log\pi(x_i,x_j)-\lambda(x_i)
\leq
D_i\log\pi(x_i,\tau^R)-\lambda(x_i) \leq 0,$ 
where the first inequality uses log-supermodularity and the second uses \cref{lem:R-single-peaked} at $\tau^R$ together with \cref{ass:payoffs}(iv), fixing rivals at $\tau^R$. The last inequality is strict when $x_i>\tau^R$. Hence, after any history at time $\tau^R$ in which the rival's technology is weakly below $\tau^R$, a firm's frozen-rival payoff is uniquely maximized by stopping immediately.

We claim that, after any history at time $\tau^R$, no $\SPE$ continuation carries the frontier above $\tau^R$ with positive probability. Consider first a history at $\tau^R$ at which only one firm is active, its rival having stopped at some $y\le\tau^R$. The active firm's continuation payoff from stopping at $\sigma\ge\tau^R$ is its frozen-rival payoff, which by the preceding paragraph is uniquely maximized by stopping immediately; sequential rationality therefore forces it to stop at $\tau^R$. Consider next a history at $\tau^R$ with both firms active, and suppose the continuation carries the frontier above $\tau^R$ with positive probability; then some firm $i$ remains active past $\tau^R$ with positive probability. Let firm $i$ deviate by stopping at $\tau^R$. After the deviation its rival is a sole survivor whose (frozen) opponent sits at $\tau^R$; by the one-active-firm step, the rival stops at $\tau^R$ immediately. The deviation therefore gives firm $i$  the payoff of stopping at $\tau^R$ with its rival frozen at $\tau^R$. On the equilibrium path, firm $i$'s payoff is no larger than the payoff from continuing while its rival is frozen at its time-$\tau^R$ technology: the rival's actual technology is weakly higher, and frontier risk lasts weakly longer. By the preceding paragraph, that frozen-rival payoff is strictly below the payoff from stopping at $\tau^R$, contradicting sequential rationality. This proves the claim.


If an equilibrium had $\max_{i\in\mathcal N}\tau_i>\tau^R$ on a positive-probability event, then conditional on survival to time $\tau^R$ the public history at $\tau^R$ would have at least one active firm and any stopped firm at technology weakly below $\tau^R$. The claim rules out any continuation from that history with frontier above $\tau^R$, a contradiction.
\end{proof}


\paragraph{Lower bounds on the frontier} Note that the equilibrium is not necessarily unique and the equilibrium set need not form a lattice. We will develop simple lower bounds on the frontier technology across all equilibria.

Define 
\[
\underline \tau := \inf \Big\{ t: \underbrace{D_i\log\pi(t,0)}_{\text{Origin index}} \leq \lambda(t) \Big\}
\]
where the origin index is the marginal payoff to a firm from pushing the technology at $t$ when its competitor's technology are at $0$ i.e., it is a monopolist. $\underline \tau$ is the time at which a monopolist finds it optimal to stop.

\begin{proposition}[Lower bound]\label{prop:frontier}
For every $\bm\tau\in\SPE(\pi,\lambda)$, $\max_i \tau_i \ge\underline\tau$ almost surely. 
\end{proposition}

\cref{prop:frontier} states that across all equilbria the technology will advance to at least $\underline \tau$. This is because at the moment the last firm contemplates stopping, its rival is weakly behind. Log-supermodularity (\cref{ass:payoffs}(iii)) says that rival technology raises the marginal return to the firm's own scaling. Hence the last active firm faces at least as strong an incentive to push the technology as a monopolist would. To see why, let $\chi(t)$ denote the optimal stopping time of the remaining firm once the other firm stops at $t$.  \cref{ass:payoffs} (iii)) implies that $\chi(t)\ge\underline\tau$, so $\max_{i}\tau_{i}< \underline\tau$ can occur with positive probability only if there is positive probability that firms stop at some common date $t < \underline \tau$. But this cannot be an equilibrium: at a history where both firms have strictly positive probability of stopping, either would do  do strictly better by stopping later.

Propositions~\ref{prop:bounds} and~\ref{prop:frontier} show that the equilibrium frontier lies in the interval $[\underline{\tau},\tau^R]$. The next lemma gives a simple characterization of when $\underline{\tau}=\tau^R$ so that the interval collapses to  a point.  
\begin{lemma}[Degenerate-bracket test]\label{lem:bracket-test}
Suppose $0 < \tau^R < +\infty$. Then
\[
D_i\log\pi(\tau^R,0) \;\le\; \lambda(\tau^R) \;=\; D_i\log\pi(\tau^R,\tau^R),
\]
with equality on the left if and only if $\underline\tau = \tau^R$; if the inequality is strict, then $\underline\tau < \tau^R$.
\end{lemma}

\begin{proof}
Continuity of both indices and of $\lambda$ gives $D_i\log\pi(\tau^R,\tau^R)=\lambda(\tau^R)$ at the first crossing, and log-supermodularity then gives $D_i\log\pi(\tau^R,0)\le\lambda(\tau^R)$, whence $\underline\tau\le\tau^R$. Since $t\mapsto D_i\log\pi(t,0)-\lambda(t)$ is continuous and strictly decreasing (\cref{ass:payoffs}(iv), the rival's technology fixed at zero), its first crossing $\underline\tau$ falls strictly below $\tau^R$ when its value at $\tau^R$ is strictly negative. The final claim combines $T\le\tau^R$ (\cref{prop:bounds}) with $T\ge\underline\tau$ (\cref{prop:frontier}).
\end{proof}

\begin{proposition}[Pure priority]\label{prop:pure-priority}
    There exists a pure-strategy subgame perfect equilibrium.
\end{proposition}

The idea behind \cref{prop:pure-priority} is that we can fix one firm (say, $1$) as the designated first stopper. Along any history at which both firms are active, let firm $1$ pick its favorite stopping date, knowing that firm $2$ will respond optimally. This is just one among many equilibria; the following proposition develops a simple classification of all of them.

\begin{proposition}[Spreading and bunching]\label{prop:taxonomy}
For any $\bm\tau\in\SPE(\pi,\lambda)$, almost surely the realized profile is one of:
\begin{itemize}
  \item[(i)] \textbf{Spread:} $\min_i\tau_i<\max_i\tau_i<\tau^R$; or 
  \item[(ii)] \textbf{Bunched:} $\min_i\tau_i = \max_i\tau_i = \tau^R$.
\end{itemize}
\end{proposition}

\cref{prop:bounds,prop:frontier} already rule out the frontier exceeding
$\tau^R$ and falling short of $\underline\tau$. As in the proof of \cref{prop:frontier}, there cannot be positive probability that both firms plan
to stop simultaneously at some $t<\tau^R$.  A
sole survivor at $t$ optimally continues well past $t$, all the way to $\chi(t)>t$, which is strictly more profitable than
stopping at $t$. The prospect of this gain makes the simultaneous stop
self-defeating. And if a  firm
planned to stop first at some $t<\tau^R$, with the other  continuing to $\tau^R$, the early-stopper would do strictly
better by waiting until $\tau^R$ itself: This induces
its rival to either stop before $\tau^R$, gaining higher own technology, a
weaker rival, and less accumulated hazard, or stops simultaneously with
its rival at the ceiling, which is also an improvement. 
Together these two exclusions and the boubds from
\cref{prop:bounds,prop:frontier}, leave exactly the taxonomy: either
firms stop at different times with the frontier strictly inside the
bracket, or they stop together at the ceiling.

Finally, we develop some simple comparative statics for how the ceiling $\tau^R$ moves with the primitives. Write $\tau^R(\pi,\lambda)$ for the representative threshold associated with primitives $(\pi,\lambda)$. The key observation is an asymmetry in how $\pi$ and $\lambda$ enter the representative firm's stopping rule (\cref{lem:R-single-peaked}): the payoff function matters only through the \emph{diagonal index} $t \mapsto D_i\log\pi(t,t)$, whereas the hazard enters only as the threshold this index must fall below.

\begin{proposition}[Comparative statics]\label{prop:cs-tauR}
For $(\pi,\lambda)$ and $(\tilde\pi,\tilde\lambda)$ that satisfy\cref{ass:payoffs}: 
\begin{itemize}
    \item[(i)] If $D_i\log\tilde\pi(t,t)\ \ge\ D_i\log\pi(t,t)$ for every $t\ge0$, then $\tau^R(\tilde\pi,\lambda)\ge\tau^R(\pi,\lambda)$. \\ 
     If $D_i \log \tilde \pi(t,0) \geq D_i \log  \pi(t,0)$ for every $t \geq 0$ then $\underline \tau (\tilde \pi, \lambda) \geq \underline \tau (\pi,\lambda)$. 
    \item[(ii)] If $\tilde\lambda(t)\ge\lambda(t)$ for every $t\ge0$, then $\tau^R(\pi,\tilde\lambda)\le\tau^R(\pi,\lambda)$ and $\underline \tau(\pi, \tilde \lambda) \leq \underline \tau(\pi,\lambda)$.
\end{itemize}
\end{proposition}

\cref{prop:cs-tauR} shows how the frontier's ceiling and floor respond to the primitives. The ceiling $\tau^R$ is set by the diagonal index $D_i\log\pi(t,t)$, the neck-and-neck racing incentive: it strengthens with the sensitivity of payoffs to relative position (the contest weight $\gamma$, the complementarity $B$), pushing the ceiling out. The floor $\underline\tau$ is set by the origin index $D_i\log\pi(t,0)$, the marginal value of scaling for a firm whose rival is stuck at the origin, so it moves with the technology's standalone profitability rather than with rivalry. Raising the hazard ($\lambda_0$, $\kappa$) pulls both bounds in.

\begin{proof}[Proof of \cref{prop:cs-tauR}]
By \cref{lem:R-single-peaked} and the definition of $\underline\tau$, both thresholds are first-crossing times of an index with the hazard:
$
\tau^R(\pi,\lambda)=\inf\{t\ge0:\partial_i\log\pi(t,t)\le\lambda(t)\}; 
$
and 
$
\underline\tau(\pi,\lambda)=\inf\{t\ge0:\partial_i\log\pi(t,0)\le\lambda(t)\}.
$
All four claims follow from the same set inclusion. If an index is replaced by a pointwise weakly larger one, any $t$ in the new crossing set belongs a fortiori to the old one, so the crossing set shrinks and its infimum weakly rises; this gives both statements in (i), with the diagonal index for $\tau^R$ and the origin index for $\underline\tau$. If instead the hazard is replaced by a pointwise weakly larger one, any $t$ in the old crossing set belongs to the new one, so the crossing set grows and its infimum weakly falls; this gives both statements in (ii).
\end{proof}

\subsection{Illustrative Examples}
We now illustrate how the upper- and lower-bounds behave across a number of parametric examples. These examples specialize to an exponential hazard rate, $\lambda(t)=\lambda_{0}e^{\kappa t}$.

\begin{example}[Score contests]
Firm $i$ earns a share $s_i(\bm x)=e^{\gamma x_i}/(e^{\gamma x_1}+e^{\gamma x_2})$, with $\gamma\ge0$, of an own-technology prize $A_0e^{\beta x_i}$, so its flow payoff is
$
\pi(x_i,x_j)=A_0e^{\beta x_i}\,s_i(\bm x).
$
At a symmetric profile, the marginal log-payoff from scaling is
$
D_i\log\pi(t,t)=\beta+\frac{\gamma}{2},
$
while with the rival at the origin it is
$
D_i\log\pi(t,0)=\beta+\frac{\gamma}{e^{\gamma t}+1}.
$
Suppose first that $\kappa>0$. Then
\[
\tau^R=\left[\frac1\kappa\log\frac{\beta+\gamma/2}{\lambda_0}\right]^+,
\qquad
\underline\tau
=
\inf\left\{t\geq 0:
\beta+\frac{\gamma}{e^{\gamma t}+1}
\leq
\lambda_0e^{\kappa t}
\right\}.
\]
If $\lambda_0\geq\beta+\gamma/2$, both thresholds equal zero. If $\lambda_0<\beta+\gamma/2$, the ceiling is interior; the two thresholds coincide when $\gamma=0$, so that the prize is split evenly regardless of play, and satisfy $\underline\tau<\tau^R$ when $\gamma>0$. If instead $\kappa=0$, the diagonal index is constant: the ceiling is zero when $\beta+\gamma/2\leq\lambda_0$ and infinite otherwise, and the $1/\kappa$ formula does not apply.

\cref{fig:score-twofirm} plots both thresholds as the curvature of the prize and the level of the hazard vary, together with the pure priority equilibrium at selected parameter values: it bunches at the ceiling when the contest is shallow and spreads once the contest is steep. \cref{fig:br} shows why: 
the left panel varies the share sensitivity $\gamma$. The bracket $[\underline\tau,\tau^R]$ collapses at $\gamma=0$, where the prize is split evenly, and widens with $\gamma$: the ceiling rises with rivalry, while the floor first rises (encouragement) and then falls (discouragement: a laggard's share gain vanishes when the contest is steep). The right panel varies the hazard level $\lambda_0$. Both thresholds fall and reach zero together at $\lambda_0=\beta+\gamma/2$. Dots mark the pure priority equilibrium at selected parameter values: it bunches at the ceiling for moderate curvature, and the two stops separate strictly inside the bracket once the contest is steep ($\gamma\gtrsim1.8$).

\end{example}

\begin{figure}[H]\centering
\includegraphics[width=\textwidth]{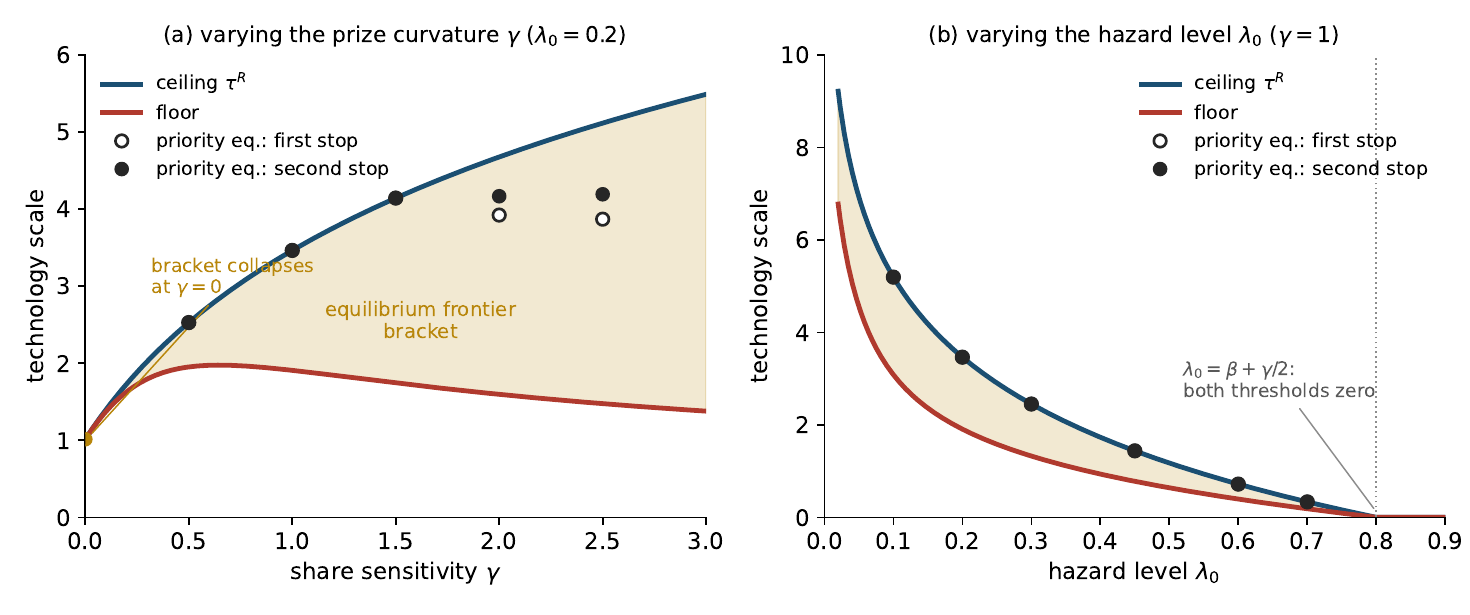}
\caption{Score-contest thresholds; parameters: $\beta=0.3$, $\kappa=0.4$, $r=0.1$; the left panel fixes $\lambda_0=0.2$, the right panel fixes $\gamma=1$}
\label{fig:score-twofirm}
\end{figure}

\cref{fig:br} illustrates both bunching and spreading regimes for the score contest. It shows the follower's reaction curve (red) and the first stopper's isopayoff curve in the plane of stopping dates. The first stopper's payoff rises toward the lower right: delaying its own stop raises its flow profit, while the follower's delay erodes it and prolongs the hazard. The priority equilibrium maximizes the first stopper's payoff along the reaction curve. At $\gamma=2$ (right panel) the optimum is an interior tangency: the isopayoff curve through the equilibrium touches the reaction curve at the dot and lies below it elsewhere. At $\gamma=1$ (left panel) the payoff rises along the entire reaction curve, so the optimum is the corner $(\tau^R,\tau^R)$ where the reaction curve meets the $45^\circ$ line, and the firms bunch.

\begin{figure}[t]\centering
\includegraphics[width=.95\textwidth]{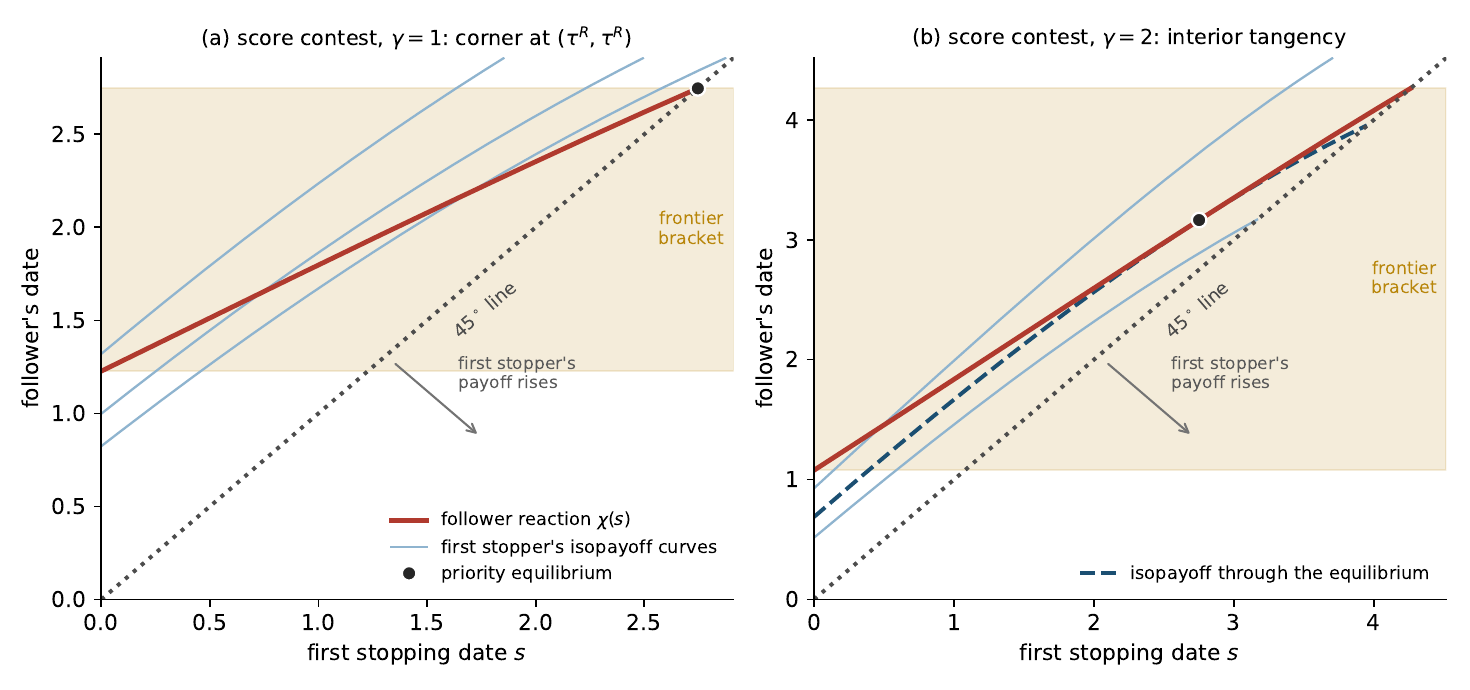}
\caption{Follower's reaction curve; parameters: $\beta=0.1$, $\lambda_0=0.2$, $\kappa=0.4$, $r=0.1$}
\label{fig:br}
\end{figure}

\begin{example}[Cournot competition] Let inverse demand be $P=a-b(q_1+q_2)$, and let capability lower marginal cost, $c(x)=c_0e^{-\delta x}$. Assume $a>2c_0$; since $c(x)\in(0,c_0]$, every technology profile then has $M_i=a-2c(x_i)+c(x_j)\geq a-2c_0>0$, so the static Cournot quantities are interior and the log-profit index below is defined globally. In the interior Cournot equilibrium,
\[
\pi_i(\bm x)=\frac{M_i^2}{9b},
\qquad
M_i:=a-2c(x_i)+c(x_j),
\]
so
\[
D_i\log\pi_i(\bm x)
=
\frac{4\delta c(x_i)}{a-2c(x_i)+c(x_j)}.
\]
The ceiling and the floor are therefore
\[
\tau^R
=
\inf\left\{t\geq0:
\frac{4\delta c_0e^{-\delta t}}{a-c_0e^{-\delta t}}
\leq
\lambda_0e^{\kappa t}
\right\},
\qquad
\underline\tau
=
\inf\left\{t\geq0:
\frac{4\delta c_0e^{-\delta t}}{a+c_0-2c_0e^{-\delta t}}
\leq
\lambda_0e^{\kappa t}
\right\}.
\]
For $c_0>0$, the two thresholds coincide when $4\delta c_0\leq\lambda_0(a-c_0)$, in which case both equal zero; when $4\delta c_0>\lambda_0(a-c_0)$, both crossings are positive and $\underline\tau<\tau^R$. When $c_0=0$, the first-crossing definitions give $\underline\tau=\tau^R=0$, with the logarithmic formulas interpreted by continuity.
\cref{fig:cournot-twofirm} illustrates these numerically. The left panel varies the hazard growth $\kappa$: at $\kappa=0$, $\underline\tau=0.511$ and $\tau^R=0.588$. As $\kappa$ rises, both crossings move earlier and the bracket tightens. A larger market size $a$ also moves both bounds earlier: with high baseline demand, firms earn substantial margins, so the proportional gain from further cost reduction is smaller. The right panel varies the market size $a$; both thresholds reach zero together at $a=c_0(4\delta+\lambda_0)/\lambda_0$. Dots mark the pure priority equilibrium at selected parameter values: at this calibration it bunches at the ceiling, so the two dots coincide on $\tau^R$. 

\begin{figure}[H]\centering
\includegraphics[width=\textwidth]{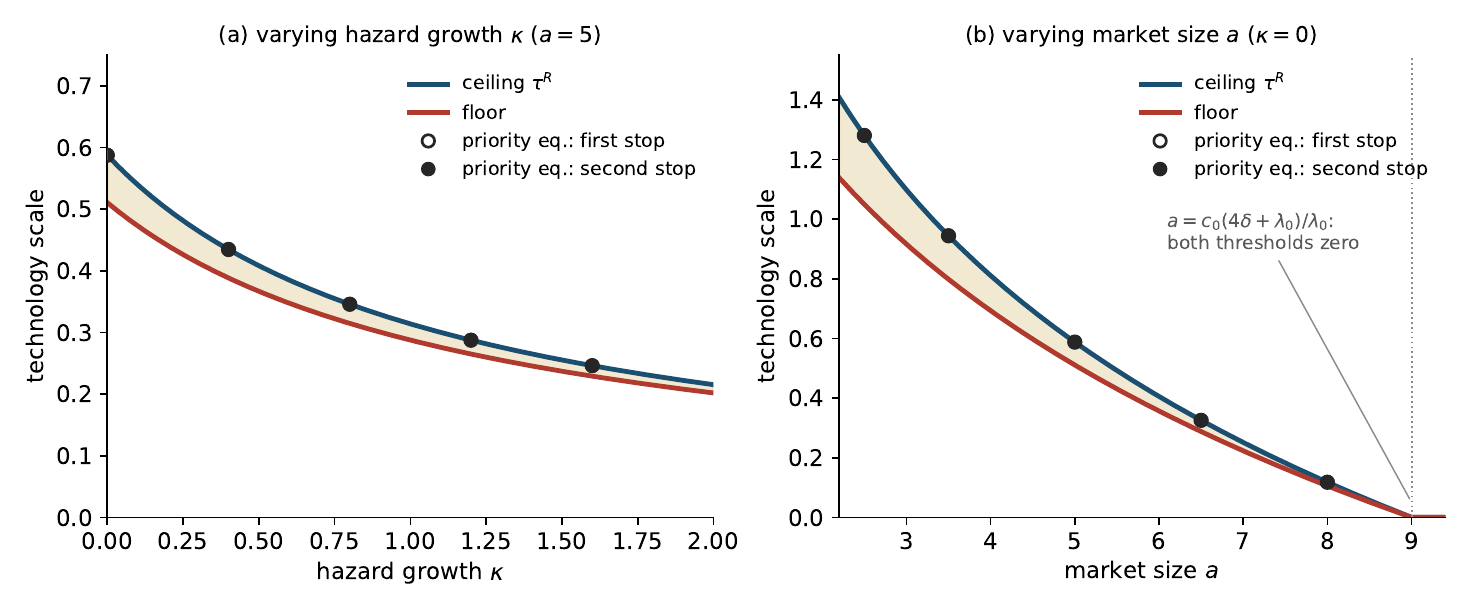}
\caption{Cournot bracket; parameters: $c_0=1$, $\delta=1$, $\lambda_0=0.5$; the left panel fixes $a=5$, the right panel fixes $\kappa=0$}
\label{fig:cournot-twofirm}
\end{figure}
\end{example}

\begin{example}[Log-separable payoffs]
Finally, suppose log-flow payoffs are additively separable with one interaction term. Let $\psi(t):=1-e^{-t}$ and take
\[
\log\pi(t_i,t_j) = A\psi(t_i) - C\psi(t_j) + B\psi(t_i)\psi(t_j),
\]
with $A>0$, $C\geq B\geq 0$, and $C>0$. The parameter $B$ controls the complementarity between own and rival technology; $C$ governs how much the rival's progress erodes a firm's profits. The origin index and its diagonal counterpart are
$
D_i\log\pi(t,0) = A\psi'(t),$ and $
D_i\log\pi(t,t) = \psi'(t)\bigl(A + B\psi(t)\bigr),$ 
with gap $B\psi'(t)\psi(t)$, so the two indices coincide when $B=0$.  The hazard-free parts of \cref{ass:payoffs} are immediate:
$D_i\log\pi=\psi'(t_i)\bigl[A+B\psi(t_j)\bigr]>0$;
$D_j\log\pi=\psi'(t_j)\bigl[B\,\psi(t_i)-C\bigr]<0$, since $\psi<1$ and $C\ge B$;
$D_{ij}\log\pi=B\,\psi'(t_i)\,\psi'(t_j)\ge0$; and since 
$\pi(t,t)$ is bounded above the integrability condition  (v) holds.
The slope of the follower's response to a first stop is governed by $D_{ij}\log\pi$, so under strong complementarity a marginal delay of the first stop drags the follower's stop outward by a comparable amount, which can pull the first stopper's optimum strictly below the ceiling.

\begin{figure}[htbp]\centering
\includegraphics[width=\textwidth]{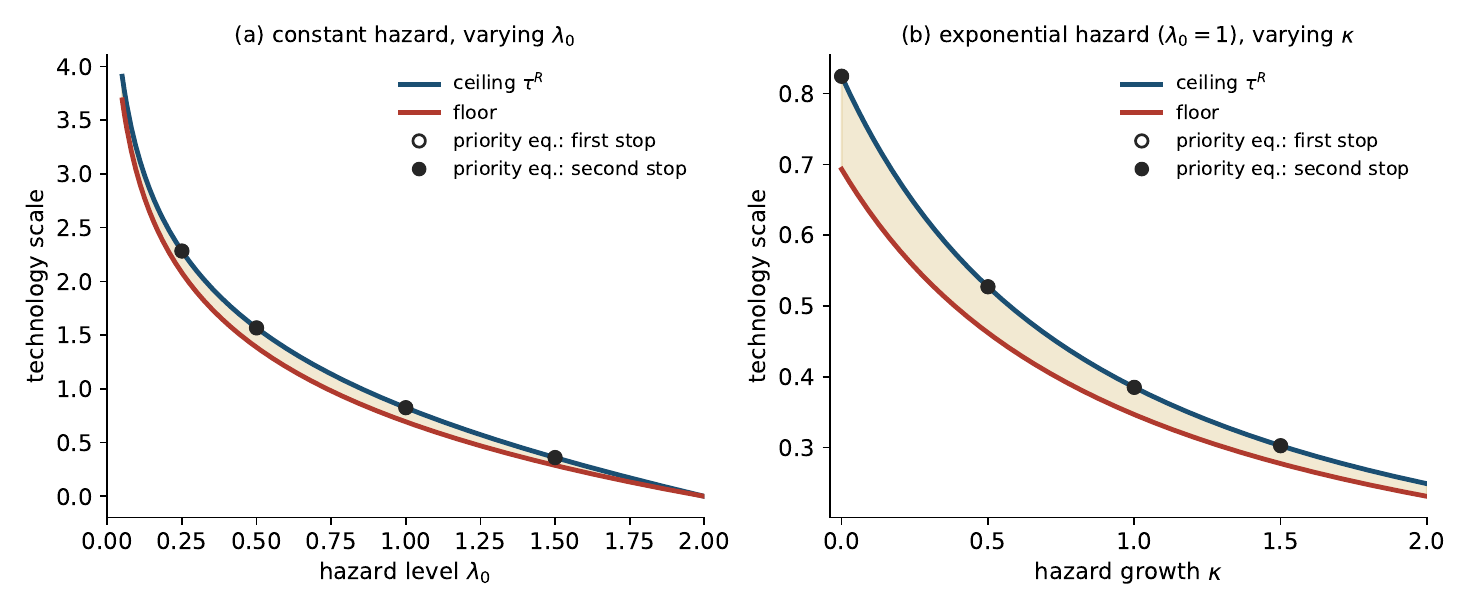}
\caption{Separable bracket; parameters: $A=2$, $B=\tfrac12$, $r=0.1$, $C=1$}
\label{fig:separable}
\end{figure}
\end{example}

\cref{prop:bounds,prop:frontier} bracket the equilibrium frontier $T:=\max_i\tau_i$ within $[\underline\tau,\tau^R]$ in every equilibrium, so competition carries the frontier all the way to the representative threshold precisely when the bracket is degenerate: $\underline\tau=\tau^R$. By \cref{lem:bracket-test}, $\underline\tau<\tau^R$  when the origin index falls strictly below the hazard at $\tau^R$; here the origin index differs from the diagonal index by the gap, so for $0<\tau^R<+\infty$,
\[
\underline\tau<\tau^R
\ \Longleftrightarrow\
B\,\psi'(\tau^R)\,\psi(\tau^R)>0
\ \Longleftrightarrow\
B>0 ,
\]
using $\psi'>0$ and $\psi(\tau^R)>0$: the bracket is degenerate iff $B=0$, and the magnitude of the complementarity is irrelevant. With an exponential hazard both thresholds are explicit. We have
$\partial_i\log\pi(t,t)=e^{-t}\big[A+B(1-e^{-t})\big]$ and
$\partial_i\log\pi(t,0)=Ae^{-t}$,
so for $\lambda_0\in(0,A)$ the floor solves $Ae^{-t}=\lambda_0e^{\kappa t}$, giving $\underline\tau=\log(A/\lambda_0)/(1+\kappa)$, while the ceiling solves $e^{-t}[A+B(1-e^{-t})]=\lambda_0e^{\kappa t}$; for a constant hazard ($\kappa=0$) the latter is a quadratic in $e^{-t}$, and for $\kappa>0$ it is computed numerically. \cref{fig:separable} plots both thresholds. The left panel keep sthe hazard constant and varies the risk level $\lambda_0$. Both thresholds fall as the technology becomes more dangerous, with $\underline\tau=\log(A/\lambda_0)$, and the bracket narrows from width $\log(1+B/A)$ as $\lambda_0\downarrow0$ to zero as $\lambda_0\uparrow A$. The right panel considers exponential hazard with $\lambda_0=1$, varying the growth rate $\kappa$; both thresholds fall as risk accelerates. Dots mark the pure priority equilibrium at selected parameter values.

\section{Imperfect Transparency}\label{sec:imperfect}
We now relax perfect observability. If firm $i$ stops at $\tau_i$, its rival receives a conclusive signal after an independent delay $E_i
\sim \mathrm{Exp}(\omega)$. Thus $Y_{t}^i$ jumps from zero to one at $\tau_i +E_i$.  A signal reveals that the firm has stopped but not when; silence is therefore ambiguous. The news rate $\omega$ measures transparency, with mean detection lag $1/\omega$; the perfectly observed benchmark is the limit $\omega\to\infty$.

Firm $i$'s filtration is, as before, generated by its private randomization device $\bm{Z}_i^*$\footnote{This consists of two iid uniform random variables. The first is drawn at time $0$, and the second is drawn upon \emph{detection} of the rival's stop.} and the signals it receives:
$
\mathcal F^i_t=\sigma\Big(\bm{Z}_i^*,\ (Y^j_s)_{s\le t,\,j\neq i}\Big).$ As before, a strategy specifies, at each of the firm's histories, a distribution over future stopping dates, and an outcome is a profile $\bm\tau$ in which each $\tau_i$ is $\mathcal{F}^i$-adapted. The formal construction (\cref{app:noisy}) mirrors that of the perfectly observed game with detections in place of stops: a strategy is a family of \emph{plans} indexed by the firm's \emph{signal record}---which rivals have been detected, and when---each plan a measurable map from the device to a stopping date. With two firms, a strategy therefore reduces to a pair: a \emph{concession plan}, the (possibly mixed) date at which the firm stops on its no-signal history, and a \emph{response plan}, the date at which it stops following a detection.

We study perfect Bayesian equilibria and write $\mathrm{PBE}(\omega)$ for the 
set of equilibrium outcomes. Beliefs are largely pinned down by Bayes' rule: 
on no-signal histories, the posterior over whether and when the rival has stopped 
is determined by the rival's strategy and the absence of a signal; on detection 
histories, the signal is conclusive, so firm $i$ assigns probability one to firm 
$j$ having stopped in the past. The only residual freedom is the posterior over 
the exact stopping date of a rival detected off path.  Before $\tau^R$, these beliefs may affect the optimal response because they determine the rival's frozen technology, but after it \cref{ass:dsc} implies that immediate stopping is optimal for every feasible stopping-date belief.

We henceforth maintain the following condition that guarantees that after $\tau^R$, a firm that discovers that its rival has already stopped also prefers to stop.

\begin{assumption}[Diagonal single-crossing]\label{ass:dsc}
The map $t\mapsto D_i\log\pi(t,t)-\lambda(t)$ crosses zero at most once, from above; equivalently, if $D_i\log\pi(t,t)\le\lambda(t)$ then $D_i\log\pi(s,s)\le\lambda(s)$ for every $s\ge t$.
\end{assumption}

\begin{proposition}[Existence under imperfect transparency]
\label{prop:existence}
Suppose that $\tau^R < +\infty$. Then for every news rate $\omega>0$, $\PBE(\omega)\neq\emptyset$.
\end{proposition}


Let $\alpha(s):=e^{-rs-\int_0^s\lambda(u)\de u}$ be the discount--survival kernel, so $\alpha'(s)/\alpha(s)=-(r+\lambda(s))$ and $\alpha$ is continuously differentiable, strictly positive, strictly decreasing, with $\alpha(s)\downarrow0$. 

\begin{definition}[Race and stop values]\label{def:VRVS} Fix a symmetric no-news history at which both firms are active at technology $t$. The \emph{race value} is the payoff from both racing forever,
\[
V_R(t):=\frac1{\alpha(t)}\int_t^\infty \alpha(s)\,\pi(s,s)\,\de s
=\int_t^\infty e^{-\int_t^s(r+\lambda(u))\de u}\,\pi(s,s)\,\de s.
\]
The \emph{stop value} $V_S(t,\omega)$ is the payoff from stopping at $t$ while the rival continues until it detects the stop and then chooses its optimal stopping date against a rival frozen at $t.$  With detection date $a:=t+\widetilde D$ and the rival's frozen-rival optimum $\sigma^\star(a):=\inf\{u\ge a:D_i\log\pi(u,t)\le\lambda(u)\}$,
\[
V_S(t,\omega):=\Ex_{\widetilde D}\!\left[\int_t^{\sigma^\star(a)}\frac{\alpha(s)}{\alpha(t)}\,\pi(t,s)\,\de s+\frac{\alpha(\sigma^\star(a))}{\alpha(t)}\,\frac{\pi(t,\sigma^\star(a))}{r}\right].
\]
Write $\Delta(t,\omega):=V_S(t,\omega)-V_R(t)$ for the difference between these values.
\end{definition}

The race and stop values capture stopping incentives. Unlike the case with perfect observability, a firm that contemplates stopping knows that it will only be observed at a mean lag $1/\omega$. In the interim, its technology is frozen while its rival keeps scaling and catastrophic risks continue to accrue.

\begin{figure}[H]\centering
\begin{tikzpicture}[x=1.05cm,y=0.5cm]
\fill[gray!10]  (0,0)   rectangle (3,9);
\fill[red!10] (3,0)   rectangle (6.5,9);
\fill[gray!5]  (6.5,0) rectangle (10.4,9);
\draw[-{Stealth[length=2mm]}] (0,0) -- (10.8,0) node[below=2pt] {\footnotesize time $s$};
\draw[-{Stealth[length=2mm]}] (0,0) -- (0,9.6) node[right=2pt] {\footnotesize technology};
\draw[densely dotted, gray!60!black] (6.5,6.5) -- (9.3,9.3);
\node[align=center, gray!60!black] at (7.7,9.0) {\scriptsize counterfactual:\\[-2pt]\scriptsize both race forever ($V_R$)};
\draw[very thick] (0,0) -- (6.5,6.5) -- (10.4,6.5);
\draw[very thick] (3,3) -- (10.4,3);
\filldraw (3,3) circle (2.2pt) node[above left=1pt] {\scriptsize $i$ stops at $t$};
\filldraw (6.5,6.5) circle (2.2pt);
\node[anchor=west] at (6.9,5.6) {\scriptsize $j$ detects, stops at $\sigma^\star(a)=a$};
\draw[thin,-{Stealth[length=1.6mm]}] (7.05,5.9) -- (6.62,6.35);
\node[rotate=26] at (5.1,5.85) {\scriptsize firm $j$ (still racing)};
\node at (8.6,3.5) {\scriptsize firm $i$ (frozen at $t$)};
\draw[line width=2.6pt, red!70!black] (0,-2.6) -- (6.5,-2.6);
\draw[densely dotted, line width=1pt, red!70!black] (6.5,-2.6) -- (10.4,-2.6);
\node[below, red!70!black] at (3.25,-2.7) {\scriptsize hazard $\lambda$ on: frontier still moving};
\node[below, red!70!black] at (8.45,-2.7) {\scriptsize hazard off};
\draw (3,0.14) -- (3,-0.14) node[below] {\scriptsize $t$};
\draw (6.5,0.14) -- (6.5,-0.14) node[below] {\scriptsize $a=t+X$};
\draw[decorate,decoration={brace,mirror,amplitude=4pt}] (3,-1.1) -- (6.5,-1.1) node[midway,below=4pt] {\scriptsize detection lag $X\sim\mathrm{Exp}(\omega)$};
\node[align=center] at (1.5,7.6) {\scriptsize both race:\\[-2pt]\scriptsize flow $\pi(s,s)$};
\node[align=center] at (4.75,1.5) {\scriptsize $i$ exposed:\\[-2pt]\scriptsize flow $\pi(t,s)$ falls};
\node[align=center] at (8.45,1.5) {\scriptsize safe:\\[-2pt]\scriptsize value $\pi(t,a)/r$};
\end{tikzpicture}
\caption{Stopping first under delayed detection.}
\label{fig:VSanatomy}
\end{figure}
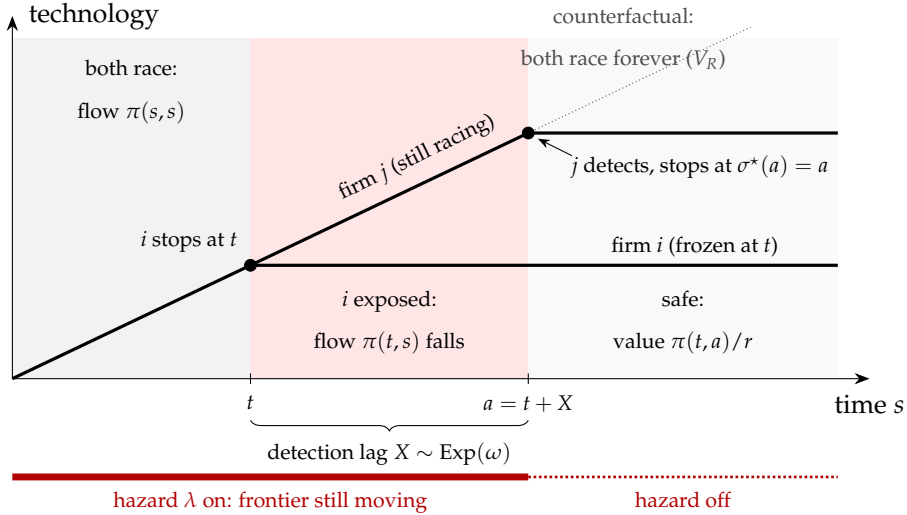
For a stop at $t\ge\tau^R$, \cref{ass:dsc} implies that the rival stops immediately upon detection ($\sigma^\star(a)=a$); \cref{fig:VSanatomy} illustrates this post-ceiling case: firm $i$ stops at $t$ and its technology freezes; firm $j$, seeing no signal, races on. During the detection lag (middle band) firm $i$ is doubly exposed: its flow $\pi(t,s)$ is eroded by the rival's still-rising technology, and the disaster hazard remains alive because the frontier is still moving. Upon detection at $a=t+X$ the rival stops, the hazard dies, and firm $i$ collects the safe perpetuity $\pi(t,a)/r$. The race value $V_R(t)$ is the counterfactual in which both firms ride the diagonal forever (dotted extension). Transparency shortens the middle band in expectation: as $\omega\to\infty$ it collapses and $V_S(t,\omega)\to\pi(t,t)/r$ (\cref{lem:mono}).

\begin{lemma}[Single-agent optimum]\label{lem:R}
Fix a rival frozen at $c$ and a detection date $a\ge c$. The value of stopping at $\sigma\ge a$,
\[
\int_a^\sigma\frac{\alpha(s)}{\alpha(a)}\pi(s,c)\,\de s+\frac{\alpha(\sigma)}{\alpha(a)}\frac{\pi(\sigma,c)}{r},
\]
has $\sigma$-derivative $\frac{\alpha(\sigma)}{\alpha(a)}\,\frac{\pi(\sigma,c)}{r}\,\big[D_i\log\pi(\sigma,c)-\lambda(\sigma)\big]$; by single-crossing (\cref{ass:payoffs}(iv)) on $[c,\infty)$ it is single-peaked, with unique maximizer the first crossing $\inf\{\sigma\ge a:D_i\log\pi(\sigma,c)\le\lambda(\sigma)\}$.
\end{lemma}

\begin{proof}
Differentiate: the flow term contributes $\frac{\alpha(\sigma)}{\alpha(a)}\pi(\sigma,c)$ and the terminal term $\frac{\alpha(\sigma)}{\alpha(a)}[-(r+\lambda(\sigma))\frac{\pi(\sigma,c)}r+\frac{D_i \pi(\sigma,c)}r]$; collecting gives the derivative. Its sign is that of $D_i\log\pi(\sigma,c)-\lambda(\sigma)$, strictly decreasing on $[c,\infty)$ by (iv), so the value is single-peaked. The argument lives entirely on $\{x_1\ge x_2=c\}$, the native domain of (iii).
\end{proof}

\begin{lemma}[Closed form past the ceiling]\label{lem:VSclosed}
For $t\ge\tau^R$ the rival stops immediately on detection ($\sigma^\star(a)=a$), so
\[
V_S(t,\omega)=\int_0^\infty\!\omega e^{-\omega x}\left[\int_t^{t+x}\frac{\alpha(s)}{\alpha(t)}\pi(t,s)\,\de s+\frac{\alpha(t+x)}{\alpha(t)}\frac{\pi(t,t+x)}{r}\right]\de x .
\]
\end{lemma}

\begin{proof}
At detection $a=t+x\ge\tau^R$, log-supermodularity ($t\le a$) and \cref{ass:dsc} give $D_i\log\pi(a,t)\le D_i\log\pi(a,a)\le\lambda(a)$, so $\sigma^\star(a)=a$ by \cref{lem:R}. 
\end{proof}

\begin{lemma}[Monotonicity and endpoints]\label{lem:mono}
For every $t\ge0$:
\begin{itemize}
\item[(a)] $\omega\mapsto V_S(t,\omega)$ is nondecreasing, and strictly increasing when $\tau^R<\infty$;
\item[(b)] $V_S(t,0):=\lim_{\omega\downarrow0}V_S(t,\omega)=\int_t^\infty\frac{\alpha(s)}{\alpha(t)}\pi(t,s)\,\de s<V_R(t)$;
\item[(c)] for $t\ge\tau^R$, $\lim_{\omega\to\infty}V_S(t,\omega)=\pi(t,t)/r\ge V_R(t)$ with strict inequality for $t>\tau^R$.
\end{itemize}
\end{lemma}

\paragraph{Cooperation threshold.} $\Delta(t,\omega)$ captures the value  of stopping now versus racing forever, conditional on no news at $t$.  Define
\[
\bar\omega:=\sup\Big\{\omega\ge0:\ \Delta(t,\omega)\le0\ \text{for all }t\ge0\Big\}.
\]
Because $\Delta(t,\omega)$ is nondecreasing in $\omega$ for every $t$ (\cref{lem:mono}(a)), the set inside the supremum is downward closed. Hence $\omega<\bar\omega$ implies $\Delta(t,\omega)\le0$ for every $t$, while every $\omega>\bar\omega$ admits some date $t$ with $\Delta(t,\omega)>0$. If $\tau^R<\infty$, then \cref{lem:mono}(c), applied at any $t>\tau^R$, implies $\bar\omega<\infty$.

Transparency level $\bar\omega$ is the boundary below which no date offers a profitable unilateral stop against a rival that races until detection. Hence, both firms prefer to race on; this is the sense in which observation lags hinder coordination. The following lemma gives an expression for the payoff from continuing until $\sigma \geq t$ then stopping.

\begin{lemma}[Recursive identity]\label{lem:rec}
Against a rival who continues until news arrives, the payoff to ``continue to $\sigma\ge t$, then stop'' is
\begin{align*}
    W_t(\sigma)&=\int_t^\sigma\frac{\alpha(s)}{\alpha(t)}\pi(s,s)\,\de s+\frac{\alpha(\sigma)}{\alpha(t)}V_S(\sigma,\omega) \\
&=V_R(t)+\frac{\alpha(\sigma)}{\alpha(t)}\,\Delta(\sigma,\omega)
\end{align*}
and $W_t(\infty)=V_R(t)$.
\end{lemma}

\begin{proof}
On $[t,\sigma)$ both race (flow $\pi(s,s)$); after $\sigma$ the continuation is $V_S(\sigma,\omega)$ by definition. Subtract the splitting identity $V_R(t)=\int_t^\sigma\frac{\alpha(s)}{\alpha(t)}\pi(s,s)\de s+\frac{\alpha(\sigma)}{\alpha(t)}V_R(\sigma)$. As $\sigma\to\infty$, $\int_t^\sigma\frac{\alpha(s)}{\alpha(t)}\pi(s,s)\de s\to V_R(t)$.
It remains to show the tail term vanishes: $\frac{\alpha(\sigma)}{\alpha(t)}V_S(\sigma,\omega)\to0$.
By Lemma~\ref{lem:mono}, $0\le V_S(\sigma,\omega)\le \pi(\sigma,\sigma)/r$ for
$\sigma\ge\tau^R$, so it suffices that $\alpha(\sigma)\pi(\sigma,\sigma)\to 0$. For $\sigma\ge\tau^R$,
\[
  \frac{d}{d\sigma}\log\pi(\sigma,\sigma)
    =D_i\log\pi(\sigma,\sigma)+D_j\log\pi(\sigma,\sigma)
    \le\lambda(\sigma),
\]
using $D_i\log\pi(\sigma,\sigma)\le\lambda(\sigma)$ (Assumption~\ref{ass:dsc})
and $D_j\log\pi<0$. Integrating from $\tau^R$ gives
$
  \pi(\sigma,\sigma)\le\pi(\tau^R,\tau^R)\,
    \exp\!\Big(\textstyle\int_{\tau^R}^{\sigma}\lambda(u)\,du\Big);$
since $\alpha(\sigma)=e^{-r\sigma-\int_0^\sigma\lambda(u)\,du}$,
\[
  \alpha(\sigma)\pi(\sigma,\sigma)
    \le \pi(\tau^R,\tau^R)\,e^{-\int_0^{\tau^R}\lambda(u)\,du}\,e^{-r\sigma}\to0 .
\] Hence $W_t(\infty)=V_R(t)$. 
\end{proof}

\begin{proposition}[Catastrophe under noisy monitoring]\label{prop:never-stop}
Suppose $\omega<\bar\omega$. Then there exists $\bm{\tau} \in \PBE(\omega)$ such that $\tau_1=\tau_2=+\infty$ almost surely. 
\end{proposition}

\cref{prop:never-stop} states that imperfect observation can generate equilibria in which both firms scale forever even when no such equilibrium exists under perfect transparency.\footnote{Note, however, that this is a statement about the ``worst'' equilibrium; others in which both firms stop in finite time still survive.}

For $\omega>\bar\omega$, some date offers a strictly profitable unilateral stop against a rival that races until detection. Even so, each firm would rather that their rival be the one to first stop, so the game after $\tau^R$ resembles a war of attrition. It remains to quantify how much equilibrium delay pushes the frontier far beyond $\tau^R$. The force that limits this delay is simple. Once both firms are near $\tau^R$, the private gain from being the frontier firm is small when detection is fast, while the joint loss from letting both firms keep racing is first order.

To measure this joint loss, fix a no-signal history at which both firms have continued until $t$, and ask what payoff a firm would get if the two firms could commit to a common cap $\sigma\ge t$: both race until $\sigma$ and then both stop. This benchmark is
\[
g_t(\sigma)
:=
\int_t^\sigma\frac{\alpha(s)}{\alpha(t)}\pi(s,s)\,\de s
+
\frac{\alpha(\sigma)}{\alpha(t)}\frac{\pi(\sigma,\sigma)}{r}
\tag{JOINT} \label{eq:joint_stop}
\]
$g_t(\sigma)$ is the value of a coordinated stop at $\sigma$ when both firms are active at time $t$.  Differentiating gives
\[
-g_t'(\sigma)
=
\frac{\alpha(\sigma)}{\alpha(t)}
\frac{\pi(\sigma,\sigma)}{r}
\Big[
\lambda(\sigma)
-D_i\log\pi(\sigma,\sigma)
-D_j\log\pi(\sigma,\sigma)
\Big].
\]
At $\sigma=\tau^R$, the private stopping condition gives
$\lambda(\tau^R)=D_i\log\pi(\tau^R,\tau^R)$ whenever $0 < \tau^R < +\infty$, so the term in the bracket above reduces to
$-D_j\log\pi(\tau^R,\tau^R)>0$.\footnote{When $\tau^R = 0$, the condition gives $\lambda(0)\ge D_i\log\pi(0,0)$ so the term in the bracket is at least $-D_j\log\pi(\tau^R,\tau^R)>0$ may strictly prefer to stop immediately and coordinated delay once again burns payoff.}  The representative firm is locally indifferent at $\tau^R$, but  when both firms move together, each firm imposes the rival externality on the other. Hence coordinated delay past $\tau^R$ burns payoff at a rate bounded away from zero.

Fast detection makes the prize from holding out small. If one firm concedes, the other learns this after an exponential delay with mean $1/\omega$, so the winner's extra frontier advantage is only $O(1/\omega)$; equivalently, a firm can secure $\pi(t,t)/r-O(1/\omega)$ by stopping at once and it cannot be optimal to hold out for an $O(1/\omega)$ prize while the common frontier loses value at a first-order rate. We will show that this implies that  the frontier can overshoot $\tau^R$ by at most order $1/\omega$. Before we show this, ruling out an infinite frontier requires the profitable stop to survive into the tail. This is guaranteed by the following condition on primitives.

\begin{definition}[Tail condition]
\label{cond:uniform-tail}
Suppose $\tau^R<\infty$, and define 
\[
\kappa_0
:=
\inf_{t\geq\tau^R}\big[-D_j\log\pi(t,t)\big],
\qquad
\Lambda_\infty
:=
\sup_{t\geq\tau^R}\ \sup_{u\geq t}
\left\{\lambda(u)+\left|D_j\log\pi(t,u)\right|\right\}.
\]

If $\kappa_0>0$ and
$\Lambda_\infty<\infty$ we say that the \emph{tail condition} is fulfilled.
\end{definition}
The lower bound $\kappa_0>0$ says that the competitive loss from a racing rival
does not disappear in the tail. The upper bound $\Lambda_\infty<\infty$ limits both
the hazard and the erosion suffered while a stopper waits to be detected. Both conditions are satisfied by the exponential example below. 
Under the tail condition, define
\[
\omega^{\mathrm{cap}}
:=
\frac{e\Lambda_\infty(r+\Lambda_\infty)}{\kappa_0} \quad \text{noting that} \quad \omega^{\mathrm{cap}} \geq \bar \omega.
\]

\begin{proposition}[Fast monitoring uniformly caps the frontier]
\label{prop:cap}
Suppose the tail condition holds. If
$\omega>\omega^{\mathrm{cap}}$, then
\[
\sup_{\bm\tau\in\PBE(\omega)}
\Ex\Big[(\max_i\tau_i-\tau^R)^+\Big]
\leq
\frac1\omega
+
\frac{e\Lambda_\infty}
{\kappa_0\big(\omega-\omega^{\mathrm{cap}}\big)}.
\]
In particular, every equilibrium stops in finite time almost surely, and $(\max_i \tau_i - \tau^R)^+ \to 0$ in probability as $\omega \to \infty$.  
\end{proposition}

\begin{proof}[Proof of \cref{prop:cap}]
Fix an equilibrium and define 
\begin{align*}
P(t):=\frac{\pi(t,t)}r \quad \text{and} \quad 
\phi(t):=\lambda(t)-D_i\log\pi(t,t)-D_j\log\pi(t,t).
\end{align*}
Here $P(t)$ is each firm's payoff if both firms stop at $t$, while $\phi(t)$
is the net instantaneous cost of delaying both firms: the hazard rate less growth rate of flow profits evaluated when firms are neck-and-neck. 

We first obtain a uniform bound on delay over a short interval and then
iterate it. By \cref{ass:dsc} and the tail condition,  $
\kappa_0\leq\phi(t)\leq\Lambda_\infty
\qquad(t\geq\tau^R),$ where the lower bound uses
$\lambda-D_i\log\pi\geq0$ and
$-D_j\log\pi\geq\kappa_0$ and the upper bound uses
$D_i\log\pi>0$ and
$\lambda+|D_j\log\pi|\leq\Lambda_\infty$. 
For $t\geq\tau^R$, define the loss from a detection lag of duration $x$ by
\[
\psi_t(x)
:=
P(t)
-
\left[
\int_t^{t+x}\frac{\alpha(s)}{\alpha(t)}\pi(t,s)\,\de s
+
\frac{\alpha(t+x)}{\alpha(t)}\frac{\pi(t,t+x)}r
\right].
\]
By \cref{lem:VSclosed},
$
P(t)-V_S(t,\omega)
=
\int_0^\infty\omega e^{-\omega x}\psi_t(x)\,\de x.$ Moreover, $\psi_t(0)=0$ and
\[
0\leq\psi_t'(x)
=
\frac{\alpha(t+x)}{\alpha(t)}
\frac{\pi(t,t+x)}r
\big[\lambda(t+x)-D_j\log\pi(t,t+x)\big]
\leq P(t)\Lambda_\infty.
\]
Hence $\psi_t(x)\leq P(t)\Lambda_\infty x$, and therefore
\[
V_S(t,\omega)
\geq
P(t)\left(1-\frac{\Lambda_\infty}{\omega}\right).
\tag{STOP}
\]

 Using the coordinated-stop value $g_t$ from \eqref{eq:joint_stop}, for
$0\leq x\leq1/(r+\Lambda_\infty)$ we have 
\begin{align*}
-\frac{g_t'(t+x)}{P(t)}
&=
e^{-rx-\int_t^{t+x}\phi(v)\de v}\phi(t+x) \geq
\kappa_0 e^{-(r+\Lambda_\infty)x}
\geq 
\frac{\kappa_0}{e}.
\end{align*}
It follows that, for every $s\geq t\geq\tau^R$,
\[
g_t(s)
\leq
P(t)\left[
1-\frac{\kappa_0}{e}\big((s-t)\wedge h\big)
\right].
\tag{DELAY}
\]
Taking $s\to\infty$ in \textup{(DELAY)} and using
$g_t(\infty)=V_R(t)$ gives
\[
V_R(t)
\leq
P(t)\left(1-\frac{\kappa_0}{e(r+\Lambda_\infty)}\right).
\]
Together with \textup{(STOP)}, this implies
$V_S(t,\omega)>V_R(t)$ for every $t\geq\tau^R$ whenever
$\omega>\omega^{\mathrm{cap}}$. Hence
$\bar\omega\leq\omega^{\mathrm{cap}}$.

We next use equilibrium optimality. Fix $t\geq\tau^R$ with
$\Pr(\min_i\tau_i\geq t)>0$. At an active no-news history, a firm may be uncertain
whether its rival is still active or has stopped without being detected. Under the latter event, stopping immediately is weakly optimal by
\cref{ass:dsc}, log-supermodularity, and \cref{lem:R}. This thus cancels in ths comparison between equilibrium continuation and the deviation to stop immediately. On the former event, stopping immediately yields at
least $V_S(t,\omega)$. This is because the  definition of $V_S$ assumes that the rival continues until it detects the stop which is the worst case scenario---the rival might stopper sooner. 

We next bound equilibrium continuation conditional on the rival still being
active. Both firms earn the diagonal flow until the first stop. At that date,
the continuation payoff of either firm is at most
$P(\min_i\tau_i)$.\footnote{For the leader that stopped, her best case is for the follower to also stop which delivers $P(\min_i \tau_i)$; for the follower who knew that the leader has stopped at $\min \tau_i$, her best response is to stop immediately which also yields $P(\min_i \tau_i)$.} 
where we recall that $g_t(X)$ defined above is the `coordinated stop value' from coordinating a stop at $x$. Thus equilibrium continuation is at most
\[
\Ex\!\left[
g_t(\min_i\tau_i)
\ \middle|\
\min_i\tau_i\geq t
\right].
\]
Combining this bound with 
\textup{(STOP)} and
\textup{(DELAY)} we have: 
\begin{align*}
P(t)\left(1-\frac{\Lambda_\infty}{\omega}\right)
&\leq V_S(t,\omega)\\
&\leq
\Ex\!\left[
g_t(\min_i\tau_i)
\ \middle|\
\min_i\tau_i\geq t
\right]\\
&\leq
P(t)\left\{
1-\frac{\kappa_0}{e}
\Ex\left[
\big((\min_i\tau_i-t)^+\big)
\wedge\frac1{r+\Lambda_\infty}
\ \middle|\
\min_i\tau_i\geq t
\right]
\right\}.
\end{align*}
Canceling $P(t)$ and rearranging yields
\[
\Ex\left[
\big((\min_i\tau_i-t)^+\big)
\wedge\frac1{r+\Lambda_\infty}
\ \middle|\
\min_i\tau_i\geq t
\right]
\leq
\frac{e\Lambda_\infty}{\kappa_0\omega} = \frac{\omega^{\mathrm{cap}}}
{\omega(r+\Lambda_\infty)}.
\tag{BLOCK}
\]
where the equality is just from the definition of $\omega^{\mathrm{cap}}$. Because the random variable in \textup{(BLOCK)} equals
$1/(r+\Lambda_\infty)$ on
$\{\min_i\tau_i\geq t+1/(r+\Lambda_\infty)\}$,
\[
\Pr\left(
\min_i\tau_i\geq t+\frac1{r+\Lambda_\infty}
\ \middle|\
\min_i\tau_i\geq t
\right)
\leq \frac{\omega^{\mathrm{cap}}}{\omega}<1.
\]
Applying this inequality successively at
$t=\tau^R,\tau^R+1/(r+\Lambda_\infty),\ldots,
\tau^R+(k-1)/(r+\Lambda_\infty)$ yields
\[
\Pr\left(
\min_i\tau_i\geq\tau^R+\frac{k}{r+\Lambda_\infty}
\ \middle|\
\min_i\tau_i\geq\tau^R
\right)
\leq
\left(\frac{\omega^{\mathrm{cap}}}{\omega}\right)^k.
\]
Applying \textup{(BLOCK)} at
$t=\tau^R+k/(r+\Lambda_\infty)$ and using the preceding bound gives
\[
\Ex\left[
\left(
\min_i\tau_i-\tau^R-\frac{k}{r+\Lambda_\infty}
\right)^+
\wedge\frac1{r+\Lambda_\infty}
\right]
\leq
\frac1{r+\Lambda_\infty}
\left(\frac{\omega^{\mathrm{cap}}}{\omega}\right)^{k+1}.
\]
Observe that we have 
\[
\big(\min_i\tau_i-\tau^R\big)^+
=
\sum_{k=0}^\infty
\left[
\left(
\min_i\tau_i-\tau^R-\frac{k}{r+\Lambda_\infty}
\right)^+
\wedge\frac1{r+\Lambda_\infty}
\right].
\]
This partitions the positive overshoot into consecutive intervals
of length $1/(r+\Lambda_\infty)$. If the overshoot lies between
$k/(r+\Lambda_\infty)$ and $(k+1)/(r+\Lambda_\infty)$, the first $k$
terms each equal $1/(r+\Lambda_\infty)$, the next term equals the remaining
overshoot, and every later term is zero.
Monotone convergence thuen gives
\[
\Ex\Big[\big(\min_i\tau_i-\tau^R\big)^+\Big]
\leq
\frac{\omega^{\mathrm{cap}}/\omega}
{(r+\Lambda_\infty)
\big(1-\omega^{\mathrm{cap}}/\omega\big)}
=
\frac{e\Lambda_\infty}
{\kappa_0\big(\omega-\omega^{\mathrm{cap}}\big)}.
\]
Finally, after the first stop at $\min_i\tau_i$, the other firm detects it
after $\widetilde D\sim\mathrm{Exp}(\omega)$. By \cref{ass:dsc},
log-supermodularity, and \cref{lem:R}, it stops no later than
$\max\{\tau^R,\min_i\tau_i+\widetilde D\}$. Thus, pathwise,
\[
\big(\max_i\tau_i-\tau^R\big)^+
\leq
\big(\min_i\tau_i-\tau^R\big)^+
+\widetilde D.
\]
Taking expectations and using the preceding bound proves the result. The bound
is finite and uniform across equilibria. Its right-hand side converges to zero
as $\omega\to\infty$, which gives the uniform $L^1$ convergence and convergence in
probability follows from Markov's inequality.
\end{proof}

\cref{prop:cap} shows how cooperation partially succeeds when news arrives sufficiently quickly ($\omega > \omega^{\mathrm{cap}}$). If so, then for the equilibrium in which the technology is pushed past $\tau^R$,  outlasting one's rival is worth only the extra flow earned during the rival's detection lag (of order $1/\omega$) while joint delay burns value at a rate bounded away from zero. Since it cannot be rational to hold out for a long time to gain a vanishing prize, the expected delay until concession is  order $1/\omega$, so the frontier overshoots $\tau^R$ by $O(1/\omega)$ in expectation where this bound is uniform across all equilibria.\footnote{The proof gives a geometric block-tail bound and an $O(1/\omega)$ bound on the expected frontier overshoot which does not require equilibrium stopping distributions to have densities. In the smooth symmetric equilibrium considered below, the short delay is implemented by a concession hazard proportional to $\omega$.} As $\omega\to\infty$, the overshoot vanishes in probability so the perfect-observability ceiling is recovered in the limit.  

\cref{prop:cap} is stated for the monitoring regime $\omega > \omega^{\mathrm{cap}} \geq \bar \omega$ in order to discipline stopping incentives in the tail. However, for stationary environments such as those analyzed in the exponential below as well as in \cref{sec:trust}, the threshold is tight: when $\omega < \bar \omega$ there exists an equilbirium in which firms race forever; when $\omega > \bar \omega$ in every equlibria $\max_i \bm{\tau} < +\infty$ almost surely; we show this in \cref{app:further}.

\begin{example}[Exponential race with imperfect transparency]\label{ex:exponential}
Let flow payoffs be exponential,
$
\pi(t_i,t_j)=e^{\alpha t_i-\beta t_j},\alpha,\beta>0$. 
Both indices of the previous section are then constant and equal: $D_i\log\pi(t,t)=D_i\log\pi(t,0)=\alpha$ and $D_{ij}\log\pi=0$, so with an increasing hazard $\lambda(t)=\lambda_0e^{\kappa t}$ ($\kappa>0$, $\lambda_0<\alpha$) the bracket is degenerate:
$
\underline\tau=\tau^R=\frac1\kappa\log\frac{\alpha}{\lambda_0}
$, 
so under perfect transparency, firms stop at $\tau^R$ in every equilibrium. This makes the exponential race a natural setting to isolate the effect of imperfect transparency: any departure from $\tau^R$ is attributable to monitoring frictions alone. We take the hazard constant at a level exceeding the index: $\lambda(t)\equiv\lambda>\alpha$, with $r>\alpha-\beta$, so that $\tau^R=0$.\footnote{This is the boundary case of \cref{ass:payoffs}(iv): the index $\alpha$ never crosses the hazard from above because it starts below it, and \cref{ass:dsc} holds trivially. Interpret $(\pi,\lambda)$ as the continuation environment once the race has reached the threshold; the next section adopts the same specification.} The environment is then \emph{stationary}: at a symmetric no-signal history at technology $t$, every continuation payoff scales with the current flow $\pi(t,t)=e^{(\alpha-\beta)t}$, so values per unit of current flow are date-independent, and $\Delta(t,\omega)$ has the same sign at every date. Direct computation gives the normalized  values
\[
V_R=\frac{1}{r+\lambda-\alpha+\beta},
\qquad
V_S(0)=\frac{1}{r+\lambda+\beta},
\qquad
V_S(\omega)=V_S(0)+\frac{\omega\big[1/r-V_S(0)\big]}{\omega+r+\lambda+\beta}\,:
\]
racing forever discounts at $r+\lambda$ while the diagonal flow grows at $\alpha-\beta$; stopping first freezes one's own flow, which then decays at the erosion rate $\beta$ under a still-live hazard until the detection date, whose expected discount factor $\omega/(\omega+r+\lambda+\beta)$ weights the upgrade to the joint-stop perpetuity $1/r$. Setting $V_S(\omega)=V_R$ and solving, $
\bar\omega=\frac{r\alpha}{\lambda-\alpha+\beta}\,.$
Every parameter moves $\bar\omega$ in the intuitive direction. The transparency needed for cooperation rises with the scaling gain $\alpha$, which pays only while racing, and with impatience $r$, which discounts the reciprocation upgrade. It falls with the hazard $\lambda$ and the erosion $\beta$: both depress the race value faster than the stop value (the race lasts forever; the exposure window ends at detection), so more dangerous technologies and stronger competitive erosion make cooperation easier to sustain.

By stationarity, $\Delta(\cdot,\omega)$ has the same sign at every
date. Hence, when $\omega>\bar\omega$ so that $\Delta>0$ everywhere,
a unilateral stop is profitable at every date and
\cref{prop:exact-threshold} in \cref{app:further} shows that every
equilibrium stops in finite time. \cref{fig:exp-example} displays the  threshold at which $V_S$ crosses $V_R$, and the vanishing prize. Panel (a) illustrates how the stop value $V_S(\omega)$ rises from $V_S(0)=1/(r+\lambda+\beta)$ toward $1/r$ as news speeds up, crossing the race value $V_R=1/(r+\lambda-\alpha+\beta)$ at $\bar\omega=r\alpha/(\lambda-\alpha+\beta)=1/15$: below the threshold racing forever is self-enforcing (\cref{prop:never-stop}), and above it all equilibria are finite, conceding quickly in expectation when transparency is high (\cref{prop:cap}).

\end{example}

\begin{figure}[H]\centering
\begin{minipage}[t]{0.52\textwidth}\centering
\begin{tikzpicture}[x=36cm,y=2.1cm]
\draw[-{Stealth[length=2mm]}] (0,0) -- (0.218,0) node[below left=1pt] {\scriptsize $\omega$};
\draw[-{Stealth[length=2mm]}] (0,0) -- (0,1.78);
\foreach \o/\lab in {0.05/{0.05},0.1/{0.1},0.15/{0.15},0.2/{0.2}}{\draw (\o,0.015) -- (\o,-0.015) node[below] {\tiny \lab};}
\foreach \v/\lab in {0.5/{1.0},1.0/{1.5},1.5/{2.0}}{\draw (0.002,\v) -- (-0.002,\v) node[left] {\tiny \lab};}
\draw[thick, red!70!black] (0,0.6765) -- (0.21,0.6765);
\node[left=1pt, red!70!black] at (0,0.6765) {\tiny $V_R$};
\node[left=1pt] at (0,0.2407) {\tiny $V_S(0)$};
\draw[very thick, blue!60!black] plot[domain=0:0.21,samples=80] (\x,{0.7407+9.2593*\x/(\x+1.35)-0.5});
\node[above left=0pt, blue!60!black] at (0.212,1.44) {\tiny $V_S(\omega)$};
\node[above, blue!60!black] at (0.155,1.5) {\tiny $V_S(\omega)\uparrow 1/r=10$};
\draw[densely dashed] (0.0667,0) -- (0.0667,0.6765);
\filldraw (0.0667,0.6765) circle (1.6pt);
\node[below] at (0.0667,-0.06) {\tiny $\bar\omega$};
\node[align=center] at (0.031,1.3) {\tiny $\Delta<0$: racing to ruin\\[-3pt]\tiny self-enforcing};
\node[align=center] at (0.15,0.38) {\tiny $\Delta>0$: profitable stop exists};
\end{tikzpicture}\\[2pt]
{\scriptsize (a) the cooperation threshold $\bar\omega$}
\end{minipage}\hfill
\begin{minipage}[t]{0.44\textwidth}\centering
\begin{tikzpicture}[x=2.9cm,y=1.85cm]
\draw[-{Stealth[length=2mm]}] (-0.06,-1.32) -- (2.18,-1.32) node[below left=1pt] {\scriptsize $\omega$};
\draw[-{Stealth[length=2mm]}] (-0.06,-1.32) -- (-0.06,0.62);
\foreach \o/\lab in {0/{1},1/{10},2/{100}}{\draw (\o,-1.305) -- (\o,-1.335) node[below] {\tiny \lab};}
\foreach \v/\lab in {0/{1},-1/{0.1}}{\draw (-0.045,\v) -- (-0.075,\v) node[left] {\tiny \lab};}
\draw[very thick, blue!60!black] plot[smooth] coordinates {(0,0.341)(0.301,0.257)(0.699,0.032)(1,-0.201)(1.301,-0.465)(1.699,-0.840)(2,-1.133)};
\node[above right=0pt, blue!60!black] at (0.05,0.36) {\tiny $V_W(\omega)-V_S(\omega)$};
\draw[densely dashed, gray!60!black] (0.55,0.325) -- (2.08,-1.205);
\node[rotate=-28, above=1pt, gray!60!black] at (1.5,-0.55) {\tiny slope $-1$: $\frac{\alpha+\beta}{r\omega}$};
\end{tikzpicture}\\[2pt]
{\scriptsize (b) the prize of outlasting (log--log)}
\end{minipage}
\caption{Exponential under imperfect monitoring ($r=0.1$, $\lambda=1$, $\alpha=0.5$, $\beta=0.25$, values per unit of current flow)}
\label{fig:exp-example}
\end{figure}
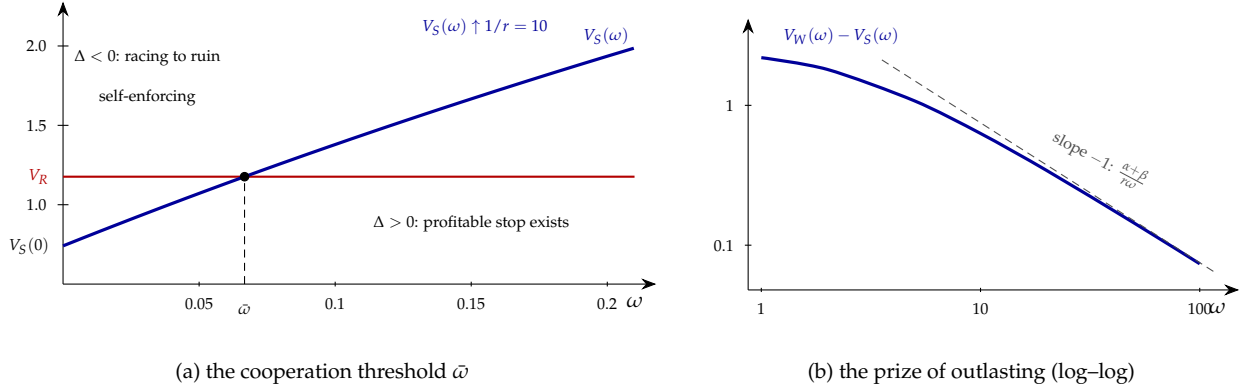

To see why concession speeds up as monitoring becomes more precise, note  that the  (normalized\footnote{We normalize this by the flow profits $\pi(t,t)$.}) value upon detection that one's rival has stopped instead of first is: 
$
V_W(\omega) - V_S(\omega) \simeq \cfrac{\alpha + \beta}{r\omega}$,
where $V_W(\omega)$ is the expected payoff from racing against an opponent who has stopped and stopping only upon detection. This is illustrated by panel (b). By contrast, joint delay burns value at the (normalized) constant rate $(\lambda-\alpha+\beta)/r$ per unit of time. Hence,  equilibrium indifference equires concession at a hazard of order $\omega$, leaving expected delay of order
$1/\omega$.

\section{Trust} \label{sec:trust}
We are interested in firms' ability to coordinate when they fear their opponent might irrationally never stop. Continue with the stationary exponential environment of Example~\ref{ex:exponential}, imposing in addition $0<\beta<\alpha$ and $\omega\ge\beta$. Thus
$
\pi(t_i,t_j)=e^{\alpha t_i-\beta t_j},
\lambda>\alpha,$ and $
r>\alpha-\beta,$ 
and the representative threshold is $\tau^R=0$. We now suppose that each firm is independently either a \emph{rational} type with these payoffs or a \emph{crazy} type that never stops. Each firm is rational with probability $p$, and we write $\ell:=p/(1-p)$ for the prior odds of rationality.

As in Example~\ref{ex:exponential}, we normalize continuation payoffs by current flow and write
\[
V_R=\frac{1}{r+\lambda-\alpha+\beta},\text{  }
\text{  }
V_S(0)=\frac{1}{r+\lambda+\beta} \text{  and }
V_S(\omega)=
V_S(0)+
\frac{\omega\bigl[1/r-V_S(0)\bigr]}
{\omega+r+\lambda+\beta}.
\]
Recall that $V_S(\omega)>V_R$ if and only if
$
\omega>\bar\omega
:=
\frac{r\alpha}{\lambda-\alpha+\beta}.$

We additionally define the \emph{verification value}
\[
V_W(\omega):=\frac{r+\omega}{r\,(\omega+r+\lambda-\alpha)}
\]
as the normalized payoff from racing beside a rival that has already stopped at the same technology, and exiting upon detection; this is the value of waiting for confirmation that a firm's rival is not crazy. Note $V_R<V_W(\omega)<1/r$: verifying beats racing forever but falls short of joining, by the \emph{lag cost} $1/r-V_W(\omega)$ .

\paragraph{Two kinds of coordination.} We now identify two distinct coordination modes and the conditions under which each is incentive-compatible. For an agent to stop first i.e., without knowing if their rival has stopped, she gambles on both their rival's type and on news arriving quickly: if their rival is rational, it stops upon receiving the news of their stop, and never stops otherwise:
\[
\frac{\ell}{1+\ell}\,V_S(\omega)+\frac{1}{1+\ell}\,V_S(0)\ \ge\ V_R
\quad\Longleftrightarrow\quad
\ell\ \ge\ \bar\ell:= \begin{cases}
    \cfrac{V_R-V_S(0)}{V_S(\omega)-V_R} \quad &\text{if $\omega > \bar \omega$} \\ 
    +\infty &\text{otherwise}
\end{cases}
\]

Alternatively, rational firms might stop at a common date. In this case, we have a new deviation that was not present when firms were commonly known to be rational---firms are tempted to \emph{wait and verify}: hold out for news that their rival has, in fact, stopped. If the rival is rational, this delivers $V_W(\omega)$ instead of $1/r$. If the rival is crazy, then the deviator obtains $V_R$ (rival never stops) instead of $V_S(0)$. Hence, stopping at a known date is optimal if and only if
\[
\ell\,\big[1/r-V_W(\omega)\big]\ \ge\ V_R-V_S(0)
\quad\Longleftrightarrow\quad
\ell\ \ge\ \underline\ell:=\frac{V_R-V_S(0)}{1/r-V_W(\omega)}.
\]
 We assume that $\underline\ell<\bar\ell$ which implies that simultaneous coordination is less demanding than sequential coordination.\footnote{Each of $V_R$, $V_S(0)$, $V_S(\omega)$, $V_W(\omega)$, and $1/r$ is the payoff of a fixed course of play and does not depend on equilibrium except for the fact that a sole survivor who learns their rival has stopped also stops immediately, which is implied by dominance ($\alpha<\lambda$). If $\underline \ell > \bar \ell$,  the middle zone disappears. Never stopping remains equilibrium for $\ell \leq \bar \ell$ and is ruled out for $\ell > \bar \ell$ when stopping first is strictly profitable.} Because$\omega \geq \beta$, waiting indefinitely for verification that their rival has stopped is the binding deviation. The two thresholds have the same  numerator $V_R - V_S(0)$, the cost of stopping 
unreciprocated, and differ in the denominator: the \emph{prize of 
reciprocation} $V_S(\omega)-V_R$ for sequential stopping and the \emph{lag 
cost} $1/r - V_W(\omega)$ for simultaneous stopping.\cref{fig:trust-values}(a) 
shows $V_S(\omega)$ and $V_W(\omega)$ rising toward $1/r$: the prize of 
reciprocation grows while the lag cost shrinks, so faster news makes both 
stopping first and waiting to verify more attractive. Panel (b) plots the 
trust-weighted values of the four candidate strategies at $\omega=0.5$. The lighter 
curves repeat the exercise at $\omega'=0.8$. Racing and stopping together are 
unaffected by $\omega$; only stopping first and waiting to verify shift. Faster 
news pulls $\bar\ell$ down and pushes $\underline\ell$ up, squeezing the medium 
zone until it closes at $\omega^{\dagger}$.

\begin{figure}[H]\centering
\begin{minipage}[t]{0.32\textwidth}\centering
\resizebox{\textwidth}{!}{%
\begin{tikzpicture}[x=1.15cm,y=0.5cm]
\draw[-{Stealth[length=2mm]}] (0,0) -- (3.3,0) node[below left=1pt] {\scriptsize $\omega$};
\draw[-{Stealth[length=2mm]}] (0,0) -- (0,11.4);
\foreach \o in {1,2}{\draw (\o,0.2) -- (\o,-0.2) node[below] {\tiny \o};}
\draw[densely dashed, gray!60!black] (0,10) -- (3.15,10);
\node[left] at (0,10) {\tiny $1/r$};
\draw[densely dashed, red!70!black] (0,1.1765) -- (3.15,1.1765);
\node[right, red!70!black] at (3.15,1.1765) {\tiny $V_R$};
\node[left] at (0,0.7407) {\tiny $V_S(0)$};
\draw[very thick, blue!60!black] plot[domain=0:3.1,samples=80] (\x,{0.7407+9.2593*\x/(\x+1.35)});
\draw[very thick] plot[domain=0:3.1,samples=80] (\x,{(0.1+\x)/(0.1*(\x+0.6))});
\node[right, blue!60!black] at (3.12,7.15) {\tiny $V_S(\omega)$};
\node[right] at (3.12,8.7) {\tiny $V_W(\omega)$};
\draw[{Stealth[length=1.6mm]}-{Stealth[length=1.6mm]}, gray!60!black] (2.6,1.1765) -- (2.6,6.835);
\node[left, gray!60!black] at (2.55,4.0) {\tiny prize};
\draw[{Stealth[length=1.6mm]}-{Stealth[length=1.6mm]}, gray!60!black] (1.0,6.875) -- (1.0,10);
\node[left, gray!60!black] at (0.95,8.5) {\tiny lag cost};
\end{tikzpicture}}\\[2pt]
{\scriptsize (a) values vs.\ transparency}
\end{minipage}\hfill
\begin{minipage}[t]{0.66\textwidth}\centering
\resizebox{\textwidth}{!}{%
\begin{tikzpicture}[x=13cm,y=1.2cm]
\draw[-{Stealth[length=2mm]}] (0,0.5) -- (0.53,0.5) node[below left=1pt] {\scriptsize $\ell$};
\draw[-{Stealth[length=2mm]}] (0,0.5) -- (0,4.45);
\node[left] at (0,0.7407) {\tiny $V_S(0)$};
\draw[thick, red!70!black] (0,1.1765) -- (0.5,1.1765);
\node[right, red!70!black] at (0.505,1.1765) {\tiny race ($V_R$)};
\draw[very thick, densely dashed] plot[domain=0:0.5,samples=60] (\x,{(10*\x+0.7407)/(1+\x)});
\node[right] at (0.505,3.827) {\tiny stop together};
\draw[very thick, black!35] plot[domain=0:0.5,samples=60] (\x,{(6.4286*\x+1.1765)/(1+\x)});
\draw[very thick, blue!35] plot[domain=0:0.5,samples=60] (\x,{(4.1860*\x+0.7407)/(1+\x)});
\draw[very thick] plot[domain=0:0.5,samples=60] (\x,{(5.4545*\x+1.1765)/(1+\x)});
\draw[very thick, blue!60!black] plot[domain=0:0.5,samples=60] (\x,{(3.2432*\x+0.7407)/(1+\x)});
\draw[-{Stealth[length=1.6mm]}] (0.42,2.452) -- (0.42,2.71);
\node[right] at (0.425,2.56) {\tiny $\omega\uparrow$};
\draw[-{Stealth[length=1.6mm]}, blue!60!black] (0.46,1.54) -- (0.46,1.81);
\node[right, blue!60!black] at (0.465,1.67) {\tiny $\omega\uparrow$};
\node[right] at (0.505,2.603) {\tiny wait \& verify};
\node[right, blue!60!black] at (0.505,1.575) {\tiny stop first};
\filldraw (0.0959,1.551) circle (1.4pt);
\draw[densely dashed, gray!60!black] (0.0959,1.551) -- (0.0959,0.5);
\node[below] at (0.0959,0.47) {\tiny $\underline\ell$};
\filldraw[blue!60!black] (0.2109,1.1765) circle (1.4pt);
\draw[densely dashed, gray!60!black] (0.2109,1.1765) -- (0.2109,0.5);
\node[below, blue!60!black] at (0.2109,0.47) {\tiny $\bar\ell$};
\filldraw[black!45] (0.1220,1.7475) circle (1.4pt);
\draw[densely dashed, black!30] (0.1220,1.7475) -- (0.1220,0.5);
\node[below, black!45] at (0.1220,0.47) {\tiny $\underline\ell'$};
\filldraw[blue!35] (0.1448,1.1765) circle (1.4pt);
\draw[densely dashed, blue!25] (0.1448,1.1765) -- (0.1448,0.5);
\node[below, blue!45] at (0.1448,0.47) {\tiny $\bar\ell'$};
\draw[-{Stealth[length=1.5mm]}] (0.0989,0.66) -- (0.119,0.66);
\draw[-{Stealth[length=1.5mm]}, blue!60!black] (0.2079,0.66) -- (0.1478,0.66);
\end{tikzpicture}}\\[2pt]
{\scriptsize (b) values vs.\ trust}
\end{minipage}
\caption{Transparency and trust; stationary environment: $(r,\lambda,\alpha,\beta)=(0.1,1,0.5,0.25)$).}
\label{fig:trust-values}
\end{figure}
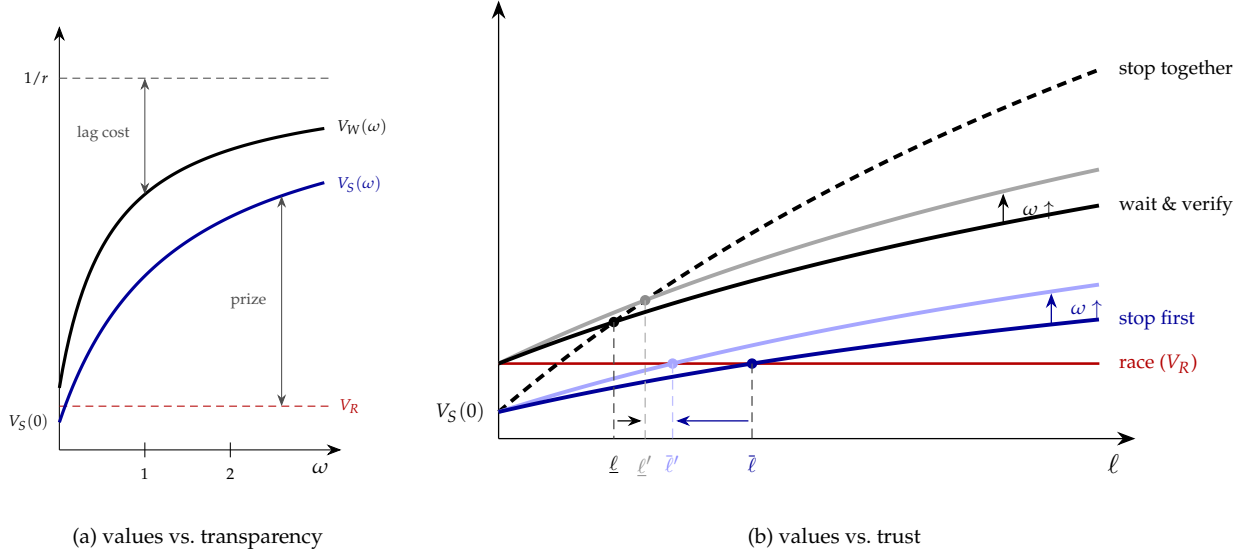

Transparency thus moves the two thresholds in opposite directions. Faster monitoring speeds reciprocation, so $\bar\ell$ falls---from $+\infty$ at $\bar\omega$ (no trust substitutes for a minimum of monitoring) toward the strictly positive limit $\frac{V_R-V_S(0)}{1/r-V_R}$---while it also cheapens verification, inviting free-riding on the signal, so $\underline\ell$ rises without bound.  This partitions the trust line into three zones (\cref{fig:LMH}): \emph{low} ($\ell<\underline\ell$), \emph{medium} ($\underline\ell\le\ell<\bar\ell$), and \emph{high} ($\ell\ge\bar\ell$).

\begin{figure}[H]
\centering
\includegraphics[width=0.75\textwidth]{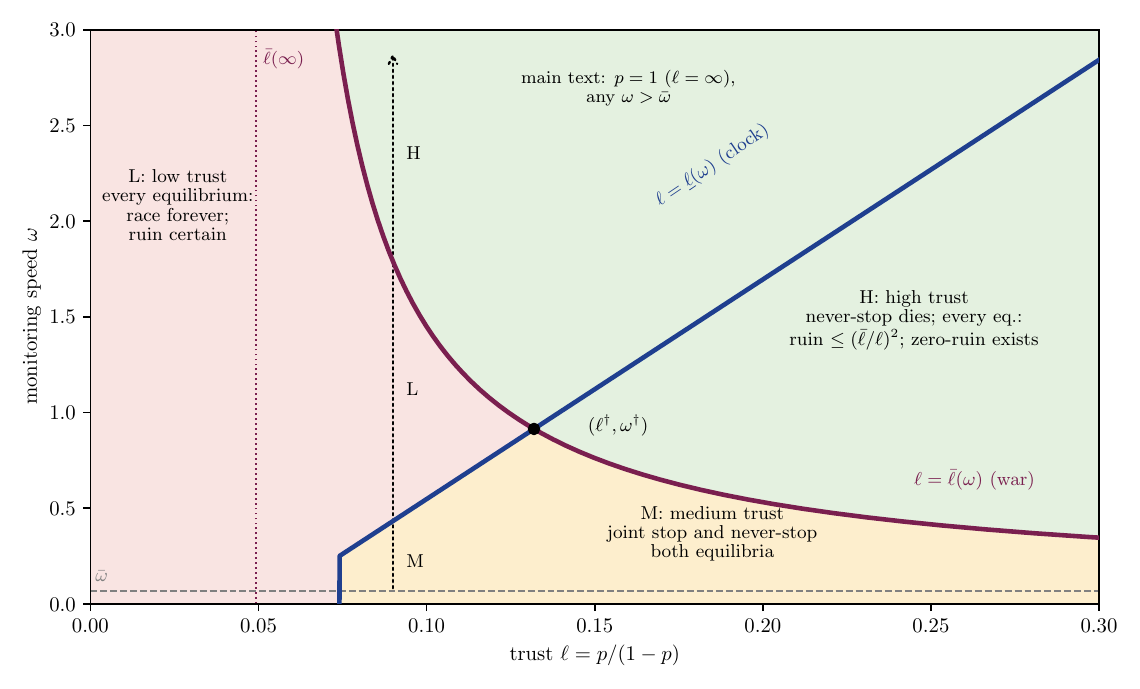}
\caption{{Trust and transparency.} Parameters: $(r,\lambda,\alpha,\beta)=(0.1,1,0.5,0.25)$.}
\label{fig:LMH}
\end{figure}

\cref{fig:LMH} illustrates the thresholds $\underline \ell$ and $\overline \ell$. The threshold $\underline\ell$ rises with $\omega$ (cheaper verification invites free-riding on the signal; below $\omega=\beta$ it is constant at its $\omega=\beta$ value---the vertical segment), while the threshold $\bar\ell$ falls (faster reciprocation), with horizontal asymptote $\bar\omega=1/15$ and vertical asymptote $\bar\ell(\infty)\approx0.049$: the high-trust zone lies entirely above the main text's cooperation threshold, and no transparency pushes its boundary below $\bar\ell(\infty)$. The curves cross once, at $(\ell^\dagger,\omega^\dagger)\approx(0.132,\,0.914)$. The dotted vertical, at trust $\ell=0.09$, shows the double-edged effect of transparency: raising $\omega$ first destroys simultaneous coordination (M$\,\to\,$L at $\omega\approx0.43$) and only later makes sequential coordination viable (L$\,\to\,$H at $\omega\approx1.79$). 

Why do trust and transparency interact so differently across the two modes of coordination? Sequential coordination asks a firm to stop first, gambling that a rational rival will reciprocate once the news lands. Hence, faster news raises the prize of reciprocation $V_S(\omega)-V_R$, so transparency helps. Conversely, simultaneous coordination requires that a firm not be tempted to keep racing, and stop only after seeing that the rival really did stop. But faster news shrinks the lag cost $1/r-V_W(\omega)$ of that deviation, so transparency hurts. This double-edged effect of transparency is visible in \cref{fig:LMH}: at intermediate levels of trust (e.g., $\ell = 0.10$), raising $\omega$ can destroy the less demanding simultaneous mode $(M \to L)$ long before it makes the sequential mode viable $(L \to H)$.

Extending the notation of \cref{sec:imperfect}, write $\PBE(\omega,\ell)$ for the set of perfect Bayesian equilibrium \emph{outcomes}---profiles of stopping times $\bm\tau=(\tau_1,\tau_2)$---when the news rate is $\omega$ and the prior odds of rationality are $\ell$.

\begin{proposition} \label{prop:LMH}
The following characterizes equilibrium outcomes across trust zones: 
\begin{itemize}[leftmargin=2.4em]
\item[(i)] \emph{Low trust.} If $\ell<\underline\ell$, then every $\bm\tau\in\PBE(\omega,\ell)$ has $\tau_1=\tau_2=+\infty$ almost surely: no firm ever stops, and the disaster arrives with probability one.
\item[(ii)] \emph{Medium trust.} If $\underline\ell\le\ell\le\bar\ell$, then $\PBE(\omega,\ell)$ contains both \emph{(a)} the outcome in which rational firms stop immediately, so that the risk survives only through crazy types, and \emph{(b)} the never-stopping outcome. 
\item[(iii)] \emph{High trust.} If $\ell>\bar\ell$, then every $\bm\tau\in\PBE(\omega,\ell)$ satisfies
\[
\Pr\big(\tau_1=\tau_2=+\infty\ \big|\ \text{both firms rational}\big)\ \le\ \big(\bar\ell/\ell\big)^2,
\]
and stopping at zero remains $\PBE(\omega,\ell)$.
\end{itemize}
\end{proposition}

\begin{figure}[H]
\centering
\includegraphics[width=\textwidth]{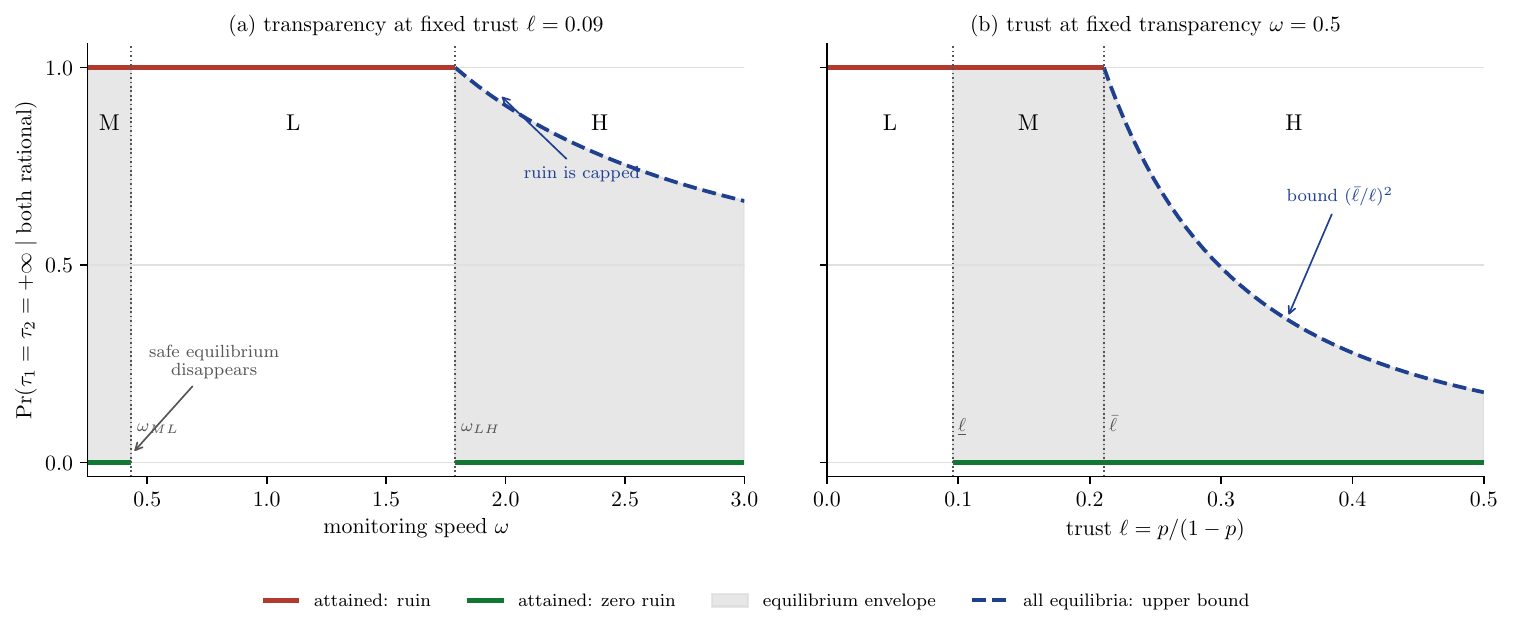}
\caption{Equilibrium probability of never stopping at different trust levels; Parameters:
\((r,\lambda,\alpha,\beta)=(0.1,1,0.5,0.25)\).}
\label{fig:trust-outcomes}
\end{figure}

 Low trust makes any unilateral stop too likely to go unreciprocated, so every equilibrium races forever. At intermediate trust, simultaneous stopping becomes sustainable, but neither firm is willing to stop first, leaving both stopping and racing equilibria. At high trust, stopping first becomes profitable whenever the rival’s holdout probability is too large, yielding the quadratic bound.

Although \cref{prop:LMH} was stated for the case in which $\underline \ell < \bar \ell$, similar results hold in the opposite case. There, Part (i) survives verbatim with $\bar\ell$ (now the smaller threshold) as the boundary of the low zone. There is no middle region. Part (iii) survives for $\ell>\bar\ell$ and the $(\bar\ell/\ell)^2$ bound holds unchanged. \cref{tab:2x2} summarizes equilibrium behavior across various combinations of trust and transparency.  

\begin{table}[H]
\centering
\footnotesize
\renewcommand{\arraystretch}{1.5}
\begin{tabular}{@{}p{1.9cm}p{5.6cm}p{5.9cm}@{}}
\toprule
& \textbf{perfect observability} ($\omega=\infty$) & \textbf{imperfect observability} ($\omega<\infty$) \\
\midrule
\textbf{no crazy types} ($p=1$) & Coordination is easy: in every equilibrium, both firms stop at once. & Single threshold $\bar\omega$: below it, never-stopping is an equilibrium alongside the joint stop; above it, every equilibrium coordinates.  \\
\addlinespace
\textbf{crazy types} ($p<1$) & Single threshold $\bar\ell(\infty)$: below it, every equilibrium races to ruin; above it, ruin $\le\big(\bar\ell(\infty)/\ell\big)^2$, with zero attained only by leader--follower outcomes & Two thresholds $\underline\ell$ rising and $\bar\ell$ falling in $\omega$, and three zones (\cref{prop:LMH}).  \\
\bottomrule
\end{tabular}
\caption{Equilibria across transparency and trust regimes}
\label{tab:2x2}
\end{table}

\section{Conclusion}
We have shown that competition and coordination play distinct roles in determining how far a dangerous technology advances. Under perfect transparency, competition places the frontier between monopoly and representative-firm thresholds. Delayed observation can sustain racing forever or generate a short war of attrition beyond the perfect-transparency ceiling. Imperfect trust is more damaging: when firms doubt their rivals' rationality, racing to ruin may occur in every equilibrium. Transparency and trust are therefore complements in some regions but substitutes in others.

\paragraph{Profile effects} We might alternatively suppose that risk also depends on the entire profile of technology states, and that when the  second firm's technology is  near the frontier it matters almost as much as the leader. \footnote{We might further suppose that risks are persistent, so that firms at frontier $T$ generate risk even after they stop scaling; the next discussion takes up persistence.} We conjecture that these profile effects can reverse a force that facilitated coordination in the baseline perfect-transparency model: A firm that stops now exerts a positive externality on the racers it leaves behind: its exit thins the frontier and lowers the hazard they face, so stopping first might embolden one's competitors to race for longer.

\paragraph{Accumulative risks} Suppose that once the technology crosses a threshold, the industry  lives with a permanent background chance of collapse: the hazard depends on the peak scale reached so far and persists after every firm has stopped.\footnote{In the language of \citet{trammell2024existential} this is \emph{state} risk rather than \emph{transition} risk.} Stopping then ends the gains from scaling but not the exposure to risk, so the value of settling at technology $T$ is the flow payoff discounted at $r+\lambda(T)$ rather than at $r$. This changes the representative firm's trade-off: the hazard level is borne whether or not the firm stops, so the cost of continuing is only the permanent increment to background risk and the stopping rule becomes\footnote{Differentiating the analog of $R_t(\tau)$ with terminal value $\pi(\tau,t)/(r+\lambda(\tau))$ gives a marginal value of delay proportional to $D_i\log\pi(\tau,t,\ldots,t)-\lambda'(\tau)/(r+\lambda(\tau))$. A formal statement and the requisite assumptions are collected in \cref{app:extensions}.}
\[
D_i\log\pi(t,t)\ \le\ \frac{\lambda'(t)}{r+\lambda(t)}.
\]
Hence, what disciplines the race is the growth of risk, not its level: a hazard that is severe but flat cannot stop the race. Moreover, comparative statics in the level of risk invert: a uniform upward shift of the hazard lowers the right-hand side and so prolongs the race. Firms that already live with severe background risk have little continuation value left to protect, so marginal risk is cheap to them---a ``nothing left to lose'' effect that is absent from the baseline model, where stopping extinguishes risk and uniformly higher hazards always shorten the race (\cref{prop:cs-tauR}(ii)).

\paragraph{Reversibility} Finally, if stopping is reversible, so that a firm that halts scaling may later resume, the ceiling $\tau^R$ no longer applies since the backward-induction argument leans on the permanence of stopping: Now races beyond $\tau^R$ can be sustained by rivals' threats to resume scaling. Moreover, those  threats can support efficient early stopping strictly below $\tau^R$, sustained by the threat of racing. We are exploring these issues in ongoing work.

\setlength{\bibsep}{0pt}
\bibliography{ref}

\appendix
\crefalias{section}{appendix}
\crefalias{subsection}{appendix}
\normalsize 

\appendix 

\begin{center}
    \large{\textbf{APPENDIX}}
\end{center}

\section{Omitted Proofs} \label{appendix:proofs}

\paragraph{Technical lemmas} We first record a few useful lemmas we will reference throughout the proofs. For stopping dates $(\tau_i,\tau_j)\in\overline\R_+^{\,2}$, define
\begin{align*}
    U_i(\tau_i,\tau_j)
    &:=
    \int_0^\infty e^{-rs-\int_0^s\lambda(u)\mathbf 1\{u\le\tau_i\vee\tau_j\}\de u}\,
    \pi\big(s\wedge\tau_i,\ s\wedge\tau_j\big)\,\de s .
\end{align*}
For a date $s\ge0$, interpret a stopping date below $s$ as a technology frozen in the past and a stopping date in $[s,+\infty]$ as the plan of a firm that remains active. The \emph{date-$s$ continuation payoff} is
\[
U_i^s(\tau_i,\tau_j)
:=
\int_s^\infty
\exp\Big\{-r(u-s)-\int_s^u\lambda(v)\,\mathbf 1\{v\le\tau_i\vee\tau_j\}\,\de v\Big\}\,
\pi\big(u\wedge\tau_i,\ u\wedge\tau_j\big)\,\de u ,
\]
so that $U_i^0=U_i$.

\begin{lemma}[Payoff continuity]\label{lem:cont}
Write $\overline\R_{+}=[0,\infty]$ and measure the distance between two dates by $d(x,y)=|e^{-x}-e^{-y}|$. Under \cref{ass:payoffs}, for every firm $i$ and date $s\ge0$, the payoff $U_i^s(\tau_i,\tau_j)$ is finite, continuous in the two stopping dates, and bounded uniformly over those dates. In particular, $U_i$ is finite and continuous, and for every rival stopping date $\tau_j$ and every finite $t\ge0$,
\[
\lim_{\e\downarrow0}\,U_i(t+\e,\tau_j)
=
U_i(t,\tau_j).
\]
Thus the payoff from stopping immediately after $t$ converges to the payoff from stopping at $t$, whether or not the rival stops at $t$.
\end{lemma}

\begin{proof}
    Let $(\tau_i^k,\tau_j^k)\to(\tau_i,\tau_j)$ in $\overline\R_+^{\,2}$ and write $T_k:=\tau_i^k\vee\tau_j^k$ and $T:=\tau_i\vee\tau_j$. Fix $u\ge s$. Both capped stopping dates converge, and
\[
\int_s^u\lambda(v)\,\mathbf 1\{v\le T_k\}\,\de v
\ \longrightarrow\
\int_s^u\lambda(v)\,\mathbf 1\{v\le T\}\,\de v ,
\]
because the integrands converge at every $v\neq T$ and $\lambda$ is integrable on $[s,u]$. Since $\pi$ is continuous, the integrand of $U_i^s$ converges pointwise in $u$.

Monotonicity of $\pi$ gives $\pi(u\wedge\tau_i^k,u\wedge\tau_j^k)\le\pi(u\wedge T_k,0)$, while, since $\lambda\ge0$ and $u\wedge T_k\le u$,
\[
\exp\Big\{-\int_s^u\lambda(v)\,\mathbf 1\{v\le T_k\}\,\de v\Big\}
\le
\exp\Big\{\int_0^s\lambda(v)\,\de v-\int_0^{u\wedge T_k}\lambda(v)\,\de v\Big\}.
\]
Hence the integrand is bounded by
\[
e^{rs+\int_0^s\lambda(z)\de z}\,
e^{-ru}\sup_{v\le u}\Big\{e^{-\int_0^v\lambda(z)\de z}\,\pi(v,0)\Big\},
\]
which is integrable on $[s,\infty)$ by \cref{ass:payoffs}(v), uniformly over stopping pairs. Dominated convergence delivers finiteness, the uniform bound, and continuity.
\end{proof}

\begin{lemma}[Earliest optimal plans]\label{lem:optimal-plan}
Let $f(t,s)$ be continuous on
\[
\big\{(t,s)\in\R_+\times\overline\R_+:t\le s\big\}.
\]
Then $a(t):=\min\argmax_{s\in[t,+\infty]}f(t,s)$ is well defined, lower semicontinuous, and hence Borel measurable. If the
maximizer is unique for every $t$, then $a$ is continuous.
\end{lemma}

\begin{proof}
The feasible tails $[t,+\infty]$ are compact and vary continuously with
$t$. Berge's maximum theorem therefore gives a nonempty, compact-valued,
upper-hemicontinuous maximizer correspondence. Given $t_m\to t$, choose a
subsequence along which $a(t_m)$ converges to its limit inferior, say
$q$. Upper hemicontinuity implies that $q$ is a maximizer at $t$, so
$q\ge a(t)$. Thus $a$ is lower semicontinuous and therefore Borel
measurable. If every maximizer is unique, every convergent subsequence of
$a(t_m)$ has limit $a(t)$, which gives continuity.
\end{proof}

\begin{proof}[Proof of \cref{lem:R-single-peaked}]
For $\tau\ge t$, split the date-zero payoff when the rival stops at
$t$ into the flow earned before $t$ and the continuation payoff from
$t$ onward:
\[
U_i(\tau,t)
=
\int_0^t e^{-ru-\int_0^u\lambda(z)\de z}\pi(u,u)\,\de u
+
e^{-rt-\int_0^t\lambda(z)\de z}R_t(\tau).
\]
Equivalently,
\[
R_t(\tau)
=
e^{rt+\int_0^t\lambda(z)\de z}
\left[
U_i(\tau,t)
-
\int_0^t e^{-ru-\int_0^u\lambda(z)\de z}\pi(u,u)\,\de u
\right].
\]
The same identity holds at $\tau=+\infty$ by taking limits. Hence
\cref{lem:cont} makes $(t,\tau)\mapsto R_t(\tau)$ continuous on the
feasible set with compactified stopping-date coordinate. Applying
\cref{lem:optimal-plan} to $f(t,\tau)=R_t(\tau)$ shows that a maximizer
exists for every $t$.

For $\tau\geq t$, differentiating the frozen-rival payoff gives
\[
\frac{\partial R_t(\tau)}{\partial \tau}
=
\frac{e^{-\int_t^\tau(r+\lambda(z))\de z}\pi(\tau,t)}{r}
\Big[
D_i\log\pi(\tau,t)-\lambda(\tau)
\Big].
\]
The prefactor is strictly positive. By \cref{ass:payoffs}(iv), the bracket
is strictly decreasing in $\tau$ when the rival is fixed at $t$. Hence
$R_t$ is single-peaked on $[t,+\infty]$, with a unique maximizer. The
unique-maximizer part of \cref{lem:optimal-plan} therefore makes the
maximizer $\chi(t)$ continuous in $t$. Since $\chi(t)\ge t$,
setting $\chi(+\infty):=+\infty$ extends the continuity of $\chi$ to $[0,+\infty]$. 

Stopping immediately is optimal exactly when the right derivative at $t$
is nonpositive, which gives the displayed equivalence. That equivalence
and continuity of the two rates make the set defining $\tau^R$ closed, so
finite $\tau^R$ belongs to it. If $t<\tau^R$, then $t$ does not belong to
that set, so the displayed equivalence gives the strict inequality.
\end{proof}

\begin{proof}[Proof of \cref{prop:frontier}]
If $\underline\tau=0$, the claim is immediate. For a date $t\ge0$ at which a firm stopped, define the sole survivor's crossing date
$
\chi(t):=\inf\big\{s\ge t:\ D_i\log\pi(s,t)\le\lambda(s)\big\}.$
By \cref{lem:R-single-peaked},  at every date in $[t,\chi(t)]$ a firm whose rival stopped at $t$ has unique optimal stopping time $\chi(t)$.  Moreover, $\chi(t)\ge\underline\tau$ for every $t\ge0$: for every $s<\underline\tau$, log-supermodularity and the definition of $\underline\tau$ give
$
D_i\log\pi(s,t)\ \geq\ D_i\log\pi(s,0)\ >\ \lambda(s),
$
so no crossing occurs before $\underline\tau$. Finally, $\chi$ is continuous because the remaining firm's payoff is jointly continuous in its stopping date and in the rival's stopping date $t$ on the compactified stopping-time space (\cref{lem:cont}), and its maximizer is unique; continuity follows by Berge's maximum theorem: every subsequential limit of $\chi(t_m)$ is optimal against the limit $t$ and hence equals $\chi(t)$.

Before either firm has stopped, let $X_i$ denote firm $i$'s candidate stopping date; $X_1$ and $X_2$ are independent because the firms' randomization devices are. By the preceding paragraph, the realized profile is $(X_1,\chi(X_1))$ if $X_1<X_2$, $(\chi(X_2),X_2)$ if $X_2<X_1$, and $(X_1,X_1)$ if $X_1=X_2$. On $\{X_1\neq X_2\}$ the frontier equals $\chi(\min_iX_i)\ge\underline\tau$. Hence, up to null events, $\{\max_j\tau_j<\underline\tau\}\subseteq\{X_1=X_2<\underline\tau\}$, and by independence a positive-probability tie requires a common atom: some $t<\underline\tau$ with $\Pr(X_1=t)>0$ and $\Pr(X_2=t)>0$.

Suppose such an atom exists. On the device event $\{X_1=t\}$, consider the deviation in which firm $1$ replaces its candidate $t$ by $t+\epsilon$, where $0<\epsilon<\underline\tau-t$, and responds at $\chi(t')$ if firm $2$ stops first at some date $t'$; leave the strategy unchanged on $\{X_1\neq t\}$. Because $\{X_1=t\}$ is measurable with respect to firm $1$'s private randomization device and the response depends only on subsequent public records, this deviation is admissible. Conditional on $X_1=t$: if $X_2<t$, the outcome is unchanged. If $X_2=t$, the profile changes from $(t,t)$ to $(\chi(t),t)$, a strict gain
$
U_1\big(\chi(t),t\big)-U_1(t,t)\ >\ 0,$
since $\chi(t)\ge\underline\tau>t$ and the frozen-rival payoff strictly increases up to its unique maximizer $\chi(t)$. The event $\{t<X_2\le t+\epsilon\}$ has probability converging to zero as $\epsilon\downarrow0$, and payoffs are uniformly bounded (\cref{lem:cont}). If $X_2>t+\epsilon$, the profile changes from $(t,\chi(t))$ to $(t+\epsilon,\chi(t+\epsilon))$, whose payoff effect vanishes as $\epsilon\downarrow0$ by continuity of $\chi$ and of $U_1$. The conditional expected gain from the deviation therefore converges to
$\Pr(X_2=t)\big[U_1(\chi(t),t)-U_1(t,t)\big]>0$,
so the deviation is strictly profitable for all sufficiently small $\epsilon$, contradicting subgame perfection. Hence $\Pr(\max_j\tau_j<\underline\tau)=0$.
\end{proof}

\begin{proof}[Proof of \cref{prop:pure-priority}]
We construct an equilibrium in which firm $1$ is the designated first
stopper. By \cref{lem:R-single-peaked}, once one firm stops at $s$, its
rival's unique sequentially rational response is $\chi(s)$. Fix a no-stop history at date $t$. If firm 1 stops first at $s\ge t$, write
\[
L_t(s):=U_1^t\big(s,\chi(s)\big),\qquad s\in[t,+\infty],
\]
for firm 1's continuation payoff from that choice. For $t\le t'\le s$, splitting the continuation payoff from racing to $s$ as evaluated at $t$ at the intermediate time $t'$ gives
\[
L_t(s)
=
\int_t^{t'}
e^{-r(u-t)-\int_t^u\lambda(z)\de z}\pi(u,u)\,\de u
+
e^{-r(t'-t)-\int_t^{t'}\lambda(z)\de z}L_{t'}(s).
\]
The identity also holds at $s=+\infty$ by taking limits. Take the first
date to be $0$ and solve for $L_{t'}(s)$. The function
$L_0(s)=U_1\big(s,\chi(s)\big)$ is continuous by \cref{lem:cont} and
continuity of $\chi$, while the prefix and the positive coefficient vary
continuously with $t'$. Hence $(t',s)\mapsto L_{t'}(s)$ is jointly
continuous on the feasible set with compactified stopping-date
coordinate. Applying
\cref{lem:optimal-plan} to
$f(t,s)=L_t(s)$ shows that the earliest maximizer
$
a(t):=\min\argmax_{s\in[t,+\infty]}L_t(s)$ 
is well defined and Borel measurable. The same identity verifies that the selected plan does not change as time passes. Its first term and the positive coefficient on
$L_{t'}(s)$ do not depend on $s$. 

Prescribe the following profile: at every no-stop history at date $t$, firm $1$ plans to stop at $a(t)$ and firm $2$ waits; after a first stop at $s$, the remaining firm, active at date $t'$, stops at $\chi(s)\vee t'$. If, off path, both firms remain active beyond firm $1$’s previously prescribed stopping date, restart the same construction from the current date: the designated firm stops at \(a(t)\), and firm $2$ continues. These prescriptions define a valid pure strategy profile. Before firm $1$  stops, \(a(t')=a(t)\), so its planned stopping date remains unchanged. After a stop at \(s\), the survivor continues to target \(\chi(s)\), stopping immediately if that date has already passed.

We verify sequential rationality at an arbitrary two-active history at date $t$. Write $a:=a(t)$. Since firm $2$ waits until a stop occurs, any deviation by firm $1$ reduces to a first stopping date $s\ge t$, with payoff $L_t(s)$; since $a(t)$ maximizes $L_t$,$
L_t(s)\le L_t(a).$
Thus the designated firm has no profitable deviation.

Along the prescribed path, firm $2$ receives $F_t(a):=U_2^t\big(a,\chi(a)\big)$. Before $a$ the two firms have identical technologies; after $a$, firm $2$'s own technology is weakly higher than firm $1$'s, while its rival's technology is weakly lower, and both firms face the same survival kernel. Because flow profit rises in own technology and falls in rival technology,
$
F_t(a)\ge L_t(a).
$
Consider any deviation by the waiting firm. If it stops first at $s<a$, firm $1$ responds at $\chi(s)$; symmetry implies that the deviator's payoff is  $L_t(s)$, and therefore
$
L_t(s)\le L_t(a)\le F_t(a)
$by the two preceding displays. If it stops at $s=a$, the firms stop together and its payoff is $U_2^t(a,a)$; once firm $1$ stops at $a$, immediate stopping is feasible for firm $2$ whereas $\chi(a)$ is its unique optimal response, so
$
U_2^t(a,a)\le F_t(a).
$
If its candidate is later than $a$, firm $1$ stops first; for every successor choice $\tau\ge a$, optimality of $\chi(a)$ gives
$
U_2^t(a,\tau)\le U_2^t\big(a,\chi(a)\big)=F_t(a).
$
These cases cover every  deviation by firm $2$. (If $a=+\infty$, every finite $s$ belongs to the first case and $s=+\infty$ gives the equilibrium tie.)

The same arguments apply at every history at which both firms are active, and the single-peakedness established at the start of the proof gives optimality at every history at which one firm has stopped. 
\end{proof}

\begin{proof}[Proof of \cref{prop:taxonomy}]
Fix $\bm\tau\in\SPE(\pi,\lambda)$. By \cref{prop:bounds}, $\max_i\tau_i\leq\tau^R$ almost surely. It thus suffices to rule out two cases: both firms stop
at the same date strictly below $\tau^R$, or the firms stop at different
dates and the later date is $\tau^R$.

First suppose that one firm stops at time $t$ while its rival remains active. By \cref{lem:R-single-peaked}, the best response $\chi(t)$ is  is unique and continuous, and 
$\chi(\tau^R)=\tau^R$.
\footnote{Payoffs earned before $t$ do not depend on $\tau$, so
$\chi(t)$ also uniquely maximizes $U_i(\tau,t)$ over $\tau\geq t$. To verify
continuity directly, take $t_m\to t$ and a convergent subsequence, relabelled
$t_m$, for which $\chi(t_m)\to q$. Compactness guarantees such a subsequence.
Since $\chi(t_m)\geq t_m$, we have $q\geq t$. For any $\tau\geq t$,
optimality gives
$U_i(\chi(t_m),t_m)\geq U_i(\max\{\tau,t_m\},t_m)$. Passing to the limit
using \cref{lem:cont} gives $U_i(q,t)\geq U_i(\tau,t)$. Thus $q$ is optimal
at $t$, and uniqueness gives $q=\chi(t)$. Every subsequential limit therefore
equals $\chi(t)$, proving continuity. Finally, $\chi(t)\leq\tau^R$ for every $t\leq\tau^R$:
for $\tau>\tau^R$, log-supermodularity gives
$D_i\log\pi(\tau,t)\leq D_i\log\pi(\tau,\tau)$, and single-crossing gives
$D_i\log\pi(\tau,\tau)<\lambda(\tau)$ strictly, so
$\partial R_t(\tau)/\partial\tau<0$ throughout $(\tau^R,+\infty)$ and the
maximizer cannot exceed $\tau^R$}
Subgame perfection therefore requires a sole survivor to stop at $\chi(t)$
almost surely after its rival stops at $t$.

Before either firm has stopped, let $X_i$ denote firm $i$'s candidate stopping date. The candidates $X_1$ and $X_2$ are independent because the firms' randomization devices are independent. The realized profile is
\[
(\tau_1,\tau_2)
=
\begin{cases}
\bigl(X_1,\chi(X_1)\bigr),&X_1<X_2,\\[2mm]
\bigl(\chi(X_2),X_2\bigr),&X_2<X_1,\\[2mm]
(X_1,X_1),&X_1=X_2.
\end{cases}
\]

We now rule out the two cases. 

\noindent\emph{Case 1: bunching below $\tau^R$.}
Since $\chi(t)>t$ for every $t<\tau^R$, bunching below $\tau^R$ can occur only when $X_1=X_2<\tau^R$. Independence then implies that a positive-probability bunching event requires a common atom: there must be some $t<\tau^R$ such that $
\Pr(X_1=t)>0
\qquad\text{and}\qquad
\Pr(X_2=t)>0.$ On the event $\{X_1=t\}$, consider the deviation in which firm $1$ replaces its candidate stopping date $t$
with $t+\epsilon$, where $\epsilon>0$ is small enough that
$t+\epsilon<\tau^R$. If firm $2$ stops first at some date $s$, firm $1$
responds at $\chi(s)$. Leave the strategy unchanged on $\{X_1\neq t\}$.
Because $\{X_1=t\}$ is measurable with respect to firm $1$'s private
randomization device and the response depends only on subsequent public
records, this deviation is admissible.

Conditional on $X_1=t$, the deviation leaves the outcome unchanged when $X_2<t$. When $X_2=t$, it changes the profile from $(t,t)$ to $(\chi(t),t)$ and yields the strict gain
$
U_1\bigl(\chi(t),t\bigr)-U_1(t,t)>0,
$
because $t$ is not optimal against a rival frozen at $t$. The event
$
t<X_2\leq t+\epsilon
$
has probability converging to zero as $\epsilon\downarrow0$, and payoffs are bounded by \cref{lem:cont}. Finally, when $X_2>t+\epsilon$, the profile changes from
$
\bigl(t,\chi(t)\bigr)
\quad\text{to}\quad
\bigl(t+\epsilon,\chi(t+\epsilon)\bigr),
$
whose payoff effect converges to zero by continuity of $\chi$ and $U_1$. The conditional expected gain from the deviation therefore converges to
$
\Pr(X_2=t)
\left[
U_1\bigl(\chi(t),t\bigr)-U_1(t,t)
\right]
>0.
$
For all sufficiently small $\epsilon$, the deviation is strictly profitable, contradicting subgame perfection. Hence
$
\Pr(\tau_1=\tau_2<\tau^R)=0.
$

\smallskip
\noindent\emph{Case 2: spread profile reaching $\tau^R$.}
Suppose, toward a contradiction, that
$
\Pr(\tau_1<\tau^R=\tau_2)>0;
$
where we assume that firm $1$ is the laggard without loss since payoffs are symmetric. This implies
that this event is contained in
$
\{X_1<X_2,\ X_1<\tau^R,\ \chi(X_1)=\tau^R\}.
$

On the event
$
\{X_1<\tau^R,\ \chi(X_1)=\tau^R\},
$
consider the deviation in which firm $1$ replaces its candidate stopping date $X_1$ with $\tau^R$. If firm
$2$ stops first at $y$, firm $1$ responds at $\chi(y)$. Leave the strategy
unchanged on the complement of this event. Continuity of $\chi$ makes the
event measurable with respect to firm $1$'s private randomization device, so
this is an admissible deviation.

Fix a realization of firm $1$'s device to which this deviation applies, and
write $x=X_1$. If $X_2<x$, firm $2$ stops first under both strategies and the
outcome is unchanged. If $X_2=x$, the outcome
changes from $(x,x)$ to
$
\bigl(\chi(x),x\bigr)=(\tau^R,x),
$
which is strictly better because by construction $\chi(x)$ is the unique optimum against a
rival frozen at $x$.

Now suppose $x<X_2=y<\tau^R$. Originally, firm $1$ stops at $x$ and firm
$2$ responds at $\chi(x)=\tau^R$, giving the profile $(x,\tau^R)$. Under
the deviation, firm $2$ stops first at $y$ and firm $1$ responds at
$\chi(y)$, giving $\bigl(\chi(y),y\bigr)$. Since
$
x<y<\chi(y)\leq\tau^R,
$
firm $1$'s technology is everywhere weakly higher, firm $2$'s technology
is everywhere weakly lower, so the disaster
risk is weakly lower. At least one comparison is strict on a positive
interval, so firm $1$'s payoff rises strictly.

Finally, suppose $X_2\geq\tau^R$. If $\tau^R<+\infty$, the deviating
profile is $(\tau^R,\tau^R)$: either the firms stop together at $\tau^R$,
or firm $1$ stops there and firm $2$ responds immediately because
$\chi(\tau^R)=\tau^R$. If $\tau^R=+\infty$, this case is simply
$X_2=+\infty$, and both firms race forever under the deviation. In either
case, relative to $(x,\tau^R)$, firm $1$'s technology rises while its
rival's technology and the frontier remain unchanged.
The deviation is therefore weakly profitable for every realization and hence
strictly profitable on the positive probability event above. This contradicts
subgame perfection. Hence
$
\Pr(\tau_1<\tau^R=\tau_2)
=
\Pr(\tau_2<\tau^R=\tau_1)
=0.$ Combining the two exclusions with \cref{prop:bounds} yields the result. 
\end{proof}

\begin{proof}[Proof of \cref{lem:mono}]

Fix $t\geq 0$. For each detection date $a\geq t$, let
$\sigma^\star(a) := \inf\{u\geq a: D_i\log\pi(u,t)\leq\lambda(u)\}$,
and define
\[
\Phi_t(a)
:=
\int_t^{\sigma^\star(a)}
\frac{\alpha(s)}{\alpha(t)}\pi(t,s)\,\de s
+
\frac{\alpha(\sigma^\star(a))}{\alpha(t)}
\frac{\pi(t,\sigma^\star(a))}{r}.
\]
Thus, if $E_\omega\sim\operatorname{Exp}(\omega)$, then
$V_S(t,\omega)=\Ex[\Phi_t(t+E_\omega)]$.

Let $q_t := \inf\{u\geq t: D_i\log\pi(u,t)\leq\lambda(u)\}$.
By \cref{ass:payoffs}(iv), $D_i\log\pi(u,t)-\lambda(u)$ is strictly
decreasing in $u$, so $\sigma^\star(a)=\max\{a,q_t\}$. Hence $\Phi_t$
is constant for $a\leq q_t$. For $a>q_t$, we have
$\sigma^\star(a)=a$, and differentiation gives
$\Phi_t'(a) = \frac{\alpha(a)}{\alpha(t)}\frac{\pi(t,a)}{r}
\left[D_j\log\pi(t,a)-\lambda(a)\right] < 0$,
because $D_j\pi<0$ and $\lambda\geq 0$. Therefore $\Phi_t$ is weakly
decreasing in the detection date.

Let $Z\sim\operatorname{Exp}(1)$ and couple the detection delays by
$E_\omega=Z/\omega$. If $\omega_2>\omega_1$, then
$t+Z/\omega_2<t+Z/\omega_1$ almost surely. Since $\Phi_t$ is weakly
decreasing, $\Phi_t\left(t+\frac{Z}{\omega_2}\right)
\geq
\Phi_t\left(t+\frac{Z}{\omega_1}\right)$.
Taking expectations shows that $V_S(t,\omega)$ is nondecreasing in
$\omega$.

If $\tau^R<\infty$, then $q_t<\infty$. Indeed, for every
$u\geq\max\{t,\tau^R\}$, log-supermodularity and \cref{ass:dsc} give
$D_i\log\pi(u,t) \leq D_i\log\pi(u,u) \leq \lambda(u)$.
Moreover, $\Phi_t$ is strictly decreasing above $q_t$. The coupled
inequality is therefore strict on an event of positive probability,
so $V_S(t,\omega)$ is strictly increasing in $\omega$. This proves
part~(a).

As $\omega\downarrow 0$, the coupling gives $E_\omega\to\infty$
almost surely. Since $\sigma^\star(a)\geq a$, it follows that
$\sigma^\star(a)\to\infty$ as $a\to\infty$. Hence
$\int_t^{\sigma^\star(a)}
\frac{\alpha(s)}{\alpha(t)}\pi(t,s)\,\de s
\longrightarrow
\int_t^\infty
\frac{\alpha(s)}{\alpha(t)}\pi(t,s)\,\de s$.
Because $\pi(t,\cdot)$ is decreasing,
$0
\leq
\frac{\alpha(\sigma^\star(a))}{\alpha(t)}
\frac{\pi(t,\sigma^\star(a))}{r}
\leq
\frac{\alpha(\sigma^\star(a))}{\alpha(t)}
\frac{\pi(t,t)}{r}
\longrightarrow
0$.
Therefore $\Phi_t(a)
\longrightarrow
\int_t^\infty
\frac{\alpha(s)}{\alpha(t)}\pi(t,s)\,\de s$.
Since $\Phi_t$ is nonnegative and weakly decreasing, it is bounded by
$\Phi_t(t)<\infty$. Bounded convergence therefore gives
$V_S(t,0)
:=
\lim_{\omega\downarrow 0}V_S(t,\omega)
=
\int_t^\infty
\frac{\alpha(s)}{\alpha(t)}\pi(t,s)\,\de s$. For every $s>t$, strict monotonicity in own technology gives
$\pi(t,s)<\pi(s,s)$. Therefore
\[
V_S(t,0)
=
\int_t^\infty\frac{\alpha(s)}{\alpha(t)}\pi(t,s)\,\de s
\ <\
\int_t^\infty\frac{\alpha(s)}{\alpha(t)}\pi(s,s)\,\de s
=
V_R(t),
\]
which proves part~(b).

Finally, suppose $t\geq\tau^R$. For every $a\geq t$,
log-supermodularity and \cref{ass:dsc} imply
$D_i\log\pi(a,t) \leq D_i\log\pi(a,a) \leq \lambda(a)$.
Thus $\sigma^\star(a)=a$, and
\[
\Phi_t(a)
=
\int_t^a
\frac{\alpha(s)}{\alpha(t)}\pi(t,s)\,\de s
+
\frac{\alpha(a)}{\alpha(t)}
\frac{\pi(t,a)}{r}.
\]
As $\omega\to\infty$, the coupling gives $E_\omega\to 0$ almost
surely. Since
$0\leq\Phi_t(t+E_\omega)\leq\Phi_t(t)=\pi(t,t)/r$, bounded
convergence yields
$\lim_{\omega\to\infty}V_S(t,\omega)
=
\frac{\pi(t,t)}{r}$.

It remains to compare this limit with $V_R(t)$. Since $D_j\pi<0$,
we have $\pi(s,s)<\pi(s,t)$ for every $s>t$, and hence
$V_R(t)
\leq
\int_t^\infty
\frac{\alpha(s)}{\alpha(t)}\pi(s,t)\,\de s$.
For $s\geq t\geq\tau^R$, log-supermodularity and \cref{ass:dsc} give
$D_i\log\pi(s,t)\leq\lambda(s)$. Consequently,
\[
-\frac{\de}{\de s}
\left[
\frac{\alpha(s)}{\alpha(t)}
\frac{\pi(s,t)}{r}
\right]
=
\frac{\alpha(s)}{\alpha(t)}
\frac{\pi(s,t)}{r}
\left[
r+\lambda(s)-D_i\log\pi(s,t)
\right]
\geq
\frac{\alpha(s)}{\alpha(t)}\pi(s,t).
\]
Moreover, $\frac{\alpha(s)}{\alpha(t)}\pi(s,t)
\leq
\pi(t,t)e^{-r(s-t)}
\longrightarrow
0$,
where the inequality follows from
$D_i\log\pi(s,t)\leq\lambda(s)$. Integrating the preceding derivative
inequality from $t$ to infinity therefore gives
$\int_t^\infty
\frac{\alpha(s)}{\alpha(t)}\pi(s,t)\,\de s
\leq
\frac{\pi(t,t)}{r}$.
It follows that
$V_R(t)
\leq
\frac{\pi(t,t)}{r}
=
\lim_{\omega\to\infty}V_S(t,\omega)$.
The first inequality is strict for $t>\tau^R$, proving part~(c).
\end{proof}
\begin{proof}[Proof of \cref{prop:never-stop}]
We construct a symmetric assessment. At every no-signal history, each firm races forever. After a first signal at date $a$, the active firm assigns probability one to the rival having stopped at $a$. At every subsequent date $s\ge a$ at which it remains active, it stops at
\[
\inf\big\{u\ge s:\ D_i\log\pi(u,a)\le\lambda(u)\big\},
\]
with the convention $\inf\varnothing=+\infty$. By \cref{lem:R}, this is the unique optimal continuation under the assigned belief. In particular, if crossing never occurs, racing forever is optimal.

We verify optimality at a no-signal history at date $t$. The rival never stops under the proposed strategy, so no signal arrives and racing forever is worth $V_R(t)$. Consider first a pure deviation that plans to stop at some finite $\sigma\ge t$ absent a signal. The rival then receives its signal at a date $a\ge\sigma$. Under the off-path belief above, its prescribed response is
$
\inf\big\{u\ge a:\ D_i\log\pi(u,a)\le\lambda(u)\big\}.$
If instead the rival knew the deviator's true stopping date $\sigma$, its response would be
$
\inf\big\{u\ge a:\ D_i\log\pi(u,\sigma)\le\lambda(u)\big\}.
$
Because $a\ge\sigma$, log-supermodularity implies that the prescribed response is weakly later. A stopped firm is hurt by any delay in its rival's stop (\cref{lem:ex-down}), so the deviator's continuation payoff after stopping at $\sigma$ is at most $V_S(\sigma,\omega)$. By \cref{lem:rec}, its total deviation payoff is thus at most
\begin{align*}
\int_t^\sigma\frac{\alpha(s)}{\alpha(t)}\pi(s,s)\,\de s
+\frac{\alpha(\sigma)}{\alpha(t)}V_S(\sigma,\omega)
&=V_R(t)+\frac{\alpha(\sigma)}{\alpha(t)}\Delta(\sigma,\omega)\\
&\le V_R(t).
\end{align*}
The last inequality follows because $\omega<\bar\omega$ and \cref{lem:mono}(a) imply $\Delta(\sigma,\omega)\le0$ at every date. A deviation that races forever yields   $V_R(t)$, and randomized deviations cannot improve on this (by linearity of expectations). Hence racing forever is optimal at every no-signal history.

The specified beliefs obey Bayes' rule on path. Signal histories are off path, and the assigned point belief is admissible. Play is sequentially rational. Thus the assessment is a perfect Bayesian equilibrium, and both firms race forever almost surely.
\end{proof}
\begin{proof}[Proof of \cref{prop:LMH}]
Throughout the proof, abbreviate the two effective discount rates $\delta:=r+\lambda-\alpha+\beta=1/V_R$ (along the race path) and $\rho:=r+\lambda+\beta=1/V_S(0)$ (stopped beside a racing rival).

\noindent \underline{\emph{Part (i)}} Crazy firms never stop, and a signal can arrive only if some firm has already stopped, so the first stop in any equilibrium---if there is one---must occur at a \emph{quiet} history (active, no signal received). It therefore suffices to show that no rational firm ever stops at a quiet history: then no stop ever occurs, no signal is ever sent, and no responsive stop ever follows. Fix an equilibrium and let $F_j$ be the (sub-)distribution of rational firm $j$'s stopping time along its quiet histories, $S_j:=1-F_j$. Since an active firm $i$ has sent no signal, its rival's behavior is governed by $F_j$, and a stop at $c$ remains undetected with probability $e^{-\omega(t-c)}$; firm $i$'s posterior at a quiet history at date $t$ therefore puts weights proportional to
\[
\underbrace{1-p}_{\text{crazy}},\qquad
\underbrace{p\,S_j(t)}_{\text{rational, active}},\qquad
\underbrace{p\,e^{-\omega(t-c)}\de F_j(c)}_{\text{rational, stopped at }c\text{, undetected}} .
\]
Compare stopping now (denoted \textsc{stop}) with the feasible plan \textsc{wait}: continue, and stop immediately when a signal arrives. Cell by cell, per unit of current flow:

\noindent \emph{Crazy.} Stopping yields $V_S(0)$; \textsc{wait} races forever, worth $V_R$: the deviation gains $V_R-V_S(0)$.

\noindent  \emph{Stopped at $c$.} \textsc{stop} ends the risk at once, worth $e^{\beta(t-c)}/r$ (the factor is the flow advantage over a rival frozen at $c$); \textsc{wait} yields  $e^{\beta(t-c)}V_W(\omega)$, so  the deviation loses $e^{\beta(t-c)}\big[1/r-V_W(\omega)\big]$.

\noindent  \emph{Rational, active.} Let $\tau$ be the rival's residual quiet-path stopping time  and $X\sim\mathrm{Exp}(\omega)$ the detection lag. \textsc{stop} ends the risk at $X\wedge\tau$, worth $V_S(0)+\big[1/r-V_S(0)\big]\Ex[e^{-\rho(X\wedge\tau)}]$; \textsc{wait} races until $\tau$ and then exits upon detection, yielding $V_R+\big[V_W(\omega)-V_R\big]\Ex[e^{-\delta\tau}]$ (recall $V_W(\omega)>V_R$. Using $\Ex_X[e^{-\rho(X\wedge\tau)}]=\frac{\omega}{\omega+\rho}+\frac{\rho}{\omega+\rho}e^{-(\omega+\rho)\tau}$, the gap (\textsc{stop} minus \textsc{wait}) on this cell is $\Ex[\varphi(\tau)]$ where
\[
\varphi(\tau)=V_S(0)-V_R+\big[1/r-V_S(0)\big]\Big[\tfrac{\omega}{\omega+\rho}+\tfrac{\rho}{\omega+\rho}\,e^{-(\omega+\rho)\tau}\Big]-\big[V_W(\omega)-V_R\big]e^{-\delta\tau}.\]
Evaluating $\varphi'$  shows that is quasi-convex, so for every (random) $\tau$,
\[
\Ex[\varphi(\tau)]\ \le\ \max\{\varphi(0),\varphi(\infty)\}\ =\ \max\big\{1/r-V_W(\omega),\ V_S(\omega)-V_R\big\} .
\]
Finally, on the stopped cell the weight carries the damping $e^{-\omega(t-c)}\le e^{-\beta(t-c)}$ (here $\omega\ge\beta$), so its weighted gap is at most $\big[1/r-V_W(\omega)\big]\,p\!\int\de F_j$. Summing the cells and using $S_j(t)+\int_0^t\de F_j\le1$ and $p=\ell\,(1-p)$,
\begin{align*}
\big(\text{stop}\big)-\big(\textsc{wait}\big)\ &\propto\ -(1-p)\big[V_R-V_S(0)\big]+p\cdot\max\big\{1/r-V_W(\omega),\ V_S(\omega)-V_R\big\}\\
&=\ (1-p)\Big[\tfrac{\ell}{\,\underline\ell\wedge\bar\ell\,}-1\Big]\big[V_R-V_S(0)\big]\ <\ 0 ,
\end{align*}
since $\min\{\underline\ell,\bar\ell\}=\big[V_R-V_S(0)\big]/\max\{1/r-V_W(\omega),\,V_S(\omega)-V_R\}$ and $\ell<\underline\ell<\bar\ell$. So at every quiet history  stopping is strictly dominated by \textsc{wait}: rational firms never stop at quiet histories, under any pure or mixed rule, and by the reduction above no stop ever occurs.

\noindent  \underline{\emph{Part (ii)}} For (b), consider the profile in which no firm ever stops at a quiet history, and any active firm stops immediately upon a signal. On path no signal arrives, so quiet histories carry the prior odds $\ell$ and no undetected-stop weight. A deviation is a plan ``stop at $\sigma$ unless a signal arrives first''; no signal ever comes, so its value is
\[
(1-e^{-\delta\sigma})\,V_R+e^{-\delta\sigma}\,\frac{\ell\, V_S(\omega)+V_S(0)}{1+\ell}\,,
\]
since a stop at $\sigma$ is reciprocated upon detection by a rational rival (worth $V_S(\omega)$; the response is forced by $\alpha<\lambda$) and never by a crazy one (worth $V_S(0)$). For $\ell<\bar\ell$ the gamble is strictly below $V_R$, so every $\sigma<\infty$ is strictly worse than racing forever, and mixtures are averages of such plans; at $\ell=\bar\ell$ every finite plan is weakly worse---payoff-equivalent at the optimum---so the profile remains an equilibrium.

For (a), consider the bunched profile: rational firms stop at $t=0$, crazy firms never, and any active firm stops immediately upon a signal. A rational firm's deviations are again plans $\sigma\in[0,\infty]$ (continue until $\sigma$ or the first signal, whichever is first). With $X\sim\mathrm{Exp}(\omega)$ the date at which the deviator would detect a rational rival's stop at $0$, the plan's value is
\[
W(\sigma)=p\,\underbrace{\Ex\Big[\tfrac{1-e^{-(r+\lambda-\alpha)(X\wedge\sigma)}}{r+\lambda-\alpha}+\tfrac{e^{-(r+\lambda-\alpha)(X\wedge\sigma)}}{r}\Big]}_{\text{rational rival: race alone, stop at }X\wedge\sigma}
+(1-p)\,\underbrace{\Big[V_R-\big(V_R-V_S(0)\big)e^{-\delta\sigma}\Big]}_{\text{crazy rival: race, stop at }\sigma} .
\]
Then $W(0)=p/r+(1-p)\,V_S(0)$ is the bunching payoff, $W(\infty)=p\,V_W(\omega)+(1-p)\,V_R$, and
\[
W'(\sigma)=-\,p\,\tfrac{\lambda-\alpha}{r}\,e^{-(\omega+r+\lambda-\alpha)\sigma}+(1-p)\,\tfrac{V_R-V_S(0)}{V_R}\,e^{-\delta\sigma} .
\]
For $\omega\ge\beta$ the negative term decays faster, so $W'$ changes sign at most once, from $-$ to $+$: $W$ is maximized at an endpoint. And $W(0)\ge W(\infty)$ iff $p\,[1/r-V_W(\omega)]\ge(1-p)\,[V_R-V_S(0)]$, i.e.\ $\ell\ge\underline\ell$. So no plan improves on stopping at once.\footnote{Off the bunched path, a deviator continues optimally against its updated posterior, which affects no on-path incentive. The same computation at any common date $T$ shows delayed bunching is also an equilibrium: the environment is stationary and nothing is learned before $T$.}

\noindent  \underline{\emph{Part (iii)}} Since $(\bar\ell/\ell)^2<1$, the bound already excludes the never-stopping outcome. To see this, conjecture that neither firm ever stops. After every quiet history, each firm continues to believe that their rival is rational with odds  $\ell$ since the conjectured profile assigns zero probability to its rival stopping. Stopping immediately yields
\[
\frac{\ell\,V_S(\omega)+V_S(0)}{1+\ell}>V_R,
\]
and is therefore a profitable deviation. Note here that our assumption that $\alpha<\lambda$ ensures that a rational rival finds it dominant to stop after receiving news. For the bound, fix an equilibrium and, with $F_j,S_j$ as in part (i), let $s_j:=\lim_{t\to\infty}S_j(t)$. Conditional on both firms rational, the risk never dies only if neither firm ever stops spontaneously (responsive stops require a prior stop), an event of probability $s_1s_2$. If $s_j\le\bar\ell/\ell$ for both firms, then $s_1s_2\le(\bar\ell/\ell)^2$ and we are done. Otherwise, say $s_1>\bar\ell/\ell$; we claim $s_2=0$, so ruin has probability zero. On firm $2$'s quiet path, the odds on ``rational and active'' are $\ell S_1(t)\ge\ell s_1>\bar\ell$, so the immediate-stop gamble
\[
u(t):=\frac{\ell S_1(t)\,V_S(\omega)+V_S(0)}{1+\ell S_1(t)}\ \ge\ V_R+\eta\quad\text{for some }\eta>0
\]
(the gamble increases in the odds because $V_S(\omega)>V_R$, which holds whenever $\bar\ell<\infty$; the undetected-stop cell only raises the stop value further). Consider any continuation plan of firm $2$ at a late quiet history $t$. Relative to the deterministic benchmark ``race to $\sigma$, then stop,'' worth
\[
(1-e^{-\delta\sigma})\,V_R+e^{-\delta\sigma}u(t+\sigma)\ \le\ u(t)-(1-e^{-\delta\sigma})\,\eta
\]
(using that $u$ decreases along the quiet path), the plan can gain only through two  channels that vanish as $t \to \infty$: the rival's residual spontaneous concession of conditional probability $(S_1(t)-s_1)/S_1(t)\to 0$, since $S_1\downarrow s_1>0$, and the undetected-stop weight, which converges  to zero; each contributes at most $1/r$ per unit of weight. Hence for $t$ large, any plan that delays stopping beyond $t+w$ with positive probability is strictly dominated by stopping at once, and sequential rationality forces firm $2$ to stop by $t+w$ almost surely on its quiet path---contradicting $s_2>0$.

Finally, zero is attained: since $\ell>\bar\ell>\underline\ell$, the bunched profile of part (ii)(a) remains an equilibrium here, and conditional on both firms rational it ends the risk immediately.
\end{proof}

\section{Technical Details on Strategies} \label{appendix:technical}
This appendix defines strategies and equilibria. When we discuss firm $i$, let $j$ denote its rival. \cref{app:extform} defines the extensive form as well as mixing under perfect observability. It then shows that a firm's stopping date uses only information the firm has observed. \cref{app:noisy} gives the corresponding definition when firms observe delayed signals and defines perfect Bayesian equilibrium for \cref{sec:imperfect,appendix:existence}. We omit ``almost surely'' when it is clear from context.

\subsection{The extensive form under perfect observability}\label{app:extform}

We now describe play after each observed stop. At any public history, each active firm uses its current private draw to choose a candidate stopping date. The firm cannot wait to see its rival's action at the same instant before choosing its own. When a rival stops, the firm observes a new public history, makes a fresh draw, and uses its response plan. This construction follows \citet{riedel2017subgame} and \citet{laraki2005continuous}. Here public information changes only when a firm stops, so at most two such events occur.

\paragraph{Records, plans, and strategies.}
A \emph{record} is either empty, $p=\varnothing$, or a pair $p=(j,t_j)$ stating that firm $j$ stopped at $t_j$. Set $t(\varnothing):=0$ and $t(j,t_j):=t_j$. Under perfect observability the record is public. A \emph{plan} for an active firm $i$ is a Borel-measurable map $\sigma_i(p,z)\in[t(p),+\infty]$, where $z\in[0,1]$. When record $p$ begins, the firm's candidate date is $\sigma_i(p,Z_i^{(p)})$, unless the record changes first. The draw $Z_i^{(p)}$ is the relevant entry of
$
\bm Z_i^\ast(t):=\bigl(Z_{i,1},\,\mathbf 1\{\tau_j\le t\}Z_{i,2}\bigr):$
It equals $Z_{i,1}$ at the empty record and $Z_{i,2}$ after the rival stops. 

At a history $(t,p)$, date $t$ has arrived, record $p$ has not changed, and the firm remains active. A strategy also specifies a distribution $G_i(\,\cdot\mid t,p)$ for its candidate date at every such history. This distribution assigns probability one to $[t,+\infty]$, and its probability of every Borel set varies measurably with $(t,p)$. When record $p$ begins, the draw $Z_i^{(p)}$ generates  distribution
$
G_i(B\mid t(p),p)
=
\int_0^1\mathbf 1\{\sigma_i(p,z)\in B\}\,\de z .$
We impose the following dynamic-consistency condition:
for every $t\le s$ and every Borel set $B$, whenever $G_i([s,+\infty]\mid t,p)>0$,
\begin{equation*}
G_i(B\mid s,p)
=
\frac{G_i(B\cap[s,+\infty]\mid t,p)}
     {G_i([s,+\infty]\mid t,p)} .
\tag{DC}\label{eq:dc-kernel}
\end{equation*}
Equation \eqref{eq:dc-kernel} says that the firm does not redraw its candidate merely because time has passed. If it reaches $s$, its candidate is the old candidate conditional on being at least $s$. If the denominator is zero, the earlier plan would never reach $s$, but the strategy must still say what the firm does there. A fresh private uniform draw then selects a new candidate from $G_i(\,\cdot\mid s,p)$; \cref{lem:rep} shows that one draw can generate this distribution. The firm observes the draw, so its later choices may depend on it. After selecting the new candidate, the firm again keeps it as time passes according to \eqref{eq:dc-kernel}. These extra draws are used only at histories that the firm's earlier plan would never reach. Play from the initial history uses only $Z_{i,1}$ and, after the rival stops, $Z_{i,2}$.

Within a fixed record, the firm's draw fixes its candidate date. The firm therefore cannot wait to see whether its rival also stops at date $t$ before choosing its own date-$t$ action. If the rival stops at $t$, a new record begins and activates a new plan, which may also prescribe stopping at $t$. By \cref{lem:cont}, this response has the same payoff as stopping at $t+\e$ in the limit as $\e\downarrow0$, holding terminal stopping dates fixed. A strategy specifies a plan even for records that equilibrium play never reaches, so subgame perfection can test behavior after deviations.

\paragraph{Mixing.}
A mixed plan at record $p$ is a distribution over candidate stopping dates in $[t(p),+\infty]$. When the record begins, a fresh uniform draw selects one date from this distribution. Draws used by different firms and at different records are independent. The firm observes its selected date, but its rival observes only its eventual stopping action. The next lemma shows that one uniform draw can implement any such distribution.

\begin{lemma}[Implementing a stopping-date distribution]\label{lem:rep}
Let $\mathcal R$ be the set of records, and let the earliest feasible date $b(\rho)$ vary measurably with the record $\rho$. Suppose that $G(\rho,\cdot)$ is a distribution on $[b(\rho),+\infty]$ and that, for every Borel set $B$, its probability $G(\rho,B)$ varies measurably with $\rho$. Then there is a Borel-measurable map
$
Q:\mathcal R\times[0,1]\longrightarrow\overline\R_+$
such that $Q(\rho,Z)$ has distribution $G(\rho,\cdot)$ whenever $Z$ is uniform and independent of the information available when record $\rho$ begins. Conversely, every such map defines a family of distributions with these measurability properties.
\end{lemma}

\begin{proof}
This is the usual inverse-distribution construction. We first map all possible dates, including $+\infty$, into the unit interval. We then apply the inverse distribution function there and map the result back. Map $\overline\R_+$ to $[0,1]$ by $\psi(x)=x/(1+x)$ and $\psi(+\infty)=1$. Let $F_\rho$ be the distribution function of the image of $G(\rho,\cdot)$ under $\psi$, and define
\[
Q(\rho,z):=
\psi^{-1}\!\left(\inf\{y\in[\psi(b(\rho)),1]:F_\rho(y)\ge z\}\right),
\qquad z\in[0,1].
\]
For every finite $x$,
$\{(\rho,z):Q(\rho,z)\le x\}
=\{(\rho,z):b(\rho)\le x,\ z\le G(\rho,[0,x])\},$
so $Q$ is Borel-measurable. It lies in $[b(\rho),+\infty]$ and has the required distribution. Conversely, $
G(\rho,B)=\int_0^1\mathbf 1\{Q(\rho,z)\in B\}\,\de z$ 
has the required measurability for every Borel-measurable $Q$.
\end{proof}

\begin{lemma}[Outcomes]\label{lem:outcome}
Every strategy profile $\bm\sigma$ determines a unique pair of stopping dates $\bm\tau(\bm\sigma)=(\tau_1,\tau_2)$. Each firm's stopping date uses only information that the firm has observed: formally, $\tau_i(\bm\sigma)$ is an $(\mathcal F^i_t)_t$-stopping time.
\end{lemma}

\begin{proof}
\emph{Construction.} Call $\sigma_i(p,Z_i^{(p)})$ firm $i$'s \emph{candidate date} at record $p$. At the empty record, compare the two candidates. If both are $+\infty$, both firms race forever. If the candidates tie at a finite date, both firms stop there. Otherwise, the firm with the earlier candidate stops first. Its rival observes that stop, activates the record $(j,t_j)$, and uses its second draw to select a new candidate. The rival stops at that date or races forever if the candidate is $+\infty$. These cases give a unique stopping pair.

\emph{The stopping date uses only observed information.} Fix $i$ and $t$. Firm $i$ first computes its candidate at the empty record. If its rival does not stop strictly before that date, the candidate is its stopping date, with a tie producing a simultaneous stop. Otherwise, when the rival stops first at $\tau_j$, firm $i$ activates record $(j,\tau_j)$, uses $Z_{i,2}$ to compute its second candidate, and stops at that date. This reproduces the construction above. On $\{\tau_j>t\}$ the computation uses only $Z_{i,1}$, so $Z_{i,2}$ enters only after the rival's stop, as the filtration requires. Hence
$
\{\tau_i\le t\}\in\sigma\big(\bm Z_i^\ast(t),\ (\mathbf 1\{\tau_j\le s\})_{s\le t}\big)=\mathcal F^i_t.$\end{proof}

\paragraph{Subgames and equilibrium.}
For a record $p$ and a date $t\ge t(p)$, the \emph{subgame after $(t,p)$} is the continuation game in which no new stop has occurred since record $p$ began. At the empty record both firms remain active. At a record $(j,t_j)$, firm $j$ is frozen at $t_j$ and only firm $i$ remains active. An active firm's strategy includes its current plan and, at the empty record, its plan after a later rival stop. Given terminal stopping dates $(\tau_i,\tau_j)$, where a firm that has already stopped keeps its recorded date, firm $i$'s continuation payoff is $U_i^t(\tau_i,\tau_j)$.

An active rival's candidate has distribution $G_j(\,\cdot\mid t,p)$. At a history reached with positive probability, this is its original candidate conditional on being at least $t$. A draw reserved for a later rival stop has not yet been used. At a history the earlier plan would never reach, or after a firm has continued past its candidate, the strategy uses the fresh draw specified above. Condition \eqref{eq:dc-kernel} then makes the firm keep that new candidate as time passes.

\begin{definition}[Subgame-perfect equilibrium]\label{def:spe}
A strategy profile $\bm\sigma$ is a \emph{subgame-perfect equilibrium} if,
after every history $(t,p)$ and for every active firm $i$, no replacement of
firm $i$'s plans from that history onward raises its expected continuation
payoff
\[
\Ex\!\left[U_i^t(\tau_i,\tau_j)
\mathrel{\big|}\text{firm $i$'s information at $(t,p)$}\right].
\]
The replacement must leave firm $i$'s earlier play unchanged, but may depend
on every private draw the firm has already observed and may specify its
response to every later observed stop. The expectation is over the rival's
unobserved candidate and all future draws. We write $\SPE(\pi,\lambda)$ for
the set of outcomes $\bm\tau(\bm\sigma)$ generated by subgame-perfect
strategy profiles.
\end{definition}

\subsection{Strategies under imperfect observability}\label{app:noisy}

Under imperfect observability, a firm sees delayed signals rather than its rival's stopping date. The first private draw governs play before any signal, and a fresh second draw is made after the first detection. This subsection defines these plans and the perfect Bayesian equilibrium used in \cref{sec:imperfect,appendix:existence}.

\paragraph{Primitives.}
The game uses independent private $\mathrm{Uniform}[0,1]$ draws to mix and independent detection clocks $E_1,E_2\sim\mathrm{Exp}(\omega)$. The draw $Z_{i,1}$ governs firm $i$'s initial plan, and $Z_{i,2}$ is drawn after its first detection. At any history that its earlier plan would never reach, the firm makes another fresh uniform draw. All private draws and detection clocks are mutually independent. The clock $E_j$ is attached to firm $j$'s stop: if firm $j$ stops at $\tau_j<\infty$, the signal process
$
Y^j_t:=\mathbf 1\{\tau_j+E_j\le t\}$
jumps at the \emph{detection date} $\tau_j+E_j$, at which point its rival learns, conclusively, that firm $j$ has stopped; if $\tau_j=+\infty$, then $Y^j\equiv0$. A jump of $Y^j$ reveals the event that firm $j$ has stopped but not the date $\tau_j$ at which it did so.

\paragraph{Signal records, plans, and strategies.}
A \emph{signal record} for firm $i$ is either empty, $q=\varnothing$, or a pair $q=(j,a)$ stating that firm $i$ detected firm $j$ at date $a$. Set $a(\varnothing):=0$ and $a(j,a):=a$. Unlike the public records under perfect observability, a signal record is private and lists the detection date, not the rival's stopping date. A \emph{plan} is a Borel-measurable map $(q,z)\mapsto\sigma_i(q,z)\in[a(q),+\infty]$. When record $q$ begins, $\sigma_i(q,Z_i^{(q)})$ is firm $i$'s candidate date; the firm stops there unless a signal changes its record first. The private draws available by date $t$ are
$
\bm Z_i^\ast(t):=\bigl(Z_{i,1},\,Y^j_tZ_{i,2}\bigr).$ The draw $Z_i^{(q)}$ used at record $q$ equals $Z_{i,1}$ before detection and $Z_{i,2}$ after detection. A strategy specifies the initial plan and a response plan for every possible signal date. A plan cannot condition on an event at the same instant: a detection at $a$ activates the response plan, which may itself prescribe stopping at $a$. By \cref{lem:cont}, the payoff from stopping immediately after $a$ converges to the payoff from this response.

At a history $(t,q)$, a strategy also specifies a distribution $G_i(\,\cdot\mid t,q)$ for the firm's current candidate date. This distribution assigns probability one to $[t,+\infty]$, and its probability of every Borel set varies measurably with $(t,q)$. When record $q$ begins, the draw $Z_i^{(q)}$ generates distribution $G_i(B\mid a(q),q)
=\int_0^1\mathbf 1\{\sigma_i(q,z)\in B\}\,\de z .$  If the history is reached with positive probability, this is simply the original candidate date conditioned on that candidate date being at least $t$, so it satisfies \eqref{eq:dc-kernel} with 
$q$ in place of $p$. The firm holds onto that candidate for as long as it lies in the future.  If the firm deviates by continuing past that candidate, or if its earlier plan would never reach the history, a fresh private uniform draw selects a new candidate from the distribution specified there, as in \cref{lem:rep}. The firm observes this draw, and equilibrium requires that the new plan be optimal. After selecting the new candidate, the firm again keeps it as time passes according to \eqref{eq:dc-kernel}.

\begin{lemma}[Outcomes under imperfect observability]\label{lem:outcome-noisy}
Together with the detection clocks, every strategy profile determines a unique stopping pair up to a probability-zero event. Each firm's stopping date uses only information that the firm has observed. Formally, $\tau_i(\bm\sigma)$ is an $(\mathcal F^i_t)_t$-stopping time for
$
\mathcal F^i_t=\sigma\big(\bm Z_i^\ast(t),\ (Y^j_s)_{s\le t}\big).
$\end{lemma}

\begin{proof}
Starting from date zero, list each scheduled stop and each pending detection. Execute the earliest event. A stop starts that firm's detection clock; a detection activates the rival's second plan if the rival is still active. Continue until both firms have stopped or no finite event remains, in which case every active firm races forever. At most two stops and two detections can occur. Tied candidate dates produce a simultaneous stop. A pending detection almost surely does not tie a stop that was already scheduled because exponential delays have continuous distributions. If the response plan chooses the detection date itself, execute that stop immediately after the detection. Thus the stopping outcome is unique except on a probability-zero event.

Each firm $i$ can compute its stopping date from its two private draws and the path $Y^j$ of its rival's signal. It first evaluates the plan at the empty record. If $Y^j$ does not jump strictly before the resulting candidate, that candidate is $\tau_i$. If the signal arrives first, the firm evaluates its second plan at the signal record and follows that candidate. This reproduces the construction above, so $\{\tau_i\le t\}\in\mathcal F^i_t$.
\end{proof}

\paragraph{Histories, beliefs, and equilibrium.}
An \emph{observable history} for an active firm $i$ at date $t$ is the date together with its signal record: $h=(t,q)$. The firm also knows its own draws and current candidate date, and optimality is tested using all this information. Beliefs can be written as a function of $h$ alone because the firm's draws are independent of whether its rival has stopped.

The rival may have stopped at some date $y\le t$, leaving its technology frozen while its signal remains pending, or it may still be active. A belief $\mu_i$ assigns probabilities to these possibilities at every observable history. To evaluate a continuation plan, firm $i$ averages over them and over any pending detection delay. If the rival is active, its signal record is empty because firm $i$ has not stopped, so its current candidate has distribution $G_j(\,\cdot\mid t,\varnothing)$.

\begin{definition}[Perfect Bayesian equilibrium]\label{def:pbe}
A \emph{perfect Bayesian equilibrium} is a strategy profile $\bm\sigma$ together with beliefs $(\mu_1,\mu_2)$ such that:
\begin{itemize}
\item[(i)] \emph{Sequential rationality:} after every observable history and every realization of the firm's private information, the firm cannot gain by changing its plans from that point onward while leaving its earlier play unchanged;
\item[(ii)] \emph{Consistent beliefs:} at every no-signal history reached with positive probability, beliefs follow Bayes' rule. After a detection at date $a$, if the Bayes denominator implied by the rival's strategy is positive, the posterior is pinned down by Bayes' rule. If the denominator is zero such that detection at $a$ is impossible under the rival's strategy, then we only impose the restriction that the belief must assign probability one to stopping dates $y\le a$. 
\end{itemize}\end{definition} 

Before any signal, the rival's strategy and the absence of news determine the firm's belief. After a signal, the firm uses the Bayes posterior at every detection date for which that posterior has a positive denominator. The existence proof computes this posterior in \cref{lem:ex-belief}. At feasible dates not covered pointwise by the fixed-point condition, it selects a measurable best response to this posterior. Only when the rival’s strategy makes detection at date $a$ impossible—that is, when the Bayes denominator is zero—may the firm choose another belief; the construction in \cref{appendix:existence} then assigns $\delta_a$, meaning that the rival stopped at the detection date $a$.

\paragraph{Concession and response plans.}
A strategy has two parts. The \emph{concession plan} is the distribution of the stopping date before any signal. The \emph{response plan} assigns a distribution of stopping dates in $[a,+\infty]$ to each first-signal date $a$. \Cref{lem:rep} shows how the private draws generate both plans. Together with the instructions at histories that the earlier plan would never reach, these plans define a complete strategy. At any history at which firm $i$ is active it has sent no signal, so an active rival's signal record is empty. The belief $\mu_i$ therefore reduces to a distribution over whether, and when, the rival has stopped, as computed in \cref{lem:ex-belief}.

Thus mixing has a simple meaning: after each signal record, the strategy assigns a distribution over candidate stopping dates, and a fresh draw selects one candidate if that record is reached. By the checking-deviations argument in the main text, we can verify optimality by checking continuation plans that leave earlier play unchanged and make no further random choices.

\clearpage

\appendix 
\renewcommand{\thesection}{\Roman{section}}

\renewcommand{\theHsection}{\Roman{section}}
\setcounter{section}{0}

\setcounter{page}{1}
{\renewcommand{\thefootnote}{\fnsymbol{footnote}}
\begin{center}
    \large{\textbf{\MakeUppercase{Online Appendix to 
    `Racing to Ruin'}}} \\
    \normalsize{{\textbf{DREW FUDENBERG\footnote[1]{Department of Economics, MIT, \href{mailto:drewf@mit.edu}{drewf@mit.edu}} \quad  ANDREW KOH\footnote[2]{Department of Economics, Columbia University, \href{mailto:ajkoh@mit.edu}{andrew.koh@columbia.edu}}}\\
    \small{FOR ONLINE PUBLICATION ONLY}}}
\end{center}}
\crefalias{section}{onlineappendix}

\section{Equilibrium Existence under Imperfect Transparency}\label{appendix:existence}

This appendix proves \cref{prop:existence} for the two-firm game. Set $\alpha(s)=e^{-rs-\int_0^s\lambda(u)\de u}$. For a realized stopping pair $(\tau_1,\tau_2)\in\overline\R_+^{\,2}$, write firm 1's realized payoff as
\[
U(\tau_1,\tau_2):=U_1(\tau_1,\tau_2)
=\int_0^\infty e^{-rs-\int_0^s\lambda(u)\mathbf 1\{u\le\tau_1\vee\tau_2\}\de u}\,
\pi\big(s\wedge\tau_1,\ s\wedge\tau_2\big)\,\de s ,
\]
firm 2's realized payoff being $U(\tau_2,\tau_1)$ by symmetry. Write $\bar U:=\sup|U|$, and let $\eta_U$ denote a modulus of continuity for $U$.

\paragraph{Overview of proof.} 
\cref{fig:appB-map} gives a roadmap. If $\omega<\bar\omega$, the never-stopping equilibrium from \cref{prop:never-stop} proves existence. The argument below covers all remaining news rates, including $\omega=\bar\omega$. It proceeds in two steps. First, we solve a smaller game in which each firm chooses a distribution of initial concession dates and a response rule after a signal (\crefrange{app:ex-response}{app:ex-fp}). Second, we complete that strategy at histories that occur with probability zero and verify optimality after every possible history (\crefrange{app:ex-swap}{app:ex-proof}).

\emph{Step 1: solve the reduced game.} A two-firm strategy has two economically relevant parts: an initial \emph{concession date} if no signal arrives, and a \emph{response date} after detecting the rival's stop. \cref{app:ex-response} shows that, after a signal at $a$, every optimal response lies between $a$ and $\chi(a)$. The rule $K$ assigns a distribution over that interval for each $a$, while $G$ is the distribution of the initial concession date. These pairs $(G,K)$ form a compact convex set. Payoffs vary continuously because stops are not observed immediately and detection dates have continuous distributions. A fixed-point theorem therefore gives a symmetric reduced equilibrium $(G^\ast,K^\ast)$: $G^\ast$ assigns probability only to optimal concession dates, and $K^\ast$ chooses an optimal response at almost every detection date that can occur.

\begin{figure}[H]\centering
\begin{tikzpicture}[x=1cm,y=1cm,
  lem/.style={draw, rounded corners=2pt, align=center, font=\footnotesize, inner sep=4pt, text=black},
  dep/.style={-{Stealth[length=2.4mm]}, semithick}]
\node[lem, text width=3.4cm] (env)   at (7.4,0) {payoffs have a common\\ integrable bound};
\node[lem, text width=5.4cm] (win) at (7.4,-2.5) {\cref{lem:ex-window}\\ every optimal response lies in $\mathcal T(a)=[a,\chi(a)]$; response payoffs vary continuously};
\node[lem, text width=3.3cm] (cont) at (1.8,-5.2)  {\cref{lem:ex-cont}\\ reduced payoffs vary continuously, including at ties};
\node[lem, text width=3.1cm] (meas) at (7.4,-5.2) {\cref{lem:ex-meas}\\ optimal responses can be chosen measurably};
\node[lem, text width=2.9cm] (bel)  at (11.7,-4.6)  {\cref{lem:ex-belief}\\ Bayes' rule after delayed signals};
\node[lem, text width=3.2cm] (down) at (1.8,-7.9) {\cref{lem:ex-down}\\ a rival's later stop hurts a stopped firm};
\node[lem, text width=3.6cm] (fp)   at (7.4,-7.9) {\cref{lem:ex-fp}\\ a fixed-point theorem gives a reduced equilibrium $(G^\ast,K^\ast)$};
\node[lem, text width=3.3cm] (swap) at (1.8,-10.5) {\cref{lem:ex-swap}\\ fill in responses at probability-zero dates};
\node[lem, text width=3.3cm] (one)  at (7.4,-10.5) {\cref{lem:ex-onestep}\\ initial optimality implies optimality after no signal};
\node[lem, text width=3.4cm] (memo) at (11.9,-10.5){\cref{lem:ex-memoryless}\\ after the latest possible concession, beliefs stop changing};
\node[lem, text width=5.2cm] (prop) at (7.4,-13) {\cref{prop:existence}\\ combine the plans and beliefs, then verify optimality after every history};
\draw[dep] (env)   -- (win);
\draw[dep] (win)   -- (cont);
\draw[dep] (win)   -- (meas);
\draw[dep] (bel)   -- (meas);
\draw[dep] ([xshift=1.2cm]win.south) .. controls (9.6,-4.6) and (9.5,-6.3) .. (fp.north east);
\draw[dep] (cont)  -- (fp);
\draw[dep] (meas)  -- (fp);
\draw[dep] (fp)    -- (swap);
\draw[dep] (down)  -- (swap);
\draw[dep] (bel)   -- (memo);
\draw[dep] (swap)  -- (prop);
\draw[dep] (one)   -- (prop);
\draw[dep] (memo)  -- (prop);
\end{tikzpicture}
\caption{Logical structure of the existence proof. Each arrow points from an input to the result that uses it.}
\label{fig:appB-map}
\end{figure}
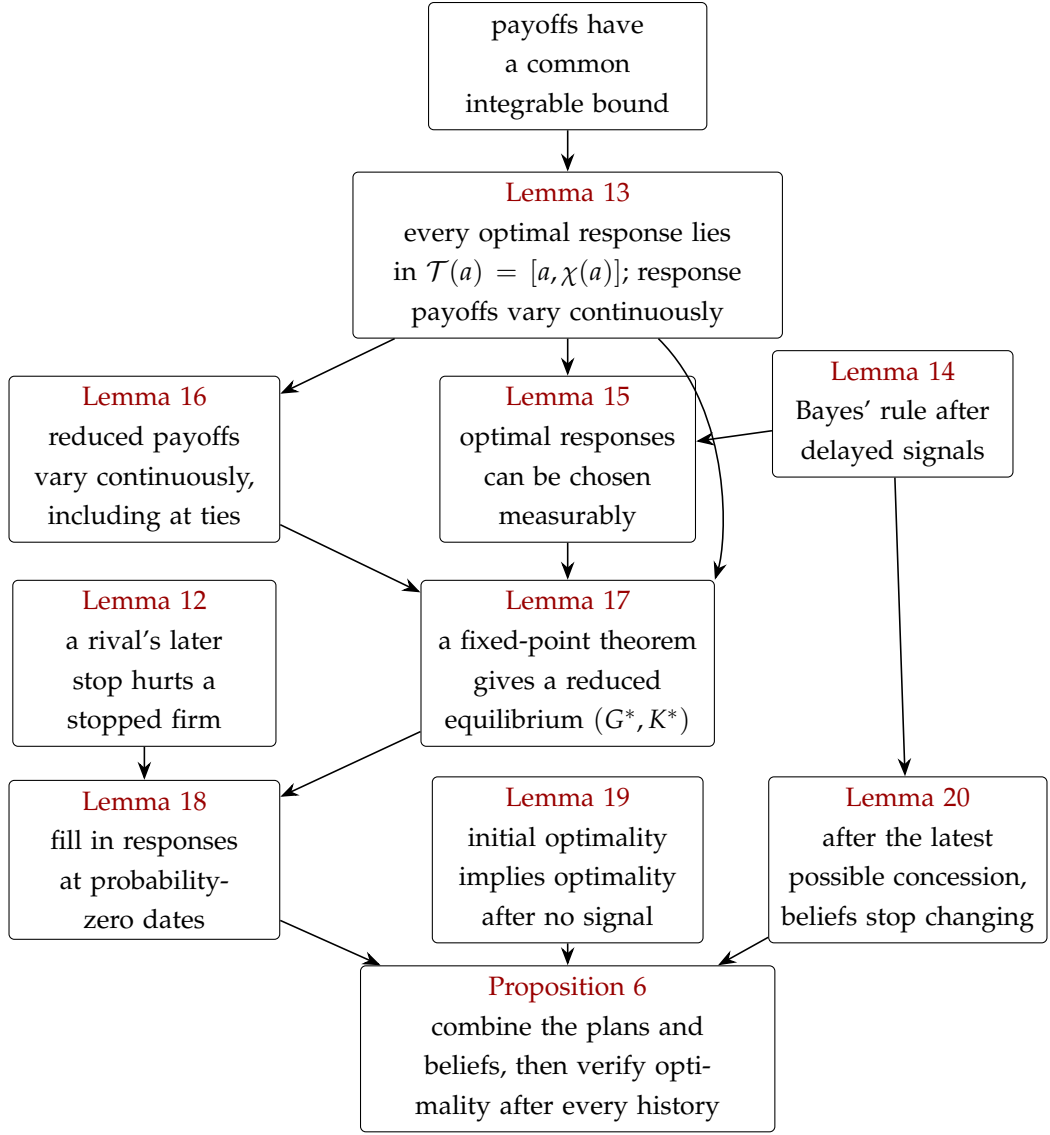

\emph{Step 2: complete the dynamic strategy.} The reduced equilibrium specifies behavior only at histories that matter with positive probability. Perfect Bayesian equilibrium also requires a plan and a belief at probability-zero histories. At such detection dates, we prescribe the latest potentially optimal response $\chi(a)$ and the belief that the rival stopped at $a$. This change does not disturb the optimality of $G^\ast$ (\cref{lem:ex-swap}). We then show that a concession date that is optimal initially remains optimal after any history with no signal (\cref{lem:ex-onestep}). After the latest date at which the rival might concede, the firm's belief no longer changes and its stopping problem is time-consistent (\cref{lem:ex-memoryless}). \cref{app:ex-proof} combines these pieces and checks every class of history.
\subsection{The response problem}\label{app:ex-response}

Note that \begin{equation}\label{eq:ex-envelope}
\alpha(y)\,\pi(y,0)\;\le\;e^{-ry}\,\sup_{v\le y}\Big\{e^{-\int_0^v\lambda(z)\de z}\,\pi(v,0)\Big\}\;\xrightarrow[y\to\infty]{}\;0 ,
\end{equation}
where the limit holds because the running supremum is nondecreasing in $y$. If the right side were at least $\e>0$ along $y_m\uparrow\infty$, we could pass to a subsequence with $y_{m+1}>y_m+1$. Monotonicity of the supremum would then give
\[
\int_{y_m}^{y_m+1}e^{-rs}\sup_{v\le s}
\Big\{e^{-\int_0^v\lambda(z)\de z}\,\pi(v,0)\Big\}\,\de s
\ge \e e^{-r}
\]
on each of these disjoint intervals, contradicting the integrability in \cref{ass:payoffs}(v). Because flow profits decrease in the rival's technology, $\pi(a,b)\leq \pi(a,0)$ for every $b\geq0$. Thus \eqref{eq:ex-envelope}, which bounds discounted profits against a rival fixed at $0$, also bounds discounted profits along every stopping history considered below.
\begin{lemma}[Rival's delay hurts]\label{lem:ex-down}
Fix $x\in\R_+$. The map $y\mapsto U(x,y)$ is continuous on $[x,+\infty]$ and strictly decreasing: for every $x\le y<y'\le+\infty$,
$
U(x,y')\;<\;U(x,y).
$\end{lemma}

\begin{proof}
Continuity is \cref{lem:cont}. For $y\in(x,\infty)$ we have $\tau_1\vee\tau_2=y$, so $\Lambda_s=\int_0^{s\wedge y}\lambda$ and
\[
U(x,y)=\int_0^y \alpha(s)\,\pi(s\wedge x,\,s)\,\de s+\frac{\alpha(y)\,\pi(x,y)}{r}.
\]
Differentiating and using $\alpha'(y)=-(r+\lambda(y))\alpha(y)$,
\begin{align*}
\frac{\partial}{\partial y}U(x,y)
&= \alpha(y)\pi(x,y)+\frac{\alpha'(y)\pi(x,y)+\alpha(y)\,D_j\pi(x,y)}{r}\\
&=\frac{\alpha(y)}{r}\Big[D_j\pi(x,y)-\lambda(y)\,\pi(x,y)\Big]<0,
\end{align*}
since $D_j\pi<0$, $\lambda\ge0$, and $\pi>0$. Strict monotonicity on $(x,\infty)$ follows; continuity extends the strict comparison to the endpoints $y=x$ and $y'=+\infty$.
\end{proof}

After a detection the game is a single-agent stopping problem against a frozen rival whose technology may be unknown. For $c\ge0$ define the \emph{crossing date}
\[
\chi(c):=\inf\big\{y\ge c:\ D_i\log\pi(y,c)\le\lambda(y)\big\}.
\]
If a signal arrives at $a$ and reveals a rival frozen at $t$, the optimal response is $a\vee\chi(t)$. For a detection date $a>0$ and a belief $\nu\in\Delta([0,a])$ over the rival's frozen technology, the value of stopping at $y\in[a,+\infty]$ is
\[
w(y;a,\nu):=\int_{[0,a]}\left[\int_a^y \frac{\alpha(s)}{\alpha(a)}\,\pi(s,c)\,\de s+\frac{\alpha(y)}{\alpha(a)}\,\frac{\pi(y,c)}{r}\right]\nu(\de c),
\]
with the terminal term absent at $y=+\infty$, and the \emph{response window} is
$
\mathcal{T}(a):=\big[\,a,\ \chi(a)\,\big].
$
From the definition of $w$, for every $a\le a'\le y$ (note $\Delta([0,a])\subseteq\Delta([0,a'])$),
\begin{equation}\label{eq:ex-split}
w(y;a,\nu)\ =\ \int_a^{a'}\frac{\alpha(s)}{\alpha(a)}\Big(\int\pi(s,c)\,\nu(\de c)\Big)\de s\ +\ \frac{\alpha(a')}{\alpha(a)}\,w(y;a',\nu);
\end{equation}
the first term does not depend on $y$, so the two continuation values differ only by an additive constant and a positive scale factor. They therefore rank every stopping date in $[a',+\infty]$ in the same way. The next lemma collects the properties of $\chi$, $w$, and $\mathcal{T}$ that we use.

\begin{lemma}[The response problem]\label{lem:ex-window}
Fix $a>0$ and $\nu\in\Delta([0,a])$. Then:
\begin{itemize}
\item[(i)] $\chi$ is nondecreasing and continuous on $\R_+$, with $c\le\chi(c)\le\tau^R$ and $\chi(c)>c$ for $c<\tau^R$, and $\chi(c)=c$ for $c\ge\tau^R$; consequently $\mathcal{T}(a)=\{a\}$ for $a\ge\tau^R$ and $\mathcal{T}(a)\subseteq[a,\tau^R]$ for $a\le\tau^R$.
\item[(ii)] For $a\le\tau^R$ and every $y\in[a,+\infty]$,
\begin{equation}\label{eq:ex-wbound}
0\ \le\ w(y;a,\nu)\ \le\ \bar w:=\frac{1}{\alpha(\tau^R)}\left[\int_0^\infty \alpha(s)\,\pi(s,0)\,\de s\ +\ \frac1r\,\sup_{v\ge0}\,\alpha(v)\,\pi(v,0)\right]\ <\ \infty .
\end{equation}
\item[(iii)] $w(\cdot\,;a,\nu)$ is continuous on $[a,+\infty]$ and strictly decreasing on $[\chi(a),+\infty]$; hence \[\argmax_{y\in[a,+\infty]}w(y;a,\nu)\] 
is a nonempty compact subset of $\mathcal{T}(a)$, and any response outside $\mathcal{T}(a)$ is strictly suboptimal under every belief in $\Delta([0,a])$.
\item[(iv)] For the point belief $\delta_c$ with $c\le a$, the unique maximizer is $a\vee\chi(c)$. In particular $\chi(a)$ is the unique optimal response under $\delta_a$, and it is the \emph{latest} response that is optimal under any belief in $\Delta([0,a])$.
\item[(v)] $w$ is jointly continuous in $(y,a,\nu)$, where beliefs with supports in varying intervals are compared as elements of $\Delta(\overline\R_+)$ with the weak topology, on the domain $\{(y,a,\nu):a>0,\ y\in[a,+\infty],\ \nu([0,a])=1\}$; consequently, by Berge's maximum theorem, $(a,\nu)\mapsto\max_{y}w(y;a,\nu)$ is continuous and $(a,\nu)\mapsto\argmax_y w(y;a,\nu)$ is upper hemicontinuous with compact values.
\end{itemize}
\end{lemma}

\begin{proof} 
\noindent \underline{\emph{Part (i).}} For $c\le\tau^R$, log-supermodularity (\cref{ass:payoffs}(iii)) and the definition of $\tau^R$ give $D_i\log\pi(\tau^R,c)\le D_i\log\pi(\tau^R,\tau^R)\le\lambda(\tau^R)$, so the defining set is nonempty and $\chi(c)\le\tau^R$. For $c\ge\tau^R$, \cref{ass:dsc} gives $D_i\log\pi(c,c)\le\lambda(c)$, so $\chi(c)=c$; for $c<\tau^R$, $D_i\log\pi(c,c)>\lambda(c)$ and continuity give $\chi(c)>c$. Monotonicity: for $c\le c'$, continuity at the crossing and log-supermodularity give $D_i\log\pi(\chi(c'),c)\le D_i\log\pi(\chi(c'),c')\le\lambda(\chi(c'))$, and $\chi(c')\ge c'\ge c$, so $\chi(c)\le\chi(c')$. Continuity: write $g(y,c):=D_i\log\pi(y,c)-\lambda(y)$, which is continuous, and strictly decreasing in $y$ on $[c,\infty)$ by \cref{ass:payoffs}(iv). If $\chi(c)>c$, fix $0<\e<(\chi(c)-c)/2$. Then $g(\chi(c)-\e,c)>0>g(\chi(c)+\e,c)$. For $c'$ sufficiently close to $c$, we have $c'<\chi(c)-\e$ and both strict inequalities persist, so $\chi(c')\in(\chi(c)-\e,\chi(c)+\e)$. If $\chi(c)=c$ then $c\ge\tau^R$; on $[\tau^R,\infty)$ the map $\chi$ coincides with the identity, and continuity at the junction $c=\tau^R$ follows from the squeeze $c'\le\chi(c')\le\tau^R$ for $c'<\tau^R$, which forces $\chi(c')\to\tau^R=\chi(\tau^R)$ as $c'\uparrow\tau^R$. The consequences for $\mathcal{T}$ are immediate.

\noindent \underline{\emph{Part (ii).}} Use $\pi(\cdot,c)\le\pi(\cdot,0)$ and $1/\alpha(a)\le1/\alpha(\tau^R)$ for $a\le\tau^R$: the integral is finite by \cref{ass:payoffs}(v), whose integrand dominates $\alpha(s)\,\pi(s,0)$, and the supremum by \eqref{eq:ex-envelope}.

\noindent \underline{\emph{Part (iii).}} For a fixed rival technology $c\le a$, the contribution to the derivative of $w$ from that technology is
\[
\frac{\alpha(y)}{r\,\alpha(a)}\,\pi(y,c)\Big[D_i\log\pi(y,c)-\lambda(y)\Big].
\]
For $y\ge\chi(a)$, log-supermodularity and the definition of $\chi(a)$ give
\[
D_i\log\pi(y,c)\le D_i\log\pi(y,a)\le\lambda(y),
\]
with the second inequality strict for $y>\chi(a)$. Integrating against $\nu$, the derivative of $w$ is strictly negative on $(\chi(a),\infty)$, which yields the strict decrease and the inclusion $\argmax\subseteq\mathcal{T}(a)$. Continuity of $w$ up to $y=+\infty$ follows from dominated convergence: the flow is dominated using \eqref{eq:ex-envelope} and the terminal term satisfies $\alpha(y)\pi(y,c)\le \alpha(y)\pi(y,0)\to0$. Nonemptiness and compactness of the argmax follow from continuity on the compact $[a,+\infty]$ and closedness of the argmax.

\noindent \underline{\emph{Part (iv).}} Under the point belief $\delta_c$, the derivative displayed in Part (iii) has the sign of $D_i\log\pi(y,c)-\lambda(y)$, which is strictly decreasing in $y$ by \cref{ass:payoffs}(iv). The point-belief payoff is therefore single-peaked with unique maximizer $a\vee\chi(c)$. Part (iii) shows that no response later than $\chi(a)$ is optimal under any belief, while $\chi(a)$ is uniquely optimal under $\delta_a$. Hence $\chi(a)$ is the latest response that can be optimal.

\noindent \underline{\emph{Part (v).}} Let $b(y,a,c)$ denote the payoff term inside the integral defining $w(y;a,\nu)$:
\[
b(y,a,c)
:=
\int_a^y \frac{\alpha(s)}{\alpha(a)}\,\pi(s,c)\,\de s
+\frac{\alpha(y)}{\alpha(a)}\,\frac{\pi(y,c)}{r},
\]
with the terminal term omitted when $y=+\infty$. The function $b$ is jointly continuous in $(y,a,c)$ and dominated by an integrable envelope via \eqref{eq:ex-envelope}, so $w$ is jointly continuous in $(y,a)$ uniformly over $c$ in compact sets and affine-continuous in $\nu$. Joint continuity follows: let $(y_m,a_m,\nu_m)\to(y,a,\nu)$ with $\nu_m([0,a_m])=1$ and set $A:=\sup_m a_m<\infty$. Extend $b$ beyond the relevant supports by replacing it with $b(y,a,c\wedge A)$, a bounded function that is continuous in $c$ on $\overline\R_+$. Then
\begin{align*}
\big|w(y_m;a_m,\nu_m)-w(y;a,\nu)\big|
&\le \sup_{c\le A}\big|b(y_m,a_m,c)-b(y,a,c)\big|\\
&\quad+\Big|\int b(y,a,c\wedge A)\,\de(\nu_m-\nu)(c)\Big|,
\end{align*}
where the first term vanishes by the uniform convergence just noted and the second by weak convergence. The domain correspondence $a\mapsto[a,+\infty]$ is continuous: upper hemicontinuity is closedness of $\{(a,y):y\ge a\}$ in the compact $\overline\R_+^{\,2}$, and lower hemicontinuity holds because any $y\in[a,+\infty]$ is the limit of $y\vee a_m\in[a_m,+\infty]$ along $a_m\to a$. Berge's maximum theorem \citep[Theorem~17.31]{aliprantis2006infinite} therefore applies.
\end{proof}

\paragraph{Response rules.} We describe every date in the response interval $[a,\chi(a)]$ by a number $\theta\in[0,1]$:
$
\chi_\theta(a):=(1-\theta)\,a+\theta\,\chi(a).$ Thus $\theta=0$ means stop when the signal arrives, while $\theta=1$ means wait until $\chi(a)$, the latest potentially optimal response. For $a\ge\tau^R$, both dates equal $a$. A \emph{response kernel} $K$ is simply a measurable rule that assigns a probability distribution over $\theta$ to each detection date $a$. In the terminology of \citet[Theorem~19.7 and Definition~19.8]{aliprantis2006infinite}, it is a measurable transition from detection dates to probability distributions over responses. Upon detection, the firm draws $\theta$ from $K(a,\cdot)$ and stops at $\chi_\theta(a)$. By \cref{lem:ex-window}(iii), no optimal response lies outside this interval.

To define convergence of response rules, we associate each kernel \(K\) with a joint distribution over detection dates and responses. If $\tau^R=0$, the response parameter is immaterial at every date, so we take the response-rule space to be a singleton. Suppose $\tau^R>0$, let $A:=[0,\tau^R]$, and associate with each kernel $K$ the probability measure
$
Q_K(\de a,\de\theta):=\frac{\de a \, K(a,\de\theta)}{\tau^R}$ on$A\times[0,1].$  The first marginal of $Q_K$ is the uniform distribution on $A$. Conversely, disintegration recovers from every joint law with this marginal a measurable kernel $K(a,\de\theta)$, unique for almost every $a$. We therefore treat two kernels as the same when they agree almost everywhere and let $\mathcal K$ denote the resulting set. We say that $K_m\to K$ when $Q_{K_m}\to Q_K$ weakly. The set $\mathcal K$ is nonempty and convex. It is also compact and metrizable. Fix a bounded Borel function $f$ on $A \times [0,1]$ that is continuous except on a null set $D \subseteq A$. Note that since $Q_K$ has the uniform distribution on $A$, it  assigns zero probability to $D\times[0,1]$. Hence if $Q_{K_m}$ converges weakly to $Q_K$ then the Portmaneau lemma gives: 
\[
\int_A\!\!\int_0^1 f(a,\theta)K_m(a,\de\theta)\,\de a
=\tau^R\!\int f\,\de Q_{K_m}
\ \longrightarrow\
\tau^R\!\int f\,\de Q_K
=\int_A\!\!\int_0^1 f(a,\theta)K(a,\de\theta)\,\de a
\tag{J}\label{eq:joint-conv}
\]


\subsection{The reduced game}\label{app:ex-reduced}

Let $\mathcal G:=\Delta(\overline\R_+)$ be the set of probability distributions over concession dates, including the possibility $+\infty$. We use weak convergence on this set, which makes $\mathcal G$ compact, convex, and metrizable. A \emph{reduced strategy} is a pair $(G,K)$: draw the no-signal concession date from $G$, and use response rule $K$ after the first detection. By \cref{lem:rep,app:noisy}, fresh independent uniform draws implement every such pair in the dynamic game.

We first record the beliefs that Bayes' rule assigns against a rival playing a reduced strategy. For $G\in\mathcal G$ set
\[
N_G(a):=\int_{[0,a]}e^{\omega y}\,G(\de y),
\qquad
\nu^G_a(\de y):=\frac{e^{\omega y}\,\mathbf 1\{y\le a\}\,G(\de y)}{N_G(a)}
\quad\text{whenever }N_G(a)>0 .
\]

\begin{lemma}[Beliefs]\label{lem:ex-belief}
Fix $\{i,j\}=\{1,2\}$ and a profile in which firm $j$ plays the reduced strategy $(G,K)$, and let $X_j\sim G$ denote its no-signal concession date. Consider firm $i$ at a history at which it is active and has received no signal by $t$. Then:
\begin{itemize}
\item[(i)] any stop by firm $j$ so far occurred at $X_j$; the conditional law of $X_j$ given the history assigns mass proportional to $e^{-\omega(t-y)}G(\de y)$ on $[0,t]$ and to $G(\de y)$ on $(t,+\infty]$;
\item[(ii)] the finite part of the date $a$ of firm $i$'s first signal has the possibly defective density $a\mapsto\omega e^{-\omega a}N_G(a)$ on $(0,\infty)$, and the signal date has an atom of size $G(\{+\infty\})$ at $+\infty$ (this law does not depend on firm $i$'s own concession plan, so long as firm $i$ is active); given a finite first signal at $a$ with $N_G(a)>0$, the conditional law of $X_j$ is $\nu^G_a$;
\item[(iii)] both conditional laws are unaffected by conditioning on survival of the disaster, and they do not depend on firm $i$'s own (past or planned) concession behavior.
\end{itemize}
\end{lemma}

\begin{proof}
\noindent \underline{\emph{Part (i).}} Because firm $i$ has not stopped, firm $j$ has received no signal and therefore stops, if at all, at its no-signal concession date $X_j$. Given $X_j=y$, the probability that firm $i$ receives no signal by $t$ is $e^{-\omega(t-y)}$ when $y\le t$ and is one when $y>t$. While firm $i$ is active, its technology at each date $s$ equals $s$, so a rival stop at $y\le s$ leaves the frontier at $\max(s,y)=s$. The survival probability through $t$ is therefore the common factor $e^{-\int_0^t\lambda(s)\de s}$ for every realization of $X_j$. Bayes' rule gives the conditional law in (i).

\noindent \underline{\emph{Part (ii).}} Given $X_j=y<+\infty$, the finite first-signal date has density $
a\longmapsto\omega e^{-\omega(a-y)}\mathbf 1\{a>y\}.$ 
Integrating this density over $y\in[0,a]$ gives $\omega e^{-\omega a}N_G(a)$. If $X_j=+\infty$, no signal arrives, which gives the atom of size $G(\{+\infty\})$ at $+\infty$. At a finite signal date $a$, Bayes' rule assigns weight proportional to $e^{\omega y}G(\de y)$ on $[0,a]$, yielding $\nu^G_a$.

\noindent \underline{\emph{Part (iii).}} The common survival factor calculated in Part (i) proves that conditioning on survival does not change either posterior. Moreover, $X_j$ and its detection clock are independent of firm $i$'s private draws and concession plan. Hence firm $i$'s own past or planned behavior does not change either conditional law.
\end{proof}

The fixed-point proof must choose an optimal response as a measurable function of the detection date. The next lemma proves that this is possible. For $a$ with $N_G(a)>0$ define
\[
\Theta^\ast(a,G):=\big\{\theta\in[0,1]:\ \chi_\theta(a)\in\argmax_{y\in[a,+\infty]}w\big(y;a,\nu^G_a\big)\big\}.
\]

\begin{lemma}[Measurability in the detection date]\label{lem:ex-meas}
Fix $G\in\mathcal G$. Then:
\begin{itemize}
\item[(i)] $N_G$ is nondecreasing and right-continuous, and on $\{N_G>0\}$ the map $a\mapsto\nu^G_a$ is Borel for the weak topology on $\Delta(\overline\R_+)$;
\item[(ii)] on $\big(\{N_G>0\}\cap(0,\tau^R)\big)\times[0,1]$, the map $(a,\theta)\mapsto w\big(\chi_\theta(a);a,\nu^G_a\big)$ is Borel in $a$ and continuous in $\theta$ (a Carath\'eodory integrand), and
$
a\ \mapsto\ \max_{y\in[a,+\infty]}w\big(y;a,\nu^G_a\big)\ =\ \max_{\theta\in[0,1]}w\big(\chi_\theta(a);a,\nu^G_a\big)$ is Borel;
\item[(iii)] $\Theta^\ast(\cdot,G)$ is nonempty- and compact-valued, and one can choose an optimal $\theta^\ast(a)$ as a Borel function of $a$.
\end{itemize}
\end{lemma}

\begin{proof} 
\noindent \underline{\emph{Part (i).}} Monotonicity is clear, and right-continuity is dominated convergence: as $a_m\downarrow a$, $e^{\omega y}\mathbf 1\{y\le a_m\}\downarrow e^{\omega y}\mathbf 1\{y\le a\}$ pointwise, dominated by $e^{\omega\sup_m a_m}$. For bounded continuous $f$, $a\mapsto\int f\,\de\nu^G_a=N_G(a)^{-1}\int_{[0,a]}f(y)\,e^{\omega y}\,G(\de y)$ is then a ratio of Borel functions of $a$, hence Borel; and Borel measurability against every bounded continuous $f$ is Borel measurability for the weak topology.

\noindent \underline{\emph{Part (ii).}} The map $a\mapsto(a,\nu^G_a)$ is Borel by (i), and $(a,\theta)\mapsto\chi_\theta(a)$ is jointly continuous (\cref{lem:ex-window}(i)); composing with the jointly continuous $w$ (\cref{lem:ex-window}(v)) gives the Carath\'eodory property. The displayed identity holds because $\argmax_y w\subseteq\mathcal T(a)=\{\chi_\theta(a):\theta\in[0,1]\}$ (\cref{lem:ex-window}(iii)), and the maximum is then Borel as the value function of a maximization of a Carath\'eodory function over the fixed compact $[0,1]$ \citep[Measurable Maximum Theorem~18.19]{aliprantis2006infinite}.

\noindent \underline{\emph{Part (iii).}} By the displayed identity, $\Theta^\ast(a,G)=\argmax_{\theta\in[0,1]}w(\chi_\theta(a);a,\nu^G_a)$. The same theorem shows that this set is nonempty, compact, and measurable in $a$, and that an optimal $\theta^\ast(a)$ can be chosen as a Borel function.
\end{proof}

\paragraph{Outcomes and payoffs.} Fix concession dates $x$ for firm 1 and $y$ for firm 2. Use the detection clocks $E_1,E_2$ from \cref{app:noisy}, where $E_i$ is attached to firm $i$'s stop, and let $\theta_1,\theta_2$ be the response draws. For a detection at $a\ge\tau^R$, the response parameter is immaterial because $\chi_\theta(a)=a$ for every $\theta$. The following three cases compare which firm plans to concede first and whether the resulting signal arrives before the other firm's planned concession. The realized stopping pair is:
\begin{itemize}
\item if $x<y$: $\tau_1=x$; with $a_2:=x+E_1$, $\tau_2=y$ if $y<a_2$, and $\tau_2=\chi_{\theta_2}(a_2)$ if $a_2\le y$;
\item if $y<x$: $\tau_2=y$; with $a_1:=y+E_2$, $\tau_1=x$ if $x<a_1$, and $\tau_1=\chi_{\theta_1}(a_1)$ if $a_1\le x$;
\item if $x=y$: $\tau_1=\tau_2=x$.
\end{itemize}
The first stopper concedes without receiving a signal---its rival stops strictly later, so no signal precedes its concession---and the later firm either concedes on schedule or, if the signal arrives first, responds optimally. Note that the law $Q_K$ determines $K(a,d\theta)$ for almost every date $a$, and we can  choose its values arbitrarily at remaining dates. Firm 1's reduced payoff against a rival playing $(G,K)$, when firm 1 concedes at $x$ and responds per the same kernel $K$, is
\[
U_K(x,y):=\Ex\big[U(\tau_1,\tau_2)\big],
\qquad
U(x;G,K):=\int_{\overline\R_+}U_K(x,y)\,G(\de y),
\]
the expectation running over $(E_1,E_2,\theta_1,\theta_2)$ with $\theta_i\sim K(a_i,\cdot)$. Note that $\theta_1$ enters only on $\{y<x\}$ and $\theta_2$ only on $\{x<y\}$, so $U_K(x,y)$ is \emph{affine} in $K$: averaging two response rules averages their payoffs. Since all concession draws, response draws, and detection clocks are independent, $U(x;G,K)$ equals firm 1's expected payoff in the dynamic game under the corresponding plans (\cref{lem:rep,lem:cont}).

The reduced game treats concession and response choices separately. A response rule that is optimal at almost every detection date remains optimal after any concession deviation: detection dates have atomless laws, and their posteriors do not depend on the firm's concession plan (\cref{lem:ex-belief}(ii)--(iii)). Thus $U(x;G,K)$ can hold the response rule $K$ fixed while varying the concession date $x$. Correlating the response draw with the concession draw also cannot help: after a detection, the abandoned concession draw no longer affects payoffs, and a fresh response draw attains the pointwise optimum.

\begin{lemma}[Continuity of reduced payoffs]\label{lem:ex-cont}
\begin{itemize}
\item[(i)] The family $\{U_K:K\in\mathcal K\}$ is uniformly bounded by $\bar U$ and uniformly equicontinuous on $\overline\R_+^{\,2}$.
\item[(ii)] For each fixed $(x,y)$, the map $K\mapsto U_K(x,y)$ is continuous on $\mathcal K$.
\item[(iii)] $U$ is jointly continuous on $\overline\R_+\times\mathcal G\times\mathcal K$.
\end{itemize}
\end{lemma}

\begin{proof} \noindent \underline{\emph{Part (i).}} Boundedness is immediate. For equicontinuity, fix $K$ and consider the branch $y<x$ (the branch $x<y$ is symmetric; the diagonal is treated at the end). Substituting $a=y+E_2$,
\[
U_K(x,y)\big|_{y<x}
=\underbrace{e^{-\omega(x-y)}\,U(x,y)}_{\text{clock rings after }x}
+\int_y^{x}\Big[\int_0^1U\big(\chi_{\theta_1}(a),\,y\big)K(a,\de\theta_1)\Big]\,\omega e^{-\omega(a-y)}\,\de a ,
\]
using $\chi_\theta(a)=a$ on $[\tau^R,\infty)$. Perturbing $x$ to $x'>x$ changes the first term by at most $\eta_U(d(x,x'))+\bar U\,\big(e^{-\omega(x-y)}-e^{-\omega(x'-y)}\big)$ and adds to the integral a piece over $(x,x')$ of mass $e^{-\omega(x-y)}-e^{-\omega(x'-y)}$; since
\[
e^{-\omega(x-y)}-e^{-\omega(x'-y)}\le 1-e^{-\omega(x'-x)}
\quad\text{and}\quad
e^{-\omega(x-y)}-e^{-\omega(x'-y)}\le e^{-\omega(x-y)} ,
\]
the change is bounded by $\eta_U(d(x,x'))+2\bar U\min\{1-e^{-\omega(x'-x)},\,e^{-\omega(x-y)}\}$. The minimum is small uniformly: if $x\le L$ then $x'-x\le e^{L+1}d(x,x')$ for $d(x,x')$ small, making the first entry small, while if $x>L$ and $y\le x-\sqrt L$ the second entry is at most $e^{-\omega\sqrt L}$; and if $x>L$, $y>x-\sqrt L$, then all realized coordinates under either concession date exceed $L-\sqrt L$, so the two payoffs lie within $2\eta_U\big(2e^{-(L-\sqrt L)}\big)$ of one another. Choosing $L$ large and then $d(x,x')$ small yields a modulus independent of $(y,K)$. 

Now we compare the two integrals at the same detection date. For $z<x$, write $p_z(a):=\omega e^{-\omega(a-z)}$. The inner expectation is bounded by $\bar U$ and, uniformly in $(a,K)$,
\[
\left|
\int_0^1 U\bigl(\chi_\theta(a),y'\bigr)K(a,\de\theta)
-
\int_0^1 U\bigl(\chi_\theta(a),y\bigr)K(a,\de\theta)
\right|
\le\eta_U\bigl(d(y,y')\bigr).
\]
Let $y'=y+\delta<x$. On the common interval $[y',x]$ we have
$p_{y'}(a)=e^{\omega\delta}p_y(a)$. The probability removed from
$[y,y']$ is $1-e^{-\omega\delta}$, and
$
\int_{y'}^x\bigl(p_{y'}(a)-p_y(a)\bigr)\de a
\le 1-e^{-\omega\delta}.
$
It follows that the detection terms differ by at most
$
\eta_U\bigl(d(y,y')\bigr)
+2\bar U\bigl(1-e^{-\omega\delta}\bigr).$ The no-detection terms differ by at most
$\eta_U\bigl(d(y,y')\bigr)
+\bar U\bigl(1-e^{-\omega\delta}\bigr).$ 
Hence, whenever $y<y'<x$,
\[
\big|U_K(x,y')-U_K(x,y)\big|
\le
2\eta_U\bigl(d(y,y')\bigr)
+3\bar U\bigl(1-e^{-\omega\delta}\bigr),
\]
and the case $y'<y<x$ follows by interchanging $y$ and $y'$.

Uniformity in $d(y,y')$ follows by considering two cases: if $y\le L$, then $\delta\le e^{L+1}d(y,y')$ for $d(y,y')$ small, making every $\delta$-dependent term above small; if $y>L$, then every realized coordinate under either concession date exceeds $L$, so the two payoffs lie within $2\eta_U\big(2e^{-L}\big)$ of one another. Choosing $L$ large and then $d(y,y')$ small yields a modulus independent of $(x,K)$. Perturbations across the diagonal $x=y$ compose the two branch moduli with the observation that, as $x\to y$ from either side, both branch formulas converge to $U(y,y)$ at rate $\eta_U(d(x,y))+2\bar U\big(1-e^{-\omega|x-y|}\big)$, again uniformized by the two-case split: from below, the integral term has vanishing mass and the first term tends to $U(y,y)$; from above, symmetrically.

\noindent \underline{\emph{Part (ii).}} In the displayed formula the kernel appears only through
\[
\int_0^{\tau^R}\!\!\int_0^1 f(a,\theta)\,K(a,\de\theta)\,\de a \quad \text{and} \quad 
f(a,\theta)=U(\chi_\theta(a),y)\,\omega e^{-\omega(a-y)}\,\mathbf 1\{y<a<x\},
\]
a bounded measurable function continuous in $\theta$; recall $(a,\theta)\mapsto\chi_\theta(a)$ and $U$ are continuous. If $x<y$, the kernel-dependent term instead has integrand
\[
U\bigl(x,\chi_\theta(a)\bigr)\,
\omega e^{-\omega(a-x)}\,
\mathbf 1\{x<a<y\}
\]
which is also bounded and continuous except possibly at the boundary dates $a=x$ and $a=y$. Since these dates form a Lebesgue-null set, \eqref{eq:joint-conv} implies that the kernel-dependent integral varies continuously with $K$. The remaining no-detection term does not depend on $K$. If $x=y$, the firms stop simultaneously and the payoff is independent of $K$. Thus $K\mapsto U_K(x,y)$ is continuous for every fixed $(x,y)$.

\noindent \underline{\emph{Part (iii).}} Let $(x_m,G_m,K_m)\to(x,G,K)$ and estimate
\begin{align*}
\big|U(x_m;G_m,K_m)-U(x;G,K)\big|
&\le 
\underbrace{\sup_{y}\big|U_{K_m}(x_m,y)-U_{K_m}(x,y)\big|}_{\to\,0\ \text{by (i)}} \\ 
& +\Big|\!\int U_{K_m}(x,y)\,\de(G_m-G)\Big|+\underbrace{\Big|\!\int\big[U_{K_m}-U_{K}\big](x,y)\,G(\de y)\Big|}_{\to\,0\ \text{by (ii) and domination}} .
\end{align*} 
For the middle term, suppose along a subsequence it stays above $\e>0$; by (i) and Arzel\`a--Ascoli we may pass to a further subsequence with $U_{K_m}(x,\cdot)\to h$ uniformly on $\overline\R_+$. Then
\[
\Big|\int U_{K_m}(x,y)\,\de(G_m-G)(y)\Big|
\le
2\big\|U_{K_m}(x,\cdot)-h\big\|_\infty
+\Big|\int h\,\de(G_m-G)\Big|
\ \longrightarrow\ 0,
\]
a contradiction.
\end{proof}

\subsection{A symmetric equilibrium of the reduced game}\label{app:ex-fp}

\begin{lemma}[Reduced equilibrium]\label{lem:ex-fp}
There exists $(G^\ast,K^\ast)\in\mathcal G\times\mathcal K$ such that
\begin{itemize}
\item[(i)] $G^\ast\big(\argmax_{x\in\overline\R_+}U(x;G^\ast,K^\ast)\big)=1$; equivalently, the support of $G^\ast$ is contained in $\argmax_x U(x;G^\ast,K^\ast)$;
\item[(ii)] for Lebesgue-almost every $a\in(0,\tau^R)$ with $N_{G^\ast}(a)>0$,
\[
K^\ast\Big(a,\ \big\{\theta\in[0,1]:\ \chi_\theta(a)\in\argmax_{y\in[a,+\infty]}w\big(y;a,\nu^{G^\ast}_a\big)\big\}\Big)=1 .
\]
\end{itemize}
\end{lemma}

\begin{proof}
Define, for $(G,K)\in\mathcal G\times\mathcal K$,
\begin{align*}
\mathrm{BR}_G(G,K)&:=\Big\{G'\in\mathcal G:\ G'\big(\argmax_x U(x;G,K)\big)=1\Big\},\\
\mathrm{BR}_K(G)&:=\Big\{K'\in\mathcal K:\ K'\big(a,\Theta^\ast(a,G)\big)=1\ \text{for a.e.\ }a\in(0,\tau^R)\text{ with }N_G(a)>0\Big\},
\end{align*}
with $\Theta^\ast(a,G)$ as defined before \cref{lem:ex-meas}. We verify the hypotheses of the Kakutani--Fan--Glicksberg theorem \citep[Corollary~17.55]{aliprantis2006infinite} on the nonempty, compact, convex set $\mathcal G\times\mathcal K$ (a subset of the product of two locally convex spaces).

The fixed-point theorem requires four properties: best replies must exist; each set of best replies must be convex and compact; and best replies must vary continuously in the closed-graph sense. We check these properties in turn.

\emph{Best replies exist.} $\argmax_x U(\cdot;G,K)$ is nonempty and closed by \cref{lem:ex-cont}(iii) and compactness of $\overline\R_+$, so $\mathrm{BR}_G\neq\emptyset$. For $\mathrm{BR}_K$, \cref{lem:ex-meas}(iii) lets us choose an optimal response $\theta^\ast(a)$ as a Borel function of the detection date. The rule $K'(a,\cdot)=\delta_{\theta^\ast(a)}$ (extended arbitrarily where $N_G(a)=0$) therefore lies in $\mathrm{BR}_K(G)$.

\emph{Best-reply sets are convex and compact.} Both sets require a distribution to put full probability on a fixed set of optimal choices. This condition is linear in the distribution, so each best-reply set is convex. The surrounding strategy spaces are compact. The closed-graph arguments below imply that each best-reply set is closed and hence compact.

\emph{Concession best replies vary continuously.} $G'\in\mathrm{BR}_G(G,K)$ iff $\int U(x;G,K)G'(\de x)\ge\max_x U(x;G,K)$. Both sides are continuous in $(G',G,K)$: the left side because $(G',G,K)\mapsto\int U\,\de G'$ inherits continuity from \cref{lem:ex-cont}(iii) (uniform equicontinuity in $x$ makes the integrand vary continuously in the sup norm, so the same Arzel\`a--Ascoli argument as in the proof of \cref{lem:ex-cont}(iii) applies), the right side by Berge. The inequality therefore defines a closed set.

\emph{Response best replies vary continuously.} If $\tau^R = 0$ then $\mathcal{K}$ is a singleton and there is nothing to prove. Suppose $\tau^R > 0$. Define the \emph{optimality gap}---the average payoff loss from using $K'$ instead of an optimal response---by
\[
\Gamma(G,K'):=\int_0^{\tau^R}\big(N_G(a)\wedge1\big)\Big[\max_{y\in[a,+\infty]}w\big(y;a,\nu^G_a\big)-\int_0^1 w\big(\chi_\theta(a);a,\nu^G_a\big)K'(a,\de\theta)\Big]e^{-a}\,\de a .
\]
Where $N_G(a)=0$, set the payoff-loss term inside the integral---the best-response payoff minus the payoff under $K'$---equal to $0$. This payoff loss is nonnegative, so $\Gamma\ge0$, and $K'\in\mathrm{BR}_K(G)$ iff $\Gamma(G,K')=0$. Let $(G_m,K_m)\to(G,K)$ with $K_m\in\mathrm{BR}_K(G_m)$; we show $\Gamma(G,K)=0$. First, for all but countably many $a$ (the atoms of $G$), $N_{G_m}(a)\to N_G(a)$, and where $N_G(a)>0$ also $\nu^{G_m}_a\to\nu^G_a$ weakly (the integrand $y\mapsto e^{\omega y}\mathbf 1\{y\le a\}$ is bounded and $G$-a.s.\ continuous at non-atoms $a$). Hence, writing $I(a,\theta;G'):=(N_{G'}(a)\wedge1)\,w(\chi_\theta(a);a,\nu^{G'}_a)$ and $W(a;G'):=(N_{G'}(a)\wedge1)\max_y w(y;a,\nu^{G'}_a)$ (both $0$ where $N_{G'}(a)=0$), \cref{lem:ex-window}(v) gives, for a.e.\ $a$,
$
W(a;G_m)\to W(a;G)$ and $\sup_{\theta\in[0,1]}\big|I(a,\theta;G_m)-I(a,\theta;G)\big|\to0$
(where $N_G(a)=0$, the factor $N_{G_m}(a)\wedge1\to0$ forces both convergences because the underlying response payoffs are bounded by $\bar w$ from  \eqref{eq:ex-wbound}). Since the $\theta$-suprema bound the kernel integrals uniformly over $K_m$, all terms lying in $[0,\bar w]$, dominated convergence yields
\begin{align*}
\Gamma(G_m,K_m)&=\widetilde\Gamma_m+o(1),\\
\widetilde\Gamma_m
&:=\int_0^{\tau^R}
\Big[W(a;G)-\int_0^1 I(a,\theta;G)\,K_m(a,\de\theta)\Big]e^{-a}\de a .
\end{align*}
To apply \eqref{eq:joint-conv}, we verify that the integrand is continuous outside atoms of $G$. At any non-atom $a$ with $N_G(a)>0$, both $N_G(a)$ and the posterior $\nu_a^G$ vary continuously with $a$, so $I(a,\theta;G)e^{-a}$ is continuous in $(a,\theta)$. If $N_G(a)=0$, then $N_G(a')\wedge1\to0$ as $a'\to a$; because response payoffs are uniformly bounded, $I(a',\theta;G)e^{-a'}\to0$ uniformly in $\theta$. Thus the only possible discontinuities occur at atoms of $G$. These atoms are countable so the discontinuity set has $Q_K$-measure zero. Since $Q_{K_m}\to Q_K$ weakly, \eqref{eq:joint-conv} gives
$
\widetilde\Gamma_m\longrightarrow\Gamma(G,K)$ so $\Gamma(G,K)=\lim_m\Gamma(G_m,K_m)=0$. Kakutani--Fan--Glicksberg now delivers a fixed point $(G^\ast,K^\ast) \in BR_G(G^\ast, K^\ast) \times BR_K(G^\ast)$, and since $\argmax_x U$ is closed, (i) is equivalent to the support inclusion.
\end{proof}

\subsection{\texorpdfstring{Completing behavior at probability-zero histories}{Completing behavior at probability-zero histories}}\label{app:ex-swap}

Fix $(G^\ast,K^\ast)$ from \cref{lem:ex-fp}. When evaluating firm 1's payoff, let $X_2\sim G^\ast$ denote firm 2's concession date. Let $t_0:=\inf\supp G^\ast\in[0,+\infty]$, so that $N_{G^\ast}(a)>0$ for every $a>t_0$ and $N_{G^\ast}(a)=0$ for every $a<t_0$. Let $Z\subseteq(0,\tau^R)$ denote the (Lebesgue-null, measurable) set of dates $a$ with $N_{G^\ast}(a)>0$ at which the optimality property of \cref{lem:ex-fp}(ii) fails. By \cref{lem:ex-meas}(iii), we can choose a Borel function $\theta^\ast(a)\in\Theta^\ast(a,G^\ast)$ at every date with $N_{G^\ast}(a)>0$. Define the modified kernel
\[
\hat K(a,\cdot):=
\begin{cases}
K^\ast(a,\cdot) & \text{if }N_{G^\ast}(a)>0\text{ and }a\notin Z,\\[2pt]
\delta_{\theta^\ast(a)} &
\text{if }N_{G^\ast}(a)>0\text{ and }a\in Z,\\[2pt]
{\delta_{1}} & {\text{if }N_{G^\ast}(a)=0.}
\end{cases}
\]
At every feasible date in $Z$, the new response is optimal under the Bayes posterior $\nu_a^{G^\ast}$. When $N_{G^\ast}(a)=0$, the detection history is impossible; there the response is $\chi_1(a)=\chi(a)$, which is uniquely optimal under the belief $\delta_a$ that the rival stopped at $a$ (\cref{lem:ex-window}(iv)).

\begin{lemma}[Responses at probability-zero dates]\label{lem:ex-swap}
\begin{itemize}
\item[(i)] $U(x;G^\ast,\hat K)=U(x;G^\ast,K^\ast)$ for every $x\in\overline\R_+$ with $x\ge t_0$;
\item[(ii)] $U(x;G^\ast,\hat K)\le U(x;G^\ast,K^\ast)$ for every $x<t_0$.
\end{itemize}
Consequently
\[
\max_x U(x;G^\ast,\hat K)=\max_x U(x;G^\ast,K^\ast)
\qquad\text{and}\qquad
\supp G^\ast\subseteq\argmax_x U(x;G^\ast,\hat K).
\]
\end{lemma}

\begin{proof}
The kernels $\hat K$ and $K^\ast$ differ only on the exceptional set $Z\cup\{a:N_{G^\ast}(a)=0\}$, and $\{a:N_{G^\ast}(a)=0\}\subseteq(0,t_0]$.

\noindent \underline{\emph{Part (i).}} Fix $x\ge t_0$. If $x=X_2$, both firms stop at $x$ and no response kernel is used; this is the only branch when $t_0=+\infty$. In the branch $x<X_2$ the rival's response draw is taken at $a_2=x+E_1$, whose conditional law is absolutely continuous and supported on $(x,\infty)\subseteq(t_0,\infty)$, on which $N_{G^\ast}>0$; in the branch $X_2<x$ firm 1's own response draw is taken at $a_1=X_2+E_2$, with $X_2\ge t_0$ almost surely, so $a_1>t_0$ almost surely. In both strict branches the detection date avoids the Lebesgue-null set $Z$ and the point $t_0$ almost surely, so the realized responses under $\hat K$ and $K^\ast$ coincide almost surely, and the payoffs are equal.

\noindent \underline{\emph{Part (ii).}} Fix $x<t_0$. Since $X_2\ge t_0>x$ almost surely, only the branch $x<X_2$ occurs: firm 1 concedes at $\tau_1=x$, and firm 1's own responses are never triggered. The rival stops at $\tau_2=X_2\ge t_0>x$ or at $\tau_2=\chi_{\theta_2}(a_2)\ge a_2>x$. The detection date $a_2=x+E_1$ has an absolutely continuous distribution, so it belongs to the null set $Z$ with probability zero. A modification on $Z$ therefore has no effect on the expected payoff. At any remaining date where the kernels differ, $N_{G^\ast}(a_2)=0$ and the response under $\hat K$ is $\chi_1(a_2)=\chi(a_2)\ge\chi_\theta(a_2)$ for every $\theta$. Holding all other randomness fixed, the rival therefore never stops earlier under $\hat K$ than under $K^\ast$ on this event. By \cref{lem:ex-down}, a later rival stop lowers the payoff of a firm that has already stopped, so taking expectations gives $U(x;G^\ast,\hat K)\le U(x;G^\ast,K^\ast)$. By (i), $U(\cdot;G^\ast,\hat K)$ and $U(\cdot;G^\ast,K^\ast)$ agree on $[t_0,+\infty]$, a set that contains $\supp G^\ast$ and on which the latter function attains its global maximum (\cref{lem:ex-fp}(i), with $\supp G^\ast\neq\emptyset$); by (ii) the modification does not raise the value anywhere below $t_0$. Hence the maximum is unchanged and is still attained on $\supp G^\ast$.
\end{proof}

\subsection{Dynamic consistency}\label{app:ex-dyn}

The next lemma asks whether an initially optimal concession date remains optimal after the firm reaches date $t$ without a signal. It does. The reason is that the initial payoff equals a term already determined before $t$ plus a positive multiple of the continuation payoff at $t$. Thus the two problems rank every future concession date in the same way. Fix a rival reduced strategy $(G,K)$, let firm 1 respond per the same kernel $K$, and for $\sigma\in[t,+\infty]$ let $U(\sigma;G,K)$ be as above. Let $N_t$ denote the event that firm 1 has received no signal by $t$ (under any concession plan $\sigma\ge t$, this event does not depend on the plan), and let $\widetilde U_t(\sigma)$ denote firm 1's expected continuation payoff, discounted to $t$, at the no-signal history at $t$, under the plan ``concede at $\sigma$ absent a signal, respond per $K$.''

\begin{lemma}[Conditional optimality at no-signal histories]\label{lem:ex-onestep}
For every $t\ge0$: $m(t):=\Pr(N_t)>0$, and there is a constant $A_t$, independent of the plan, such that
\[
U(\sigma;G,K)=A_t+\alpha(t)\,m(t)\,\widetilde U_t(\sigma)
\qquad\text{for every }\sigma\in[t,+\infty].
\]
Consequently $\argmax_{\sigma\in[t,+\infty]}\widetilde U_t(\sigma)=\argmax_{\sigma\in[t,+\infty]}U(\sigma;G,K)$, and the same identity extends to plans that mix over $\sigma\ge t$, by linearity.
\end{lemma}

\begin{proof}
Since firm 1 has not stopped before $t$ under any plan $\sigma\ge t$, firm 2 is unsignalled and stops, if at all, at its concession date $X_2\sim G$; the first signal to firm 1 arrives at $X_2+E_2$. Thus $N_t=\{X_2>t\}\cup\{X_2\le t,\ X_2+E_2>t\}$, with
\[
m(t)=G\big((t,+\infty]\big)+\int_{[0,t]}e^{-\omega(t-y)}\,G(\de y)\ >\ 0 ,
\]
the two terms not both null. Decompose $U(\sigma;G,K)=\Ex[U\,\mathbf 1_{N_t^c}]+\Ex[U\,\mathbf 1_{N_t}]$. On $N_t^c$ the first signal arrives at some $a\le t$ and firm 1 stops at $\chi_{\theta_1}(a)$ while firm 2 stopped at $X_2$: the realized pair, hence the payoff, does not depend on $\sigma$. On $N_t$, firm 1 is active at $t$, so $\tau_1\ge t$ and the frontier equals $s$ for all $s\le t$; splitting the payoff integral at $t$,
\begin{align*}
    U(\tau_1,\tau_2)\,\mathbf 1_{N_t}
&=\underbrace{\mathbf 1_{N_t}\int_0^{t}\alpha(s)\,\pi\big(s,\ s\wedge X_2\big)\de s}_{\text{independent of }\sigma}
 \\ 
&\quad + \mathbf 1_{N_t}\,\alpha(t)\times\big(\text{continuation flows discounted to }t\big).
\end{align*}
Collecting the $\sigma$-independent pieces into $A_t$ and conditioning the last term on $N_t$ gives the display: the conditional law of $(X_2,E_1,E_2,\theta_1,\theta_2)$ given $N_t$ is the distribution of the primitives at the no-signal history at $t$, so the conditional expectation of the continuation flows is $\widetilde U_t(\sigma)$ by definition.
\end{proof}

The final lemma treats no-signal histories after the latest date at which the rival might concede. Such histories can occur only after the firm has itself deviated from its plan.

\begin{lemma}[After the latest possible concession date]\label{lem:ex-memoryless}
Let $G\in\mathcal G$ with $T_G:=\sup\supp G<\infty$, and fix a history at which firm 1 is active at some $t>T_G$ against a rival playing $(G,K)$, and at which firm 1 has either received no signal, or received its first signal at a date $a>T_G$. The conclusions below also hold at $t=T_G$ whenever $G(\{T_G\})=0$. Under this no-atom condition, they also hold when the first signal arrives  at $T_G$, because the rival stopped strictly before $T_G$ with probability one. Then:
\begin{itemize}
\item[(i)] the rival has stopped, and the Bayes posterior over its (frozen) stopping date is
\[
\hat\mu(\de y)\ \propto\ e^{\omega y}\,G(\de y)\quad\text{on }[0,T_G],
\]
at every such history, independently of $t$ and of the signal date, if any;
\item[(ii)] in particular, signal arrivals after $T_G$ carry no information: $\nu^G_a=\hat\mu$ for every $a>T_G$; the same equality holds at $a=T_G$ under the endpoint condition $G(\{T_G\})=0$;
\item[(iii)] the continuation problem is the deterministic stopping problem with objective
$
w(\cdot\,;t,\hat\mu).
$
This objective is well defined because $\hat\mu$ is supported on $[0,T_G]\subseteq[0,t]$. Its maximum is attained on $\mathcal T(t)\subseteq[t,\,t\vee\tau^R]$. Let $y^\ast(t)$ be the earliest maximizing date. This rule is time-consistent: if the firm plans at $t$ to stop at $y^\ast(t)$ and reaches an intermediate date $t'\le y^\ast(t)$, it still chooses the same stopping date. Ignoring later signals and following this rule is optimal at every such history. When $t\ge\tau^R$, $y^\ast(t)=t$, so the firm stops immediately.
\end{itemize}
\end{lemma}

\begin{proof}
\noindent \underline{\emph{Parts (i)--(ii).}} Let $X_2\sim G$ denote firm 2's concession date. Since $T_G<\infty$, we have $X_2\le T_G<t$ almost surely, so the rival has stopped. By \cref{lem:ex-belief}(i), the no-signal posterior on $[0,t]$ is proportional to $e^{-\omega(t-y)}G(\de y)\propto e^{\omega y}G(\de y)$, and by \cref{lem:ex-belief}(ii) the posterior upon a first signal at $a>T_G$ is $\nu^G_a\propto e^{\omega y}\mathbf 1\{y\le a\}G(\de y)=e^{\omega y}G(\de y)$: both equal $\hat\mu$ after normalization, for every $t,a>T_G$. (As recorded in \cref{lem:ex-belief}(iii), firm 1's own past behavior does not enter this inference.)

\noindent \underline{\emph{Part (iii).}} Since the rival is frozen at $c\sim\hat\mu$ with $c\le T_G<t$ and no future observation changes $\hat\mu$, the continuation payoff of any (possibly signal-contingent) plan is that of some distribution over stopping dates $y\ge t$ evaluated through $w(\cdot\,;t,\hat\mu)$; it therefore suffices to maximize the latter. By \cref{lem:ex-window}(i),(iii), the maximum is attained on $\mathcal T(t)=[t,\chi(t)]\subseteq[t,\,t\vee\tau^R]$, and for $t\ge\tau^R$ we have $\mathcal T(t)=\{t\}$, so $y^\ast(t)=t$. Time-consistency is the splitting identity \eqref{eq:ex-split}: $w(\cdot\,;t,\hat\mu)$ and $w(\cdot\,;t',\hat\mu)$ share their maximizers on $[t',+\infty]$, so if $y^\ast(t)\ge t'$ then $y^\ast(t')=y^\ast(t)$. Optimality at every such history follows since the prescription attains the maximum there.

\noindent \underline{\emph{Endpoint.}} If $t=T_G$ and $G(\{T_G\})=0$, the rival stopped strictly before $t$ almost surely. The no-signal belief on $[0,T_G]$ is again $\hat\mu(\de y)\propto e^{\omega y}G(\de y)$, and the same formula holds after any signal at $a\ge T_G$. When $a=T_G$, every possible rival stopping date already lies below $a$, so the signal formula removes no possible date. The continuation objective is therefore $w(\cdot\,;t,\hat\mu)$, and parts (i)--(iii) follow by the arguments just given.
\end{proof}

\subsection{Proof of \cref{prop:existence}}\label{app:ex-proof}

\begin{proof}[Proof of \cref{prop:existence}]
Fix $(G^\ast,K^\ast)$ from \cref{lem:ex-fp}, let $\hat K$, $t_0$, $Z$ be as in \cref{app:ex-swap}, and write $T^\ast:=\sup\supp G^\ast\in[0,+\infty]$. Define the following symmetric profile for both firms:
\begin{itemize}
\item[(a)] \emph{No-signal plan.} Draw a concession date $X_i$ from $G^\ast$ (\cref{lem:rep}). If the firm reaches date $t$ without a signal and $X_i$ is still in the future, it keeps that candidate. Conditional on reaching $t$ without a signal while following this plan, its candidate has distribution $G^\ast$ conditional on $[t,+\infty]$. If the firm deviated by continuing past $X_i$, it uses a fresh draw from this same conditional distribution whenever $G^\ast([t,+\infty])>0$. If that probability is zero, all possible rival concession dates have passed, and the firm stops at $y^\ast(t)$ from \cref{lem:ex-memoryless}(iii). Keeping a surviving candidate, or drawing once and then keeping the new candidate, gives condition \textup{(DC)}.
\item[(b)] \emph{Response plan.} Upon a first signal at $a$, stop immediately if $a\ge\tau^R$. If $0<a<\tau^R$, draw $\theta$ from $\hat K(a,\cdot)$ and stop at $\chi_\theta(a)$. At the impossible history $a=0$, stop at $\chi(0)$. Let $H_a$ denote the resulting distribution of response dates and let $\mu_a\in\Delta([0,a])$ be the belief held after the signal. The firm keeps its realized response date while that date remains in the future. If it deviates by continuing past that date to some $s$, it uses a fresh draw from $H_a$ conditional on $[s,+\infty]$ whenever this event has positive probability. If no candidate in $H_a$ lies at or after $s$, it stops at
$
\min\argmax_{y\in[s,+\infty]}w(y;s,\mu_a).$ \cref{lem:ex-window}(iii) guarantees that this set of optimal dates is nonempty and compact. Equation \eqref{eq:ex-split} shows that a surviving candidate remains optimal and that the earliest optimal fallback remains unchanged as time passes. Thus the response plan also satisfies \textup{(DC)}.
\end{itemize}
All these rules vary measurably with the observed history. The conditional distributions are Borel, and the measurable maximum theorem used in \cref{lem:ex-meas} implies that the earliest optimal dates above can also be chosen as Borel functions. \cref{lem:rep} then implements each required distribution with a uniform draw.

\emph{Beliefs.}
At every no-signal history at date $t$, use the posterior in \cref{lem:ex-belief}(i), treating $X_j\le t$ as a past stop and $X_j>t$ as an active rival. By \cref{lem:ex-belief}(iii), the same posterior applies after histories created by the firm's own deviations.

At every detection date $a$ with $N_{G^\ast}(a)>0$, use the Bayes posterior $\nu^{G^\ast}_a$, including dates in the null set $Z$. The fixed-point condition already makes $K^\ast$ optimal outside $Z$, while the measurable selection $\theta^\ast(a)$ makes $\hat K$ optimal on $Z$. If $N_{G^\ast}(a)=0$, a signal at $a$ is impossible under $G^\ast$, so Bayes' rule is undefined there. We then assign $\delta_a$, under which the prescribed response $\chi(a)$ is optimal. At the impossible history $a=0$, we likewise assign $\delta_0$ and prescribe $\chi(0)$.

We now verify that no firm can gain after any history. As usual, it is enough to consider deviations that use no new random draw. By \cref{lem:ex-window}(iii), we need consider only response dates in $\mathcal T(\cdot)$, because every date outside that interval is strictly worse at the history where it is chosen.

\emph{Signal histories.} At $a\ge\tau^R$: by \cref{lem:ex-window}(iii), $\mathcal T(a)=\{a\}$ and $w(\cdot;a,\nu)$ is strictly decreasing beyond $a$ for \emph{every} belief $\nu\in\Delta([0,a])$, so immediate stopping is optimal regardless of the belief held. At $a<\tau^R$ with $N_{G^\ast}(a)>0$ and $a\notin Z$: the prescribed response mixes over $\argmax_y w(y;a,\nu^{G^\ast}_a)$ by \cref{lem:ex-fp}(ii) and the definition of $Z$, hence is optimal under the (Bayes) belief. At $a<\tau^R$ with $N_{G^\ast}(a)>0$ and $a\in Z$, the prescribed response $\chi_{\theta^\ast(a)}(a)$ is optimal under the same Bayes belief by the definition of $\theta^\ast(a)$. If $N_{G^\ast}(a)=0$, the assigned belief is $\delta_a$ and the prescribed response $\chi(a)$ is uniquely optimal (\cref{lem:ex-window}(iv)). Finally, at histories strictly between a signal at $a$ and the prescribed stop $y$: no further signals arrive (the rival has already stopped) and survival is uninformative (\cref{lem:ex-belief}(iii)), so the belief is unchanged. By \eqref{eq:ex-split}, moving from $a$ to an intermediate date $s$ only adds a constant and multiplies later payoffs by a positive number. It does not change which remaining stopping date is best, so continuing to $y$ stays optimal. At histories at some $s$ strictly past the prescribed stop---reached only by the firm's own deviations---the belief is still unchanged: no signal has intervened (the rival is stopped), survival is uninformative, and the firm's own play is uninformative about its rival (\cref{lem:ex-belief}(iii)). If $H_a$ assigns positive probability to dates at or after $s$, draw a new response date from that part of $H_a$. Every date selected this way was optimal at $a$ and remains optimal at $s$ by \eqref{eq:ex-split}. If $H_a$ assigns no probability to such dates, choose the earliest date that maximizes $w(\cdot\,;s,\mu_a)$. This choice is optimal, and the same equation shows that both rules remain optimal as time passes.

\textbf{No-signal histories with $G^\ast([t,+\infty])>0$.}
A continuation deviation consists of a (possibly mixed) concession
date $\sigma\ge t$ on the no-signal history, together with responses
at any future signal histories that may arise.

We treat responses first. Whenever a finite first signal arrives, its
date satisfies $a_1=X_2+E_2>t_0$ and its law is absolutely
continuous; if $X_2=+\infty$, no signal arrives
(\cref{lem:ex-belief}(ii)). Every such finite signal date satisfies
$N_{G^\ast}(a_1)>0$, so the posterior at $a_1$ is
$\nu^{G^\ast}_{a_1}$ by \cref{lem:ex-belief}(iii), irrespective of
the deviator's own behavior. The prescribed response is
optimal at every such date: $K^\ast$ is optimal outside $Z$, and
$\theta^\ast(a_1)$ is optimal on $Z$. Response deviations therefore
cannot raise the continuation value.

It remains to compare the payoffs of the concession dates. By \cref{lem:ex-onestep},
applied with the rival's strategy $(G^\ast,\hat{K})$ and own
responses $\hat{K}$, we have
\[
  \widetilde{U}_t(\sigma)
  =\frac{U(\sigma;G^\ast,\hat{K})-A_t}{\alpha(t)\,m(t)}
  \qquad\text{for all }\sigma\in[t,+\infty],
\]
so maximizing the continuation value over $\sigma\ge t$ is equivalent
to maximizing $U(\,\cdot\,;G^\ast,\hat{K})$ over $[t,+\infty]$. The
prescribed rule continues according to
$G^\ast(\,\cdot\mid[t,+\infty])$, whose support is contained in
\[
  \supp G^\ast\cap[t,+\infty]
  \;\subseteq\;
  \argmax_{x\in\overline{\mathbb{R}}_+}U(x;G^\ast,\hat{K})
  \cap[t,+\infty],
\]
by \cref{lem:ex-swap}. Every date selected by the prescription
therefore attains the global maximum of
$U(\,\cdot\,;G^\ast,\hat{K})$, and hence also the maximum restricted
to dates no earlier than~$t$, so the specified dates are optinmal

\emph{No-signal histories with $G^\ast([t,+\infty])=0$.} Here $T^\ast<\infty$ and either $t>T^\ast$, or $t=T^\ast$ with $G^\ast(\{T^\ast\})=0$. By \cref{lem:ex-memoryless}, the rival has already stopped, the belief is $\hat\mu$ both before and after any later signal, and stopping at $y^\ast(t)$ is optimal. At every later feasible detection date for which Bayes' rule is defined, the posterior equals $\hat\mu$. The response plan is optimal under this belief: the fixed-point kernel supplies an optimal response outside $Z$, and the measurable selection supplies one on $Z$. The time-consistency result in \cref{lem:ex-memoryless}(iii) then completes the strategy. No deviation improves on it.

Histories at which the firm has already stopped involve no further decisions. Thus this profile, together with the specified beliefs, is a perfect Bayesian equilibrium of the game with news rate $\omega$, and $\PBE(\omega)\neq\emptyset$.
\end{proof}

\section{Extensions}\label{app:extensions}

This appendix formalizes the two risk extensions discussed in the conclusion.
We retain perfect transparency and common knowledge of rationality. The results
bound every subgame-perfect equilibrium, including mixed equilibria. The statements and proofs are written for an arbitrary finite set of firms and specialize to the two-firm model of the main text.  

\subsection{Additive Risk }\label{app:additive risk}

Keep the baseline technology law, but let risk depend on the number of firms
that remain active. For a stopping profile $\bm\tau$, let
\[
k_s(\bm\tau):=\#\{j:\tau_j>s\}.
\]
When $k_s(\bm\tau)\geq1$, every active firm is at the frontier $s$, and the
date-$s$ hazard is $\lambda(s,k_s(\bm\tau))$; when $k_s(\bm\tau)=0$, the
hazard is zero. Whether a firm that stops at $s$ enters the count is
immaterial because a single date has measure zero. Payoffs are otherwise as in
the baseline model, using this hazard process.

\begin{assumption}[additive risk regularity]\label{ass:crowded-sc}
For each $k\in\{1,\ldots,n\}$, the hazard $\lambda(\cdot,k)$ is continuous
and nondecreasing in technology. It is also nondecreasing in the number of
active firms, so in particular
\[
\lambda(t,k)\geq\lambda(t,1)
\qquad\text{for every }t\geq0\text{ and }k\geq1.
\]
Flow payoffs satisfy \cref{ass:payoffs}(ii)--(iii). We will also impose the analog of \cref{ass:payoffs} (iv): for every firm $i$ and
fixed rival technology vector $\bm c$, the map
\[
x\longmapsto D_i\log\pi(x,\bm c)-\lambda(x,1)
\]
is continuous and strictly decreasing on
$[\max_j c_j,+\infty)$. Finally, the solo-hazard analogue of
\cref{ass:payoffs} (v) holds:
\[
\int_0^\infty e^{-rs}\sup_{0\leq x\leq s}
\left\{
e^{-\int_0^x\lambda(u,1)\de u}\pi(x,0,\ldots,0)
\right\}\de s<+\infty.
\]
\end{assumption}

Say we are in the \emph{additive risk environment} when the above assumptions hold and there is perfect transparency and common knowledge of rationality.

\begin{proposition}[The solo hazard bounds a additive risk]
\label{prop:additive risk}
Suppose that we are in the additive risk environment. Then every subgame-perfect equilibrium satisfies:
\[
\max_{i\in\mathcal N}\tau_i
\leq
\inf\left\{
t\geq0:
D_i\log\pi(t,\ldots,t)\leq\lambda(t,1)
\right\}
\qquad\text{almost surely}. 
\]
\end{proposition}

\begin{proof}
We first solve the last-active firm's problem, then induct on the number of
active firms. Let $\tau^C$ denote the infimum in the statement. If
$\tau^C=+\infty$, the claim is immediate. Suppose $\tau^C<+\infty$.
Continuity implies
\[
D_i\log\pi(\tau^C,\ldots,\tau^C)\leq\lambda(\tau^C,1).
\]

Fix a public history at date $t$, and suppose that firm $i$ scales alone while
its rivals remain frozen at $\bm c\leq(t,\ldots,t)$. Conditional on survival
to $t$, its payoff from stopping at $\sigma\geq t$ is
\[
R^C_{t,\bm c}(\sigma)
:=
\int_t^\sigma
e^{-\int_t^s[r+\lambda(u,1)]\de u}\pi(s,\bm c)\de s
+
e^{-\int_t^\sigma[r+\lambda(u,1)]\de u}
\frac{\pi(\sigma,\bm c)}r,
\]
with the terminal term absent when $\sigma=+\infty$. At every finite
$\sigma>t$,
\[
\frac{\partial}{\partial\sigma}R^C_{t,\bm c}(\sigma)
=
e^{-\int_t^\sigma[r+\lambda(u,1)]\de u}
\frac{\pi(\sigma,\bm c)}r
\left[D_i\log\pi(\sigma,\bm c)-\lambda(\sigma,1)\right].
\]
At $t=\tau^C$, log-supermodularity and
\cref{ass:crowded-sc} give, for every $x\geq\tau^C$,
\begin{align*}
D_i\log\pi(x,\bm c)-\lambda(x,1)
&\leq
D_i\log\pi(x,\tau^C,\ldots,\tau^C)-\lambda(x,1)\\
&\leq
D_i\log\pi(\tau^C,\ldots,\tau^C)-\lambda(\tau^C,1)
\leq0,
\end{align*}
where the second inequality is strict when $x>\tau^C$. Hence
\[
R^C_{\tau^C,\bm c}(\sigma)
<R^C_{\tau^C,\bm c}(\tau^C)
=\frac{\pi(\tau^C,\bm c)}r
\qquad\text{for every }\sigma>\tau^C,
\]
including $\sigma=+\infty$, because the integrability condition makes the
terminal term vanish and the objective continuous there. Thus a last active
firm uniquely prefers to stop at $\tau^C$.

The additive risk payoff is continuous in the stopping profile. Indeed,
along any convergent sequence of stopping profiles, the active-firm count
converges at every date other than the finitely many limiting stopping dates,
so cumulative hazards converge on every bounded interval. The solo-hazard
integrability condition supplies the same dominating envelope as in
\cref{lem:cont}. Dominated convergence therefore applies. In particular,
stopping immediately after an observed stop is payoff-equivalent to stopping
simultaneously.

We now claim that no equilibrium continuation from a public history at
$\tau^C$ carries the frontier above $\tau^C$. We prove the claim by induction
on the number of active firms. The preceding calculation proves the case of
one active firm. Suppose the claim holds with fewer than $m$ active firms, and
consider a history with $m$ active firms. If an equilibrium continuation
carries the frontier above $\tau^C$ with positive probability, then, because
$m$ is finite, some active firm $i$ satisfies
\[
\Pr(\tau_i>\tau^C)>0.
\]
Let $\bm c$ be its rivals' technologies at this history. If firm $i$ stops at
$\tau^C$, the successor subgame has $m-1$ active firms. Subgame perfection and
the induction hypothesis make every remaining active firm stop immediately,
so the deviation gives firm $i$ the payoff
$R^C_{\tau^C,\bm c}(\tau^C)$.

For comparison, fix any realized equilibrium continuation and write $\sigma$
for firm $i$'s stopping date. Rival technologies are weakly above $\bm c$,
which weakly lowers firm $i$'s flow payoff. While firm $i$ remains active, the
actual hazard is at least $\lambda(s,1)$; after it stops, any remaining hazard
can only lower its payoff further. Hence, path by path, its continuation payoff
is no greater than $R^C_{\tau^C,\bm c}(\sigma)$, which is strictly below
$R^C_{\tau^C,\bm c}(\tau^C)$ whenever $\sigma>\tau^C$. The comparison holds
before averaging over equilibrium randomization. Since firm $i$ delays with
positive probability, stopping at $\tau^C$ strictly raises its expected payoff,
contradicting sequential rationality. This proves the induction claim.

If an equilibrium frontier exceeded $\tau^C$ with positive probability, its
public continuation at the finite date $\tau^C$, conditional on survival,
would contradict the claim. Therefore the frontier is at most $\tau^C$ almost
surely.
\end{proof}

\subsection{Accumulative risk}\label{app:persistent-risk}

Keep the baseline hazard $\lambda$, but suppose that risk persists at the peak
technology. For a stopping profile $\bm\tau$, let
\[
T(\bm\tau):=\max_{j\in\mathcal N}\tau_j.
\]
The date-$s$ hazard is now $\lambda(s\wedge T(\bm\tau))$. Thus firm $i$'s
payoff is
\[
\int_0^\infty
\exp\left\{-rs-\int_0^s\lambda(u\wedge T(\bm\tau))\de u\right\}
\pi\Bigl(s\wedge\tau_i,(s\wedge\tau_j)_{j\neq i}\Bigr)\de s.
\]
If $T(\bm\tau)<+\infty$, the hazard remains equal to
$\lambda(T(\bm\tau))$ after the last stop. Its survival factor is no larger
than the baseline survival factor, so \cref{ass:payoffs}(v) keeps payoffs
finite.

The marginal condition below need not be single-crossing under
\cref{ass:payoffs}. We therefore impose the analogue of \ref{ass:payoffs} (iv):

\begin{assumption}[Persistent-risk single-crossing]
\label{ass:persistent-sc}
In addition to \cref{ass:payoffs}, for every firm $i$ and fixed rival
technology vector $\bm c$, the map
\[
x\longmapsto
D_i\log\pi(x,\bm c)
-\frac{\lambda'(x)}{r+\lambda(x)}
\]
is strictly decreasing on $[\max_j c_j,+\infty)$.
\end{assumption}

Say we are in the \emph{accumulative risk environment} when the above assumptions hold and there is perfect transparency and common knowledge of rationality.

\begin{proposition}[Marginal persistent risk bounds the frontier]
\label{prop:persistent-upper}
Suppose we are in the accumulative risk environment. Then every subgame-perfect equilibrium satisfies
\[
\max_{i\in\mathcal N}\tau_i
\leq
\inf\left\{
t\geq0:
D_i\log\pi(t,\ldots,t)
\leq
\frac{\lambda'(t)}{r+\lambda(t)}
\right\}
\qquad\text{almost surely}.
\]
\end{proposition}

\begin{proof}
Again, we solve the last-active firm's problem, then induct. Let $\tau^A$ denote the infimum in the statement. If $\tau^A=+\infty$, the claim is immediate. Suppose $\tau^A<+\infty$. Continuity implies that the defining
inequality holds at $\tau^A$.

Fix a history at date $t$, and suppose that firm $i$ scales alone while its
rivals remain frozen at $\bm c\leq(t,\ldots,t)$. Conditional on survival to
$t$, its payoff from stopping at $\sigma\geq t$ is
\[
R^A_{t,\bm c}(\sigma)
:=
\int_t^\sigma
e^{-\int_t^s[r+\lambda(z)]\de z}\pi(s,\bm c)\de s
+
e^{-\int_t^\sigma[r+\lambda(z)]\de z}
\frac{\pi(\sigma,\bm c)}{r+\lambda(\sigma)},
\]
with the terminal term absent when $\sigma=+\infty$. At every finite
$\sigma>t$,
\[
\frac{\partial}{\partial\sigma}R^A_{t,\bm c}(\sigma)
=
e^{-\int_t^\sigma[r+\lambda(z)]\de z}
\frac{\pi(\sigma,\bm c)}{r+\lambda(\sigma)}
\left[
D_i\log\pi(\sigma,\bm c)
-\frac{\lambda'(\sigma)}{r+\lambda(\sigma)}
\right].
\]
At $t=\tau^A$, log-supermodularity and
\cref{ass:persistent-sc} give, for every $x\geq\tau^A$,
\begin{align*}
D_i\log\pi(x,\bm c)-\frac{\lambda'(x)}{r+\lambda(x)}
&\leq
D_i\log\pi(x,\tau^A,\ldots,\tau^A)
-\frac{\lambda'(x)}{r+\lambda(x)}\\
&\leq
D_i\log\pi(\tau^A,\ldots,\tau^A)
-\frac{\lambda'(\tau^A)}{r+\lambda(\tau^A)}
\leq0,
\end{align*}
where the second inequality is strict when $x>\tau^A$. Therefore
\[
R^A_{\tau^A,\bm c}(\sigma)
<R^A_{\tau^A,\bm c}(\tau^A)
=\frac{\pi(\tau^A,\bm c)}{r+\lambda(\tau^A)}
\qquad\text{for every }\sigma>\tau^A,
\]
including $\sigma=+\infty$ by \cref{ass:payoffs}(v). Thus a last active firm
uniquely prefers to stop at $\tau^A$.

The accumulative-risk payoff is continuous in the stopping profile. Along any
convergent sequence, the peak path $s\mapsto s\wedge T(\bm\tau)$ and its
cumulative hazard converge on every bounded interval. The persistent survival
factor is bounded above by the baseline one, so \cref{ass:payoffs}(v) supplies
an integrable envelope. Dominated convergence applies, and immediate successor
stopping is payoff-equivalent to simultaneous stopping.

We claim that no equilibrium continuation from a public history at $\tau^A$
carries the frontier above $\tau^A$. Induct on the number of active firms. The
last-active calculation proves the base case. Suppose the claim holds with
fewer than $m$ active firms, and consider a history with $m$ active firms. If
the equilibrium continuation carries the frontier beyond $\tau^A$ with
positive probability, some active firm $i$ satisfies
$\Pr(\tau_i>\tau^A)>0$. Let $\bm c$ be its rivals' current technologies. If
firm $i$ stops at $\tau^A$, the successor subgame has $m-1$ active firms; by
subgame perfection and the induction hypothesis, every remaining active firm
stops immediately. The deviation therefore gives firm $i$ the payoff
$R^A_{\tau^A,\bm c}(\tau^A)$.

Fix any realized equilibrium continuation and write $\sigma$ for firm $i$'s
stopping date. Rival technologies are weakly above $\bm c$, which weakly lowers
firm $i$'s flow payoff. Moreover, the realized peak is always weakly above the
path $s\mapsto s\wedge\sigma$ generated when firm $i$ scales alone, so the
actual persistent hazard is weakly higher than the hazard in
$R^A_{\tau^A,\bm c}(\sigma)$. Firm $i$'s realized continuation payoff is thus
no greater than this frozen-rival benchmark, which is strictly below
$R^A_{\tau^A,\bm c}(\tau^A)$ whenever $\sigma>\tau^A$, including
$\sigma=+\infty$. This comparison is pathwise and therefore covers mixed
strategies. Because firm $i$ delays with positive probability, immediate
stopping is a profitable deviation, a contradiction. This proves the
induction claim.

If an equilibrium frontier exceeded $\tau^A$ with positive probability, its
public continuation at the finite date $\tau^A$, conditional on survival,
would contradict the claim. Hence the frontier is at most $\tau^A$ almost
surely.
\end{proof}

\section{Further Results}\label{app:further}

\subsection{Tight monitoring thresholds under stationarity} We now formalize the claim in the main text that the threshold $\bar \omega$ is exact in the stationary environment we saw in 
\cref{ex:exponential,sec:trust}. To recap, flow payoffs and the hazard are
\[
\pi(t_i,t_j)=e^{\alpha t_i-\beta t_j},
\qquad
\lambda(t)\equiv\lambda,
\]
where $\alpha,\beta>0$, $\lambda>\alpha$, and $r>\alpha-\beta$. We maintain
common knowledge of rationality. The inequality $\lambda>\alpha$ places the
game in the coordination phase: once a firm knows that its rival has stopped,
it strictly prefers to stop immediately.

We call this environment \emph{stationary} because a common translation of
both technology paths multiplies every continuation payoff by the current
diagonal flow $\pi(t,t)=e^{(\alpha-\beta)t}$. Values per unit of current flow
therefore do not depend on the date or technology level. At a symmetric quiet
history, the normalized value of racing forever is
\[
V_R=\frac{1}{r+\lambda-\alpha+\beta}.
\]
The normalized value of stopping first is
\[
V_S(\omega)
=V_S(0)
+\frac{\omega\big[1/r-V_S(0)\big]}
{\omega+r+\lambda+\beta},
\qquad
V_S(0)=\frac{1}{r+\lambda+\beta}.
\]
Thus
\[
\Delta(t,\omega)
=\pi(t,t)\big[V_S(\omega)-V_R\big],
\]
and the bracketed difference is positive exactly when
\[
\omega>\bar\omega
:=\frac{r\alpha}{\lambda-\alpha+\beta}.
\]

\begin{proposition}[Exact racing-to-ruin threshold in the stationary environment]
\label{prop:exact-threshold}
In the stationary environment:
\begin{itemize}
\item[(i)] If $\omega\le\bar\omega$, there is an equilibrium in which both
firms race forever, so ruin occurs with probability one.
\item[(ii)] If $\omega>\bar\omega$, every
$\bm\tau\in\PBE(\omega)$ satisfies
$\max_i\tau_i<+\infty$ almost surely.
\end{itemize}
\end{proposition}

\begin{proof}
We construct a ruin equilibrium at and below the threshold, then rule out ruin
above it.

\emph{Part (i).} Consider the assessment used in
\cref{prop:never-stop}: each firm races forever at every quiet history and
stops immediately after detecting its rival's stop. Against a rival that
races until detection, any finite stop at $\sigma$ yields, by
\cref{lem:rec}, the payoff from racing forever plus a strictly positive multiple of
$\Delta(\sigma,\omega)$. Stationarity and $\omega\le\bar\omega$ give
$\Delta(\sigma,\omega)\le0$ at every date, so no finite stop is profitable.
Immediate stopping after a signal is uniquely optimal because
$\lambda>\alpha$. The assessment
is therefore a perfect Bayesian equilibrium,
including at $\omega=\bar\omega$, where the stopping deviation is
payoff-equivalent to racing forever.

\emph{Part (ii).} Fix an equilibrium with outcome
$(\tau_1,\tau_2)$ and let
$
E:=\{\tau_1=\tau_2=+\infty\}.
$

\emph{Step 1: reduce an infinite frontier to joint racing.}
Because $\lambda>\alpha$, a sole survivor's payoff against any frozen rival is
strictly decreasing in its stopping date. Immediate stopping is therefore
uniquely optimal at every signal history under every admissible belief. If one
firm stops in finite time, its rival detects the stop in finite time almost
surely and then stops immediately. Hence
$
\{\max_i\tau_i=+\infty\}=E
$
up to a null event.

\emph{Step 2: show that joint racing has zero probability.}
Suppose instead that $\Pr(E)>0$. Fix firm $i$ and its quiet history $H_t$ at
date $t$, meaning that it is active and has received no signal. Since
$E\subseteq H_t$, we have $\Pr(H_t)\ge\Pr(E)>0$. Conditional on $H_t$, its
rival is in one of three states: it has stopped but remains undetected; it is
active and will stop at a finite date; or it is active and never stops. The
conditional probability of the second state is at most
\[
\frac{\Pr(t<\tau_j<+\infty)}{\Pr(E)}
\longrightarrow0.
\]
The joint conditional probability that both firms' residual plans never stop
is at least $\Pr(E\mid H_t)\ge\Pr(E)$.

On the last two states, the rival's technology remains weakly above $t$.
Thus, at every future date $s\ge t$, the discounted, survival-weighted flow
from firm $i$'s never-stopping plan is bounded by
\[
\pi(t,t)e^{-(r+\lambda-\alpha)(s-t)}.
\]
Its continuation payoff is consequently at most
$\pi(t,t)/(r+\lambda-\alpha)$.

Now modify firm $i$'s plan at $H_t$ only on the device realizations for which
it would never stop: stop at $t$ on those realizations and leave every other
plan unchanged. If the
rival has already stopped, the change is a weak gain. If both residual plans
never stop, it replaces $V_R\pi(t,t)$ with
$V_S(\omega)\pi(t,t)$, a gain of
$[V_S(\omega)-V_R]\pi(t,t)>0$. This gain receives conditional weight at least
$\Pr(E)$. The only possible loss comes when the rival is active but will stop
at a finite date; its conditional probability tends to zero, and its magnitude
is bounded by $\pi(t,t)/(r+\lambda-\alpha)$. The conditional expected gain is
therefore positive for all sufficiently large $t$, contradicting sequential
rationality at $H_t$. Hence $\Pr(E)=0$. Step 1 now gives
$\max_i\tau_i<+\infty$ almost surely.
\end{proof}

\paragraph{Trust for the case $\omega < \beta$} We now consider the other case \cref{sec:trust} excluded there: news that is slow relative to payoff erosion: $\omega<\beta$. The setup is otherwise identical to that of \cref{sec:trust}. Define
\[
\underline\ell_\beta:=\frac{V_R-V_S(0)}{1/r-V_W(\beta)},
\qquad\text{noting}\qquad
\frac1r-V_W(\beta)=\frac{\lambda-\alpha}{r(r+\lambda-\alpha+\beta)}=\frac{(\lambda-\alpha)V_R}{r}:
\]
this is the verification threshold $\underline\ell$ of \cref{sec:trust} evaluated at news rate $\beta$, and it does not depend on $\omega$.

\begin{proposition}[Trust with slow news]\label{prop:slow-news}
In the stationary coordination environment with types as above, suppose $0<\omega<\beta$ and maintain $\underline\ell_\beta<\bar\ell$. Then \cref{prop:LMH} holds with $\underline\ell$ replaced by $\underline\ell_\beta$:
\begin{itemize}[leftmargin=2.4em]
\item[(i)] \emph{Low trust.} If $\ell<\underline\ell_\beta$, then every $\bm\tau\in\PBE(\omega,\ell)$ has $\tau_1=\tau_2=+\infty$ almost surely.
\item[(ii)] \emph{Medium trust.} If $\underline\ell_\beta\le\ell\le\bar\ell$, then $\PBE(\omega,\ell)$ contains both the outcome in which rational firms stop immediately and the never-stopping outcome.
\item[(iii)] \emph{High trust.} If $\ell>\bar\ell$, then every $\bm\tau\in\PBE(\omega,\ell)$ satisfies
$\Pr\big(\tau_1=\tau_2=+\infty\mid\text{both firms rational}\big)\le\big(\bar\ell/\ell\big)^2$, and stopping at zero remains an equilibrium outcome.
\end{itemize}
Moreover, rational firms stopping at a common date is an equilibrium if and only if $\ell\ge\underline\ell_\beta$, and the binding deviation is a marginal delay rather than indefinite verification: for $\ell<\underline\ell_\beta$, the most profitable plan against the bunched profile is to delay by $\sigma^\ast=\log\big(\underline\ell_\beta/\ell\big)/(\beta-\omega)$.
\end{proposition}

\begin{proof}
As in the proof of \cref{prop:LMH}, abbreviate $\delta:=r+\lambda-\alpha+\beta=1/V_R$ and $\rho:=r+\lambda+\beta=1/V_S(0)$, and write $N:=V_R-V_S(0)$; the displayed identity for $1/r-V_W(\beta)$ is a direct computation. That proof used $\omega\ge\beta$ in two places: in part (i), to damp the undetected-stop cell, and in part (ii)(a), to conclude that the deviation payoff $W$ is maximized at an endpoint of $[0,\infty]$. Both steps are replaced by an analysis of marginal delay. \\ 

\noindent \emph{The marginal-delay deviation.} Fix a quiet history of firm $i$ at date $t$ and, for $h>0$, let \textsc{delay}$(h)$ denote the plan: continue until $t+h$, stopping immediately upon a signal, and stop at $t+h$ if none has arrived. Per unit of current flow, we compute the derivative at $h=0$ of its payoff, net of stopping at once, against each rival state in the posterior decomposition of part (i) of the proof of \cref{prop:LMH}. \\ 

\noindent  \emph{Crazy.} \textsc{delay}$(h)$ yields $(1-e^{-\delta h})V_R+e^{-\delta h}V_S(0)$: postponing an unreciprocated stop gains at rate $\delta N>0$. \\ 

\noindent  \emph{Rational, stopped at $c\le t$, undetected.} The delaying firm races beside a frozen rival under the live hazard, and $\alpha<\lambda$ makes stopping at once optimal on this cell: the deviation loses at rate $e^{\beta(t-c)}(\lambda-\alpha)/r=e^{\beta(t-c)}\,\delta\big[1/r-V_W(\beta)\big]$. \\ 

\noindent \emph{Rational, active, with residual quiet-path stopping time $\tau$.} Stopping now ends the risk at $X\wedge\tau$, where $X\sim\mathrm{Exp}(\omega)$ is the detection lag, and is worth $V_S(0)+\big[1/r-V_S(0)\big]\Ex[e^{-\rho(X\wedge\tau)}]$; on $\{\tau>h\}$, delaying replaces this by $(1-e^{-\delta h})V_R+e^{-\delta h}\big\{V_S(0)+\big[1/r-V_S(0)\big]\Ex[e^{-\rho(X\wedge(\tau-h))}]\big\}$. Differentiating at $h=0$ and using $\Ex[e^{-\rho(X\wedge\tau)}]=\tfrac{\omega}{\omega+\rho}+\tfrac{\rho}{\omega+\rho}e^{-(\omega+\rho)\tau}$, the deviation gains at rate $-\delta\,\psi(\tau)$, where
\[
\psi(\tau):=\Big[\frac1r-V_S(0)\Big]\Big[\frac{\omega}{\omega+\rho}+\rho\Big(\frac{1}{\omega+\rho}-\frac1\delta\Big)e^{-(\omega+\rho)\tau}\Big]-N
\quad\text{for }\tau\in(0,\infty],
\]
and $\psi(0):=1/r-V_W(\beta)$ at an atom at zero, where the rival stops at $t$ itself and the previous cell applies with $c=t$. (A residual stop at $t+u$ with $u\in(0,h]$ contributes $\alpha u/r-(\lambda-\alpha)(h-u)/r+o(h)\ge-\delta\,\psi(0)\,h+o(h)$, so such mass is covered by the bound below.) Since $\omega+\rho-\delta=\omega+\alpha>0$, the coefficient of $e^{-(\omega+\rho)\tau}$ is negative: $\psi$ is strictly increasing on $(0,\infty)$, from $\psi(0^+)=-N-\alpha\big[1/r-V_S(0)\big]/\delta<0$ to $\psi(\infty)=V_S(\omega)-V_R$. Hence, using the maintained $\underline\ell_\beta<\bar\ell$,
\[
\sup_{[0,\infty]}\psi=\max\Big\{\frac1r-V_W(\beta),\;V_S(\omega)-V_R\Big\}=\frac{N}{\min\{\underline\ell_\beta,\bar\ell\}}=\frac{N}{\underline\ell_\beta}.
\]
A marginal delay is thus costly only against a rival that has just stopped---the slight-delay exposure that defines $\underline\ell_\beta$---or against one that will never stop spontaneously, forgoing the prize $V_S(\omega)-V_R$ of stopping first. \\ 

\noindent  \emph{Part (i).} Suppose, for contradiction, that some rational firm stops at a quiet history with positive probability. With $F_j$ and $S_j$ as in part (i) of the proof of \cref{prop:LMH}, let $t_0:=\inf\{t\ge0:F_1(t)+F_2(t)>0\}<\infty$, and fix $\varepsilon>0$ with $e^{(\beta-\omega)\varepsilon}<\underline\ell_\beta/\ell$. At a quiet history of firm $i$ at any date $t\le t_0+\varepsilon$, every undetected stop has age $t-c\le\varepsilon$. Weighting the three cells in proportion to $1$, $\ell S_j(t)$, and $\ell e^{-\omega(t-c)}\de F_j(c)$, the net gain rate of \textsc{delay} is proportional to
\begin{align*}
\delta\Big\{N-\ell\Big(S_j(t)\,&\Ex[\psi(\tau)]+\Big[\frac1r-V_W(\beta)\Big]\int_{[t_0,t]}e^{(\beta-\omega)(t-c)}\,\de F_j(c)\Big)\Big\}
\ \\ 
&\ge\ \delta N\Big[1-\frac{\ell}{\underline\ell_\beta}\,e^{(\beta-\omega)\varepsilon}\Big]\ \\
&>\ 0,    
\end{align*}
by $\Ex[\psi(\tau)]\le N/\underline\ell_\beta$, $e^{(\beta-\omega)(t-c)}\le e^{(\beta-\omega)\varepsilon}$, and $S_j(t)+\int_{[t_0,t]}\de F_j\le1$. A small delay therefore strictly improves on stopping at every quiet history at dates $t\le t_0+\varepsilon$, so no rational firm stops there, under any pure or mixed rule---contradicting the definition of $t_0$. Hence no rational firm ever stops at a quiet history; since the first stop in any equilibrium must occur at a quiet history, no stop ever occurs and $\tau_1=\tau_2=+\infty$ almost surely. \\ 

\noindent  \emph{Part (ii).} The never-stopping profile is an equilibrium for $\ell\le\bar\ell$ by part (ii)(b) of the proof of \cref{prop:LMH}, which nowhere compares $\omega$ with $\beta$. For the bunched profile, the deviation payoff $W(\sigma)$ is as computed in part (ii)(a) there, and
\[
e^{\delta\sigma}W'(\sigma)=-p\,\frac{\lambda-\alpha}{r}\,e^{(\beta-\omega)\sigma}+(1-p)\,\delta N .
\]
For $\omega<\beta$ the right side is strictly decreasing in $\sigma$, so $W'$ changes sign at most once, from $+$ to $-$, and $W$ is maximized at $\sigma=0$ if and only if $W'(0)\le0$, i.e.
\[
p\,\frac{\lambda-\alpha}{r}\ \ge\ (1-p)\,\delta N
\quad\Longleftrightarrow\quad
\ell\ \ge\ \frac{N}{(\lambda-\alpha)V_R/r}\ =\ \underline\ell_\beta .
\]
Hence for $\ell\ge\underline\ell_\beta$ no plan improves on stopping at once and the bunched profile is an equilibrium; stationarity extends this to any common date. Conversely, for $\ell<\underline\ell_\beta$ the unique maximizer of $W$ is the interior delay $\sigma^\ast=\log(\underline\ell_\beta/\ell)/(\beta-\omega)>0$, a strictly profitable deviation---so no equilibrium has rational firms stopping at a common date. This proves part (ii) and the final claim. \\ 

\noindent  \emph{Part (iii).} The exclusion of never-stopping for $\ell>\bar\ell$ uses only the immediate-stop gamble, and zero ruin is attained by the bunched profile, an equilibrium since $\ell>\bar\ell>\underline\ell_\beta$; both are as in \cref{prop:LMH}. For the $\big(\bar\ell/\ell\big)^2$ bound, the proof there lets a delaying plan gain at most $1/r$ per unit of the vanishing undetected-stop weight; when $\omega<\beta$ this bound fails, since payoffs on that cell grow like $e^{\beta(t-c)}$ while its weight decays only like $e^{-\omega(t-c)}$. But no bound is needed: because $\alpha<\lambda$, stopping at once attains the cell-maximum $e^{\beta(t-c)}/r$ on the undetected cell, so, relative to stopping at once, a continuation plan gains nothing there and can profit only on the crazy and active cells, where the argument of \cref{prop:LMH} applies verbatim. The conclusion follows unchanged.
\end{proof}

Together with \cref{prop:LMH}, the simultaneous-coordination threshold at every news rate is $\big[V_R-V_S(0)\big]/\big[1/r-V_W(\omega\vee\beta)\big]$: constant in $\omega$ below $\beta$---the vertical segment in \cref{fig:LMH}---and rising thereafter. The two regimes fit together continuously: as $\omega\uparrow\beta$ the binding delay $\sigma^\ast$ lengthens without bound, deforming into the indefinite-verification deviation that binds when $\omega\ge\beta$. Transparency below $\beta$ therefore neither helps nor hurts simultaneous coordination---the binding temptation is instantaneous, and its cost and benefit are both $\omega$-free---while faster news still lowers the sequential threshold $\bar\ell$. If instead $\underline\ell_\beta>\bar\ell$, the remarks following \cref{prop:LMH} apply with $\underline\ell_\beta$ in place of $\underline\ell$.

\end{document}